\documentclass[12pt, a4paper]{article}
\pdfoutput=1
\usepackage{jheppub}
\usepackage{subcaption}
\usepackage[T1]{fontenc}
\usepackage{float}
\usepackage{latexsym}
\usepackage{amsmath,amstext,amssymb,amsfonts,
amscd,bm,array,multirow,amsbsy,mathrsfs,amsthm,calc}
\usepackage{t1enc}
\usepackage{indentfirst}
\usepackage{pb-diagram}
\usepackage{graphicx}
\usepackage{enumerate}
\usepackage{color}
\usepackage[all]{xy}
\usepackage{hyperref}
\hypersetup{colorlinks=false,pdfborderstyle={/S/U/W 0}}
\usepackage[usenames,dvipsnames]{xcolor}
\usepackage{url}
\usepackage{upgreek}

\theoremstyle{plain}
\newtheorem{thm}{Theorem}[section]             

\newtheorem{prop}[thm]{Proposition}

\newtheorem{lemma}[thm]{Lemma}
\newtheorem{cor}[thm]{Corollary}

\theoremstyle{definition}
\newtheorem{definition}[thm]{Definition}

\theoremstyle{remark}
\newtheorem{remark}[thm]{Remark}

\newcommand{\be}{\begin{equation*}}
\newcommand{\ee}{\end{equation*}}
\newcommand{\ben}{\begin{equation}}
\newcommand{\een}{\end{equation}}
\newcommand{\beqa}{\begin{eqnarray*}}
\newcommand{\eeqa}{\end{eqnarray*}}
\newcommand{\beqan}{\begin{eqnarray}}
\newcommand{\eeqan}{\end{eqnarray}}
\newcommand{\nn}{\nonumber}
\def\N{\mathbb{N}}
\def\Z{\mathbb{Z}}
\def\C{\mathbb{C}}
\def\R{\mathbb{R}}
\def\Q{\mathbb{Q}}
\def\H{\mathbb{H}}

\def\Crit{\mathrm{Crit}}
\def\Noncrit{\mathrm{Noncrit}}

\def\pd{\partial}

\def\reg{\mathrm{reg}}
\def\sing{\mathrm{sing}}
\def\crit{\mathrm{crit}}
\def\noncrit{\mathrm{noncrit}}

\def\Int{\mathrm{Int}}

\def\ren{{\mathrm{ren}}}
\def\Par{\mathrm{Par}}

\def\rT{\mathrm{T}}

\def\dd{\mathrm{d}}

\newcommand{\id}{\mathrm{id}}
\newcommand{\Tr}{\mathrm{Tr}}

\newcommand{\sign}{\mathrm{sign}}

\def\cA{\mathcal{A}}
\def\cB{\mathcal{B}}
\def\cC{\mathcal{C}}

\def\cG{\mathcal{G}}
\def\cH{\mathcal{H}}
\def\cL{\mathcal{L}}
\def\cM{\mathcal{M}}
\def\cN{\mathcal{N}}

\def\cS{\mathcal{S}}
\def\cX{\mathcal{X}}
\def\cU{\mathcal{U}}
\def\cV{\mathcal{V}}

\def\IR{\mathrm{IR}}

\def\cA{\mathcal{A}}

\def\Sol{\mathrm{Sol}}

\def\Proj{\mathrm{Proj}}

\def\rP{\mathrm{P}}

\def\cO{\mathcal{O}}

\def\param{{\mathrm{par}}}

\def\fS{\mathfrak{S}}

\def\h{{\bf h}}

\def\kin{\mathrm{kin}}
\def\pot{\mathrm{pot}}
\def\rj{\mathrm{j}}
\def\wdeg {\mathrm{wdeg}}

\def\graph{\mathrm{graph}}
\def\C{\mathfrak{C}}

\newcommand{\eqdef}{\stackrel{{\rm def.}}{=}}

\def\O{\mathrm{O}}

\def\rP{\mathrm{P}}

\def\bvert{\big{\vert}}

\def\Pot{\mathrm{Pot}}

\def\fH{\mathfrak{H}}
\def\v{\mathrm{v}}

\def\eff{\mathrm{eff}}
\def\r{\mathfrak{r}}

\def\t{\mathrm{t}}
\def\h{\mathfrak{h}}

\def\vol{\mathrm{vol}}

\def\cT{\mathcal{T}}

\def\H{\mathrm{H}}
\def\re{\mathrm{e}}

\title{On the slow roll expansion of one-field cosmological models}

\author{Calin Iuliu Lazaroiu}

\affiliation{Horia Hulubei National
  Institute of Physics and Nuclear Engineering,\\
  Reactorului 30, Bucharest-Magurele, 077125, Romania\\ 
  }

\emailAdd{lcalin@theory.nipne.ro}

\abstract{We study the infrared scale expansion of single field
  cosmological models using the Hamilton-Jacobi formalism, showing
  that its specialization at unit scale parameter recovers the slow
  roll expansion. In particular, we show that the latter coincides
  with a Laurent expansion of the Hamilton-Jacobi function in powers
  of the Planck mass, whose terms are controlled by certain
  recursively-defined polynomials. This allows us to give an explicit
  recursion procedure for constructing all higher order terms of the
  slow roll expansion. We also discuss the corresponding effective
  potential and the action of the universal similarity group.}

\begin{document}

\maketitle

\pagebreak

\section*{Introduction}

An important problem in scalar field cosmology is to develop useful
approximation methods for various regimes. Ideally, such methods should
have a sound conceptual basis which is grounded in the dynamical
features of the model. They should also be {\em effective} in the
sense that they lead to systematic higher order expansion schemes
which are computable in an efficient manner.

In this paper, we consider the scale expansion for single field scalar
cosmology, following the proposal made in \cite{ren,grad} for the much
more general case of multifield models. We show that the scale
expansion recovers the slow roll expansion of single field models
while presenting it in a new conceptual light, which allows one to
give an explicit recursion procedure for computing its higher order
terms.

An important special feature of one-field models (as compared to their
multifield counterparts) is the existence of a straightforward
description of cosmological flow orbits through the so-called {\em
  Hamilton-Jacobi formalism} of Salopek and Bond \cite{SB}. This
encodes the geometry of regular flow orbits through regular {\em
  Hamilton-Jacobi functions}, which are solutions of a certain first
order non-autonomous and non-linear ODE. In this approach, the
cosmological curves of the model can be recovered (up to a constant
translation of cosmological time) by integrating the {\em speed
  equation} of a given flow orbit. The formalism was used for various
purposes in references \cite{LPB, Alvarez,Handley} (see
\cite[Sec. 18.6]{LL} for a brief introduction) and we give a
mathematically precise treatment in the present work.

The scale transformations and universal similarities of \cite{ren}
have a natural action on Hamilton-Jacobi functions, which leads to the
consideration of {\em scale expansions} of the latter. These
expansions can be constructed by looking for solutions of an
appropriate scale deformation of the Hamilton-Jacobi equation which
admit an appropriate expansion in the scale parameter. In this paper,
we focus on the {\em infrared expansion}, which is relevant in the
slow motion regime. We show that the latter is essentially equivalent
with the {\em small Planck mass expansion}, which produces an {\em
  infrared Hamilton-Jacobi function} $\fH_\IR$, distinguished by the
property that it admits a Laurent expansion in powers of the Planck
mass. This follows from the structure of the cosmological and
Hamilton-Jacobi equations and matches the Wilsonian expectation that
infrared dynamics is controlled by the Planck mass, which sets the
natural energy scale of the problem.

The coefficients of the Laurent expansion of $\fH_\IR$ are governed by
a recursion relation, which encodes the condition that $\fH_\IR$
formally satisfies the Hamilton-Jacobi equation. Studying this
recursion, we show that all terms can be expressed in terms of certain
universal multivariate polynomials $Q_n$ with rational coefficients
which are themselves determined by an explicit recursion. These {\em
  Hamilton-Jacobi polynomials} encode all relevant information
required for describing the expansion of $\fH_\IR$, whose coefficients
are recovered by evaluating certain universal rational functions $v_n$
constructed from $Q_n$ on the scalar potential of the model and its
successive derivatives. These results reduce the problem of
constructing $\fH_\IR$ to the recursive computation of Hamilton-Jacobi
polynomials, which can be achieved effectively on a computer. As an
application, we list the polynomials $Q_n$ and rational functions
$v_n$ up to tenth order inclusively.

We also study the homogeneity properties of $Q_n$ and $v_n$ with
respect to two relevant gradings on the algebra $\cA=\R[\N]$ of
polynomials in a countable number of variables, showing that these
properties allow one to rewrite the expansion of $\fH_\IR$ as an
infinite sum of monomials in the {\em potential slow roll parameters}
of the model as defined in \cite{LPB}. This shows that the small
Planck mass expansion is equivalent with the slow roll expansion of
\cite{LPB}, which therefore also describes the infrared expansion of
one field models in the sense of \cite{ren} up to an appropriate
rescaling of the Planck mass by the scale parameter. Performing the
corresponding substitutions to fifth order in the square of the Planck
mass, we show explicitly that our expansion recovers the fifth order
slow roll expansion as computed in \cite{LPB}. Unlike the approach of
loc. cit., which is rather involved and somewhat inefficient due to
the lack of a master recursion formula, our procedure uses the
recursive construction of Hamilton-Jacobi polynomials and hence allows
for the efficient computation of higher order terms in the slow roll
expansion. This opens the way for precision studies of one-field
cosmological models, thus removing one of the technical limitations in
the subject.

The paper is organized as follows. In Section \ref{sec:models}, we
give a mathematical review of one-field cosmological models,
introducing the notation and terminology used in latter sections. In
Section \ref{sec:sim}, we discuss the one-field realization of the
universal similarities and classical RG transformations which were
introduced in a much more general setting in reference
\cite{ren}. Section \ref{sec:HJ} gives a careful treatment of the
Hamilton-Jacobi formalism of \cite{SB} and describes the action of
universal similarities and RG transformations on Hamilton-Jacobi
functions.  Section \ref{sec:deformed} introduces the deformed
Hamilton-Jacobi equation and studies the formal infrared scale
expansion of Hamilton Jacobi functions. We express the terms of this
expansion through Hamilton-Jacobi polyomials and study their
bihomogeneity properties, showing that the infrared expansion at unit
scale parameter is equivalent with the slow roll expansion of
\cite{LPB}, which turns out to be equivalent with a Laurent expansion
in powers of the Planck mass. Section \ref{sec:conclusions} presents
our conclusions and some directions for further research. Appendix
\ref{app:jets} recalls some basic facts about jets of univariate
real-valued functions while Appendix \ref{app:list} lists the
Hamilton-Jacobi polynomials $Q_n$ and associated rational functions
$v_n$ up to order $10$.

\section{Single field cosmological models}
\label{sec:models}

By a {\em single field cosmological model} we mean a classical
cosmological model with one real scalar field (called {\em inflaton})
and target space $\R$, which is derived from the following action on a
spacetime with topology $\R^4$:
\ben
\label{S}
\cS[g,\varphi]=\int_{\R^4} \vol(g')\cL[g,\varphi]~~.
\een
Here $\vol(g)$ is the volume form of the Lorentzian spacetime metric $g$ defined on
$\R^4$ (which we take to be of ``mostly plus'' signature) and:
\ben
\label{cL}
\cL[g,\varphi]=\frac{M^2}{2} \mathrm{R}(g)-\frac{1}{2}g^{\mu \nu} \pd_\mu \varphi \pd_\nu \varphi -\Phi\circ \varphi~~,
\een
where the positive constant $M$ is the reduced Planck mass, $\mathrm{R}(g)$ is
the Ricci scalar of $g$, the {\em coordinatized scalar field}
$\varphi:\R^4\rightarrow \R$ is a smooth map and the smooth function
$\Phi:\R\rightarrow \R_{>0}$ is the {\em coordinatized scalar
  potential}. For simplicity, we assume throughout this paper that
$\Phi$ is strictly positive everywhere. The action \eqref{S} is
well-defined when $g$ and $\varphi$ obey appropriate conditions at
infinity.

The Lagrangian density of the scalar field can be written in the
standard form of a sigma model with potential:
\ben
\label{cLs}
\cL_s[{\hat \varphi}]=-\frac{1}{2}\Tr_g{\hat \varphi}^\ast(\cG) -{\hat \Phi}\circ {\hat \varphi}~~,
\een
where $\cG$ is a Riemannian metric on an abstract one-dimensional manifold
$\cM$ which is diffeomorphic with the real line and plays the role of
target space of the {\em uncoordinatized scalar field}, which is a smooth function ${\hat
\varphi}:\R^4\rightarrow \cM$. The {\em uncoordinatized scalar potential}
${\hat \Phi}:\cM\rightarrow \R_{>0}$ is a smooth function.
The coordinatized formulation given above is recovered as follows. Let
$x:\cM\stackrel{\sim}{\rightarrow} \R$ be a global coordinate
(uniquely determined by $\cG$ and by a choice of orientation for
$\cM$) for which $\cG$ takes the standard Euclidean form:
\ben
\label{cG}
\dd s_\cG^2=\dd x^2~~.
\een
Then the coordinatized scalar field $\varphi$ and coordinatized scalar
potential $\Phi$ used in \eqref{cL} are recovered as:
\ben
\label{coordinatized}
\varphi=x\circ {\hat \varphi}~~,~~\Phi\eqdef {\hat \Phi}\circ x^{-1}~~.
\een
In particular, the formulation in terms of $\varphi$ and $\Phi$ has a
hidden dependence on the choice of the Riemannian metric $\cG$ on the
target space $\cM$ (which, together with the orientation of $\cM$,
determines the choice of the Euclidean coordinate $x$). This fact will
play a role in the discussion of universal similarities given in
Subsection \ref{sec:sim} below, namely those
similarities differ from the general treatment given in \cite{ren}
when expressed in the coordinatized formulation.

\subsection{Parameterization, phase space and the rescaled Hubble function}

\noindent The model is parameterized by the pair $(M_0,\Phi)$,
where\footnote{We prefer to use $M_0$ instead of $M$ since this
simplifies various formulas later on.}:
\ben
\label{M0def}
M_0\eqdef M\sqrt{\frac{2}{3}}
\een
is the {\em rescaled Planck mass} and the scalar potential $\Phi\in
\cC^\infty(\R)$ is an everywhere positive function. We denote by:
\be
\Crit\Phi\eqdef \{x\in \R~\vert~(\dd \Phi)(x)=0\}~~\mathrm{and}~~\Noncrit\Phi\eqdef \R\setminus \Crit\Phi
\ee
the {\em critical and noncritical sets} of $\Phi$. 

By definition, the {\em phase space} $\cT$ of the model is the total space
of the tangent bundle of the target space $\R$:
\be
\cT\eqdef T\R=\R^2~,
\ee
whose bundle projection we denote by $\pi_1:\cT\rightarrow \R$. We
denote by $\pi_2:T\R\rightarrow \R$ the vertical projection of the
trivial vector bundle $T\R$.  The {\em position} and {\em speed}
coordinates $x,v$ are the base and fiber coordinates, which give the
Cartesian coordinate system $(x,v)$ of $\R^2$.  Thus $\pi_1$ and
$\pi_2$ are the ordinary Cartesian projections:
\be
\pi_1(x,v)\eqdef x~~,~~\pi_2(x,v)\eqdef v~~\forall (x,v)\in \cT=\R^2~~.
\ee
Moreover, we denote by:
\be
Z\eqdef \{(x,0)\vert x\in \R\}=\R\times \{0\}\subset \cT
\ee
the image of the zero section of $T\R$. The total space of the slit
tangent bundle ${\dot T}\R$ is the complement of the set $Z$ inside $\cT$:
\be
\dot{T}\R =\cT\setminus Z=\R\times \R^\times~~.
\ee
This open set has two connected components (namely the upper and lower
open half-planes in $\R^2$), which we denote by:
\be
\cT^+\eqdef \{(x,v)\in \R^2\vert v>0\}~~\mathrm{and}~~\cT^-\eqdef \{(x,v)\in \R^2\vert v<0\}~~.
\ee

\begin{definition}
The {\em canonical phase space lift} of a smooth curve
$\varphi:I\rightarrow \R$ is the smooth curve $c(\varphi):I\rightarrow
\R^2$ defined through:
\ben
\label{cvarphi}
c(\varphi)(t)\eqdef (\varphi(t),\dot{\varphi}(t))~~\forall t\in I~~.
\een 
The {\em phase space manifold} of $\varphi$ is the smooth
one-dimensional submanifold $\cO_\varphi$ of the phase space
$\cT=\R^2$ defined through:
\ben
\label{cOvarphi}
\cO_\varphi=c(\varphi)(I)=\{(\varphi(t),\dot{\varphi}(t))~\vert~t\in I\}~~.
\een
\end{definition}

\begin{definition}
The {\em rescaled Hubble function} of the model $(M_0,\Phi)$ is the map
$\cH:\cT\rightarrow \R_{>0}$ defined through:
\ben
\cH(x,v)\eqdef \frac{1}{M_0}\sqrt{v^2+2\Phi(x)}~~\forall~~(x,v)\in \R^2~~.
\een
\end{definition}

\subsection{The cosmological equation}

\noindent The classical single field cosmological model parameterized
by $(M_0,\Phi)$ is obtained by assuming that $g$ is a
Friedmann-Lemaitre-Robinson-Walker (FLRW) metric with flat spatial
section:
\ben
\label{FLRW}
\dd s^2_g=-\dd t^2+a(t)^2\sum_{i=1}^3 \dd x_i^2
\een
(where $a\in$ is a smooth positive function) and that $\varphi$
depends only on the {\em cosmological time} $t\eqdef x^0\in
\R$. Define the {\em Hubble parameter} through:
\ben
\label{H}
H(t)\eqdef \frac{\dot{a}(t)}{a(t)}~~,
\een
where the dot denotes derivation with respect to $t$. 

\begin{prop}
When $H>0$ (which we assume throughout this paper), the variational
equations of the action \eqref{S} reduce to the system formed by the
{\em cosmological equation}:
\ben
\label{eomsingle}
\ddot{\varphi}(t)+\cH(\varphi(t),\dot{\varphi}(t))\dot{\varphi}(t)+ \Phi'(\varphi(t))=0~~
\een
together with the condition:
\ben
\label{Hvarphi}
H(t)=\frac{1}{3}\cH_\varphi(t)\eqdef \frac{1}{3}\cH(\dot{\varphi}(t),\dot{\varphi}(t))~~.
\een
\end{prop}

\begin{proof}
Follows by direct computation.
\end{proof}  

\subsection{Cosmological curves and cosmological orbits}

\noindent An interval $I\subset \R$ is called non-degenerate if it is
nonempty and not reduced to a point.  Let $\Int$ be the set of all
non-degenerate intervals in $\R$ and:
\be
\C(\R)\eqdef \sqcup_{I\in \Int}\cC^\infty(I,\R)
\ee
be the set of all smooth curves in $\R$.

\begin{definition}
The solutions $\varphi:I\rightarrow \R$ of the cosmological equation
\eqref{eomsingle} (where $I$ is a non-degenerate interval) are called
      {\em cosmological curves}, while their images $\varphi(I)\subset
      \R$ (which are intervals on the real axis) are called {\em
        cosmological orbits}.  The cosmological curve $\varphi$ is
      called {\em pointed} if $0\in I$.
\end{definition}

\noindent We say that a cosmological orbit is {\em trivial} if it
is reduced to a point. This happens iff the orbit is the image of a 
constant cosmological curve. Given a cosmological curve
$\varphi$, relations \eqref{H} and \eqref{Hvarphi} determine the scale
parameter $a$ of the FLRW metric \eqref{FLRW} up to a positive
multiplicative constant:
\ben
\label{a}
a(t)=C e^{\frac{1}{3}\int_{t_0}^t \dd t \cH_\varphi(t)}~~\forall t\in I~~,
\een
where $C>0$ and $t_0\in I$ are arbitrary.

We denote by $\Sol(M_0,\Phi)\subset \C(\R)$ the set of cosmological
curves of the model $(M_0,\Phi)$ and by $\Sol_m(M_0,\Phi)$ the subset
of pointed maximal cosmological curves. Given $T\in \R$, the {\em
  $T$-translate} of a smooth curve $\varphi\in \C(\R)$ is the curve
$\varphi^T:I^T\rightarrow \R$ defined through:
\be
I^T\eqdef I-T=\{t-T~\vert~t\in I\}~~,~~\varphi^T(t)\eqdef \varphi(t+T)~~\forall t\in I^T.
\ee
This gives an action of the abelian group $(\R,+)$ on the set
$\C(\R)$. Since the cosmological equation \eqref{eomsingle} is
autonomous, the $T$-translate of any cosmological curve is again a
cosmological curve for all $T\in \R$ and hence the set
$\Sol(M_0,\Phi)$ is invariant under this action. Notice that any
cosmological curve admits translates which are pointed.

A cosmological curve $\varphi:I\rightarrow \R$ need not be an immersion.
Accordingly, we define the {\em singular and regular parameter sets}
of $\varphi$ through:
\beqan
&& I_\sing\eqdef \{t\in I~\vert~\dot{\varphi}(t)=0\}\subset I\\
&& I_\reg\eqdef I\setminus I_\sing= \{t\in I~\vert~\dot{\varphi}(t)\neq 0\}\subset I~~.
\eeqan
When $\varphi$ is non-constant, it is easy to see that $I_\sing$
is an at most countable set and that the restriction of $\varphi$ to
each connected component of $I_\reg$ is a homeomorphism onto its image
and hence an embedded curve in $\R$.

The sets of {\em critical and noncritical times} of a cosmological
curve $\varphi:I\rightarrow \R$ are defined through:
\beqan
&& I_\crit\eqdef \{t\in I~\vert~\varphi(t)\in \Crit\Phi \}\nn\\
&& I_\noncrit \eqdef I\setminus I_\crit \eqdef \{t\in I~\vert~\varphi(t)\not\in \Crit\Phi\}~~.
\eeqan
The cosmological curve $\varphi$ is called {\em noncritical} if
$I_\crit=\emptyset$, which means its orbit is contained in the
non-critical set of $\Phi$. It is easy to see that a cosmological
curve $\varphi$ is constant iff its image coincides with a critical
point of $\Phi$, which in turn happens iff there exists some $t\in
I_\sing$ such that $\varphi(t)\in \Crit \Phi$. Hence for any
{\em non-constant} cosmological curve we have:
\be
I_\sing\cap I_\crit=\emptyset~~.
\ee

\subsection{The cosmological dynamical system, cosmological flow curves and cosmological flow orbits}

\noindent The cosmological equation \eqref{eomsingle} is equivalent
with the following first order system of ODEs:
\beqan
\label{dyneq}
&& \dot{x}=v\nn\\
&& \dot{v}=-\cH(x,v) v-\Phi'(x)~~,
\eeqan
which describes the integral curves of the vector field $S\in
\cX(\cT)$ defined through:
\be
S(x,v)\eqdef v\pd_x- \left[\cH(x,v) v+\Phi'(x)\right] \pd_v~~\forall (x,v)\in \cT=\R^2~~,
\ee
which is the single field model incarnation of the cosmological
semispray of \cite{ren}. The flow of this vector field on the phase
space $\cT=\R^2$ is called the {\em cosmological flow} of the model
$(M_0,\Phi)$. The planar dynamical system $(\R^2,S)$ is called the
{\em cosmological dynamical system} (see \cite{Palis,Katok} for an
introduction to geometric dynamical systems theory).

\begin{definition}
An integral curve $\gamma:I\rightarrow \R^2$ of the vector field $S$
is called a {\em cosmological flow curve} while its image $\gamma(I)$
(considered as a submanifold of $\cT=\R^2$) is called a {\em
cosmological flow orbit}.
\end{definition}

\noindent Notice that $S$ vanishes precisely at the
{\em trivial lifts} ${\tilde x}\eqdef (x,0)\in Z$ of the critical
points $x\in \Crit \Phi$ of the scalar potential. Accordingly, a
cosmological flow curve is constant iff its orbit coincides with a
point of the trivial lift:
\be
\widetilde{\Crit \Phi}\eqdef \{(x,0)\vert x\in \Crit \Phi\}\subset Z
\ee
of the critical set of $\Phi$; in this case, we say that the
cosmological flow orbit is {\em trivial}. It is easy to see that any
non-constant cosmological flow curve $\gamma$ is an immersion. Since
the cosmological energy is strictly decreasing along such a curve (see
Proposition \ref{prop:dissip} below), it follows that $\gamma(I)$ has
no self-intersections and hence is an embedded submanifold of
$\R^2$. The canonical phase space lift $c(\varphi)$ of any
cosmological curve $\varphi:I\rightarrow \R$ is a cosmological flow
curve $c(\varphi):I\rightarrow \R^2$.  Conversely, any cosmological
flow curve $\gamma:I\rightarrow \R^2$ determines a cosmological curve
$\varphi\eqdef \pi\circ \gamma$ by projection to $\R$:
\be
\varphi(t)=\pi(\gamma(t))~~\forall t\in I
\ee
and this cosmological curve satisfies $c(\varphi)=\gamma$. In
particular, we have $\varphi(I)=\pi(\cO)$, where $\cO\eqdef \gamma(I)$
is the cosmological flow orbit determined by $\gamma$.

Also notice that a cosmological flow orbit $\cO$ determines cosmological
curves $\varphi:I\rightarrow \R$ such that $\cO=c(\varphi)(I)$ through
the first order ODE:
\be
F(\varphi(t),\dot{\varphi}(t))=0~~,
\ee
where $F(x,v)=0$ is the implicit equation of $\cO$. Since this ODE is
autonomous, it determines $\varphi$ only up to translation of $t$ by an 
arbitrary constant.

\subsection{Regular cosmological flow curves and regular cosmological flow orbits}
\label{subsec:reg}

\begin{definition}
A smooth curve $\varphi:I\rightarrow \R$ is called {\em regular} if it is
an immersion, i.e. if its derivative does not vanish on $I$.
\end{definition}

\noindent Let $\varphi:I\rightarrow \R$ be a regular smooth curve
whose image (which is an interval in $\R$) we denote by $J\eqdef
\varphi(I)$. Then $\dot{\varphi}$ has constant sign on $I$, which we
denote by $\zeta_\varphi$ and call the {\em signature} of
$\varphi$. Since the map $\varphi:I\rightarrow J$ is a diffeomorphism,
we can use $x\in J$ as a parameter for $\varphi$. Set:
\be
\t_\varphi\eqdef \varphi^{-1}:J\rightarrow I
\ee
(so that $t=\t_\varphi(x)$) and let $v_\varphi\eqdef
\dot{\varphi}:I\rightarrow \R^\times$ be the speed of $\varphi$.

\begin{definition}
\label{def:v}
The {\em speed function} of the regular curve $\varphi:I\rightarrow \R$ is the map:
\be
\v_\varphi\eqdef v_\varphi\circ \t_\varphi=v_\varphi\circ \varphi^{-1}:J\eqdef \varphi(I)\rightarrow \R^\times~~.
\ee
\end{definition}

\noindent Notice that:
\be
\v_\varphi=\frac{1}{\t'_\varphi}~~\mathrm{and}~~\sign(\v_\varphi)=\sign(\t'_\varphi)=\zeta_\varphi~~,
\ee
where the prime denotes derivation with respect to $x$. The regular
curve $\varphi$ can be recovered up to a time translation from its
speed function $\v:=\v_\varphi$ by solving the {\em speed equation}:
\ben
\label{speedeq}
\dot{\varphi}(t)=\v(\varphi(t))~~,
\een
which determines $\varphi$ up to a shift of $t$ by an arbitrary
constant. The speed equation amounts to:
\be
\t'_\varphi(x)=\frac{1}{\v(x)}~~(x\in J)
\ee
and has general solution: 
\be
\t_\varphi(x)=C+\int_{x_0}^x \frac {\dd y}{\v(y)}~~\forall x\in J~~,
\ee
where $C\in \R$ and $x_0\in J$ are arbitrary.

\begin{definition} A one-dimensional connected smooth submanifold
$\cO$ of $\cT=\R^2$ is called {\em regular} if $\cO\cap
  Z=\emptyset$. A smooth phase space curve $\gamma:I\rightarrow \cT$
  is called {\em regular} if $\gamma(I)$ is a regular submanifold of
  $\cT$, i.e. if $\pi_2(\gamma(t))\neq 0$ for all $t\in I$.
\end{definition}

\noindent The discussion above implies that a connected submanifold
$\cO\subset \cT$ is the image of the phase space lift $c(\varphi)$ of
a regular curve $\varphi:I\rightarrow \R$ iff it is the graph of a
smooth function $\v_\cO:\pi(\cO)\rightarrow \R^\times$, which we call
the {\em speed function of $\cO$}. The following statement is obvious:

\begin{prop}
\label{prop:cO}
Let $\varphi:I\rightarrow \R$ be a smooth curve and let:
\be
\gamma\eqdef c(\varphi)=(\varphi,\dot{\varphi}):I\rightarrow \R^2
\ee
be its phase space lift. Let $\cO\eqdef \gamma(I)$ be the phase space
manifold of $\varphi$.  Then the following statements are equivalent:
\begin{enumerate}[(a)]
\item $\varphi$ is a regular curve.
\item $\gamma$ is a regular phase space curve.
\item $\cO$ is a regular one-dimensional submanifold of $\R^2$.
\end{enumerate} In this case, we have $\cO\subset \cT^\zeta$, where
$\zeta\in \{-1,1\}$ is the signature of $\varphi$ and the speed
function of $\varphi$ coincides with the speed function of $\cO$. In
particular, $\cO$ has equation:
\ben
\label{vv}
v=\v(x)~~(x\in \pi(\cO)),
\een
where $\v$ is this speed function.
\end{prop}

\begin{prop}
\label{prop:floworbiteq}
Let $\varphi:I\rightarrow \R$ be a regular curve and set $J\eqdef
\varphi(I)$. Then the following statements are equivalent:
\begin{enumerate}[(a)]
\item $\varphi$ is a cosmological curve, i.e. its phase space manifold
$\cO$ is a cosmological flow orbit.
\item The speed function $\v$ of $\varphi$ (which coincides with that
of $\cO$) satisfies the {\em regular flow orbit equation}:
\ben
\label{floworbiteq}
\v'(x)=-\cH(x,\v(x))-\frac{\Phi'(x)}{\v(x)}~~.
\een
\end{enumerate}
\end{prop}

\begin{proof}
Follows by noticing that \eqref{floworbiteq} is the ratio of the
second equation in \eqref{dyneq} to the first equation. 
\end{proof}

\noindent By Proposition \ref{prop:cO}, a cosmological flow curve
$\gamma:I\rightarrow \cT=\R^2$ and its cosmological flow orbit
$\cO=\gamma(I)\subset \R^2$ are regular iff the cosmological curve
$\varphi=\pi\circ \gamma:I\rightarrow \R$ is regular. In this case, we
define the signature of $\gamma$ and $\cO$ to equal the signature
$\zeta$ of $\varphi$. By Proposition \ref{prop:cO}, we have
$\cO=\gamma(I)\subset \cT^\zeta$. The map which takes a regular
one-dimensional submanifold $\cO$ of $\R^2$ to its speed function
$\v:\pi(\cO)\rightarrow \R^\times$ restricts to a bijection between regular
cosmological flow orbits and solutions of the regular flow orbit
equation, which presents the former as the graphs of the latter.

\begin{remark}
The function $\fH:\pi(\cO)\rightarrow \R_{>0}$, $\fH(x)\eqdef \cH(x,\v(x))$
which appears in the right hand side of \eqref{floworbiteq} is called
the {\em Hamilton-Jacobi function} of the cosmological flow orbit
$\cO$ (see Subsection \ref{sec:HJfunctions}).  We will see in Section
\ref{sec:HJ} that the regular flow orbit equation \eqref{floworbiteq}
is equivalent with the {\em Hamilton-Jacobi equation} by the change of
function $\v(x)\rightarrow \fH(x)$, which leads to
the {\em Hamilton-Jacobi parameterization} of regular cosmological
flow orbits. This corresponds to using {\em Hamilton-Jacobi
  coordinates} $(x,h)$ on $\cT^\pm$ instead of the Cartesian
coordinates $(x,v)$ used above.
\end{remark}

\subsection{Basic cosmological observables and the dissipation equation}

\begin{definition}
A {\em basic local cosmological observable} is a real-valued function
$F:\cT=\R^2\rightarrow \R$. The {\em dynamical reduction} of $F$ along
a cosmological curve $\varphi:I\rightarrow \R$ is the function
$F_\varphi:I\rightarrow \R$ defined through:
\be
F_\varphi(t)\eqdef F(\varphi(t),\dot{\varphi}(t))~~\forall t\in I~~.
\ee
\end{definition}

\noindent Notice that the rescaled Hubble function $\cH$ is a basic
local cosmological observable. Its dynamical reduction $\cH_\varphi$
along a cosmological curve $\varphi$ is also called the {\em rescaled
Hubble parameter} of that curve.

\begin{definition}
The {\em cosmological energy function} is the basic local cosmological
observable $E:\R^2\rightarrow \R_{>0}$ defined through:
\be
E(x,v)\eqdef \frac{1}{2}v^2+\Phi(x)~~\forall (x,v)\in \R^2~~.
\ee
The {\em cosmological kinetic energy function} $E_\kin:\R^2\rightarrow
\R_{\geq 0}$ and the {\em cosmological potential energy function}
$E_\pot:\R^2\rightarrow \R_{>0}$ are the basic local cosmological
observables defined through:
\ben
E_\kin(x,v)\eqdef \frac{1}{2}v^2~~,~~E_\pot(x,v)\eqdef \Phi(x)~~\forall (x,v)\in \R^2~~.
\een
\end{definition}

\noindent With these definitions, we have:
\ben
E=E_\kin+E_\pot~~\mathrm{and}~~\cH=\frac{1}{M_0}\sqrt{2E}~~.
\een

\begin{prop}
\label{prop:dissip}
The dynamical reduction of $E$ along any cosmological curve
$\varphi:I\rightarrow \R$ satisfies the {\em cosmological dissipation
equation}:
\ben
\label{DissipEq}
\frac{\dd E_\varphi(t)}{\dd t}=-\frac{\sqrt{2E_\varphi(t)}}{M_0}\dot{\varphi}(t)^2\Longleftrightarrow \frac{\dd \cH_\varphi}{\dd t}=-\frac{\dot{\varphi}^2}{M_0^2}~~.
\een
In particular, $E_\varphi$ and $\cH_\varphi$ are strictly decreasing functions of $t$
if $\varphi$ is a non-constant cosmological curve.
\end{prop}

\begin{proof} Follows immediately by using the cosmological equation
\eqref{eomsingle}.
\end{proof}

\begin{cor}
The {\em Hubble inequality}:
\ben
\label{cHineq}
\cH_\varphi(t_1)-\cH_\varphi(t_2)\geq \frac{(\varphi(t_1)-\varphi(t_2))^2}{M_0^2(t_2-t_1)}
\een
holds for any cosmological curve $\varphi:I\rightarrow \R$ and any
$t_1,t_2\in I$ such that $t_1\leq t_2$.
\end{cor}

\begin{proof}
Integrating \eqref{DissipEq} from $t_1$ to $t_2$ gives:
\be
\cH_\varphi(t_1)-\cH_\varphi(t_2)=\frac{1}{M_0^2}\int_{t_1}^{t_2}\dd t \dot{\varphi}(t)^2
\ee
When the endpoints $\varphi(t_1)$ and $\varphi(t_2)$ are fixed, the
integral in the right hand side is minimized for the curve $\varphi_0$ which satisfies
the equation $\ddot{\varphi}_0=0$, which has solution
$\varphi_0(t)=\varphi(t_1)+\frac{\varphi(t_2)-\varphi(t_1)}{t_2-t_1}
(t-t_1)$. The minimum value is: 
\be
\int_{t_1}^{t_2}\dd t \dot{\varphi}_0(t)^2=\frac{(\varphi(t_2)-\varphi(t_1))^2}{t_2-t_1}~~,
\ee
which gives \eqref{cHineq}.
\end{proof}

\begin{prop}
\label{prop:regdissip}
Suppose that $\varphi:I\rightarrow \R$ is a regular curve. Then
the dissipation equation \eqref{DissipEq} for $\varphi$ is equivalent
with the cosmological equation \eqref{eomsingle} and hence also with the
regular flow orbit equation \eqref{floworbiteq}.
\end{prop}

\begin{proof}
Let $J\eqdef \varphi(I)$. Since $\varphi$ is a regular curve, the map
$\varphi:I\rightarrow J$ is a diffeomorphism and hence the
cosmological equation \eqref{eomsingle} for $\varphi$ is equivalent
with the regular flow orbit equation \eqref{floworbiteq} by
Proposition \ref{prop:floworbiteq}. In turn, the regular flow orbit
equation can be written as:
\be
\frac{\dd}{\dd x} [\cH(x,\v(x))^2]=-\frac{2\cH(x,\v(x)) \v(x)}{M_0^2} \Longleftrightarrow \frac{\dd}{\dd x}\cH(x,\v(x)) =-\frac{\v(x)}{M_0^2}~~,
\ee
where we used the fact that $\cH>0$. This is equivalent with the
dissipation equation \eqref{DissipEq} when combined with the speed
equation \eqref{speedeq} of $\varphi$.
\end{proof}

\section{Universal similarities and the dynamical renormalization group}
\label{sec:sim}

\noindent Cosmological models with an arbitrary number of scalar
fields admit a two-parameter group of similarities which relate the
cosmological curves of models with different parameters (see
reference \cite{ren}). This group acts naturally on the rescaled Planck mass,
scalar field metric and scalar potential of multifield models, as well
as on their phase space and basic observables. As we shall see below,
the realization of these actions in single field cosmological models
is somewhat special since our formulation uses the {\em coordinatized}
scalar field and scalar potential -- unlike the abstract formulation
\eqref{cLs}, whose multifield version was used in loc. cit.

Let $\Pot(\R)$ be the set of smooth and everywhere positive functions
defined on $\R$. Then the space of parameters of single-field models
with smooth and positive scalar potential is:
\be
\Par\eqdef\R_{>0}\times \Pot(\R)~~,
\ee
where the first factor parameterizes the rescaled Planck mass
$M_0$. Since most physically interesting basic observables depend on
the parameters $M_0$ and $\Phi$, in this section we will view them as
functions $F:\Par\times \R^2\rightarrow \R$ rather than as functions
from $\cT=\R^2$ to $\R$.

\begin{definition}
\label{def:scaletf}
Let $\epsilon>0$. The {\em $\epsilon$-scale transform} of a smooth map
$f:I\rightarrow \R$ (where $I$ is a non-degenerate interval) is the
smooth map $f_\epsilon:I_\epsilon\rightarrow \R$ defined through:
\be
I_\epsilon \eqdef \epsilon I=\{\epsilon t \vert t\in I\}
\ee
and:
\be
f_\epsilon(t)\eqdef f(t/\epsilon)~~\forall t\in I_\epsilon~~.
\ee
\end{definition}

\noindent Notice that $f$ and $f_\epsilon$ have the same image. Scale
transforms define an action of the multiplicative group $\R_{>0}$ on
the set $\C(\R)$ of smooth curves in $\R$. The parameter $\epsilon$ is
called the {\em scale parameter}.

\begin{remark}
\label{rem:psi}
Let $\psi$ be the scale transform of $\varphi$ at scale parameter
$\epsilon=M_0$:
\be
\psi\eqdef \varphi_{M_0}~~{\mathrm i.e.}~~\psi(t)=\varphi(t/M_0)\Longleftrightarrow \varphi(t)=\psi(M_0 t)~~.
\ee
Then:
\be
\frac{\dd^k \varphi}{\dd  t^k}=M_0^k \frac{\dd^k \psi}{\dd  t^k}~~\forall k\geq 0
\ee
and the cosmological equation \eqref{eomsingle} is equivalent with:
\ben
\label{eompsi}
M_0^2\ddot{\psi}(t)+\sqrt{M_0^2 \dot{\psi}(t)^2+2\Phi(\psi(t))} \dot{\psi}(t)+ \Phi'(\psi(t))=0~~.
\een
This form of the cosmological equation has the property that $M_0$
``counts the number of derivatives of $\psi$''. When $M_0$ is small,
one can look for solutions $\psi$ of \eqref{eompsi} which are expanded
as a power series in $M_0$:
\ben
\label{psiexp}
\psi(t)=\sum_{n\geq 0} \psi_n(t) M_0^n\Longleftrightarrow \varphi(t)=\sum_{n\geq 0} \psi_n(M_0 t) M_0^n=\sum_{n\geq 0} \varphi_n(t) M_0^n~~,
\een
where $\varphi_n(t)\eqdef \psi_n(M_0 t)$ for all $n\geq 0$.
Substituting this expansion into \eqref{eompsi} gives a differential
recursion relation for $\psi_n$, together with the zeroth order
condition:
\be
\sqrt{2\Phi(\psi_0(t))} \dot{\psi}_0(t)+ \Phi'(\psi_0(t))=0\Longleftrightarrow M_0\dot{\psi}_0(t)=-V'_0(\psi_0(t))\Longleftrightarrow \dot{\varphi}_0(t)=-V'_)(\varphi_0(t))~~,
\ee
where $V_0\eqdef M_0\sqrt{2\Phi}$ is the ``leading classical effective
potential'' of reference \cite{ren}. This shows that $\varphi_0$ is a
gradient flow curve of $V_0$. In latter sections, we will construct the
rescaled Planck mass expansion \eqref{psiexp} indirectly, namely by
expanding the Hamilton-Jacobi function of the cosmological flow orbit
of $\varphi$ as a Laurent series in $M_0$. Notice that \eqref{psiexp}
can be viewed as a particular instance of the infrared scale expansion
proposed in \cite{ren}, namely the scale expansion considered at scale
parameter $\epsilon=M_0$. This is natural from the Wilsonian perspective
adopted in loc. cit., since $M_0$ sets the natural energy scale of cosmological
dynamics. 
\end{remark}

\noindent Following \cite{ren}, we define:

\begin{definition}
The {\em cosmological similarity group} is the Abelian multiplicative group
$\rT\eqdef \R_{>0}\times \R_{>0}$ (where elements multiply componentwise). 
\end{definition}

\begin{remark}
\label{rem:Hscaling}
Notice that $\rT$ is isomorphic with the additive group $(\R^2,+)$
through the componentwise exponential map:
\be
\exp:\R^2\rightarrow \rT~~,~~\exp(s_1,s_2)\eqdef (e^{s_1},e^{s_2})~~\forall s_1,s_2\in \R~~.
\ee
However, the multiplicative formulation is more convenient for our purpose. 
\end{remark}

\begin{definition}
The {\em parameter action} $\rho_\param$ is the action of $\rT$ on the
set of parameters $\Par=\R_{>0}\times \Pot(\R)$ given by:
\be
\rho_\param(\lambda,\epsilon)(M_0,\Phi)\eqdef (\lambda^{1/2}M_0,\frac{\lambda}{\epsilon^2} \Phi_{\lambda^{1/2}})~~\forall (\lambda,\epsilon)\in \rT~~\forall (M_0,\Phi)\in \Par~~.
\ee
The {\em curve action} $\rho_0$ is the action of $\rT$ on the
set $\C(\R)$ given by:
\be
\rho_0(\lambda,\epsilon)(\varphi)=\lambda^{1/2}\varphi_\epsilon~~\forall (\lambda,\epsilon)\in \rT~~\forall \varphi\in \C(\R)~~.
\ee
The {\em similarity action} is the action $\rho\eqdef
\rho_\param\times \rho_0$ of $\rT$ on the set $\Par\times \C(\R)$:
\be
\rho(\lambda,\epsilon)(M_0,\Phi,\varphi)\eqdef (\lambda^{1/2}M_0, \frac{\lambda}{\epsilon^2}\Phi_{\lambda^{1/2}},\lambda^{1/2}\varphi_\epsilon)~~\forall (\lambda,\epsilon)\in \rT~~\forall
(M_0,\Phi,\varphi)\in \Par\times\C(\R)~~.
\ee
\end{definition}

The transformations $\rho(\lambda,1)$ are called {\em parameter
  homotheties}, while the transformations $\rho(1,\epsilon)$ are called
{\em scale similarities}. As in \cite{ren}, scale similarities fix the
abstract scalar field ${\hat \varphi}:\R^4\rightarrow \cM$ of
\eqref{cLs} while acting as follows on the system $(M_0,\cG,{\hat
  \Phi})$, where $\cG$ is the target space metric \eqref{cG} of the
abstract single field model and ${\hat \Phi}:\cM\rightarrow \R_{>0}$
is the abstract scalar potential:
\be
(M_0,\cG,{\hat \Phi})\rightarrow (\lambda^{1/2} M_0,\lambda \cG,\lambda {\hat \Phi})~~.
\ee
Then the formulation given above follows by noticing that under such a
transformation the Euclidean coordinate $x$ of \eqref{cG}
changes as:
\be
x\rightarrow \lambda^{1/2} x
\ee
and hence the coordinatized scalar field $\varphi=x\circ { \hat
  \varphi}$ and potential $\Phi={\hat \Phi}\circ x^{-1}$ of
\eqref{coordinatized} change as:
\be
\varphi\rightarrow \lambda^{1/2}\varphi~~,~~\Phi\rightarrow \lambda \Phi_{\lambda^{1/2}}~~.
\ee
In particular, a parameter homothety induces a rescaling of the {\em
  coordinatized} scalar field $\varphi$, because the latter is defined
using the Euclidean coordinate $x$ on $\cM$, which itself depends on the
target space metric $\cG$. Notice that the action $\rho_\param$ of $\rT$ on
the parameter set $\Par$ is free.

Let $\Sol(M_0,\Phi)\subset \C(\R)$ be the set of cosmological curves of
the model parameterized by $(M_0,\Phi)\in \Par$.

\begin{prop}
\label{prop:sim}
We have:
\be
\rho_0(\lambda,\epsilon)(\Sol(M_0,\Phi))=\Sol(\rho_\param(M_0,\Phi))~~\forall (\lambda,\epsilon)\in \rT~~\forall (M_0,\Phi)\in \Par~~.
\ee
In particular, the set:
\be
\Sol=\sqcup_{(M_0,\Phi)\in \Par}\Sol(M_0,\Phi)=\{(M_0,\Phi,\varphi)\in \Par\times \C(\R)\vert \varphi \in \Sol(M_0,\Phi)\}
\ee
is invariant under the similarity action $\rho$ of $\rT$.
\end{prop}

\begin{proof}
Let $\varphi$ be a solution of the cosmological equation of the model
$(M_0,\Phi)$. Then it is easy to see that
$\lambda^{1/2}\varphi_\epsilon$ is a solution of the cosmological
equation of the model $(\lambda^{1/2}M_0,\frac{\lambda}{\epsilon^2}
\Phi_{\lambda^{1/2}})$. In other words, the cosmological equation is
invariant under the universal similarity action.
\end{proof}

\begin{remark}
The parameter homothety $\rho(\lambda,1)$ with parameter
$\lambda=M_0^{-2}$ can be used to eliminate the rescaled Planck mass
$M_0$. For this, let:
\ben
\varphi_0\eqdef \rho_0(M_0^{-2},1)(\varphi)=\frac{\varphi}{M_0}
\een
and define the {\em reduced
scalar potential} $\Phi_0:\R\rightarrow \R_{>0}$ and the {\em reduced
Hubble function} $\cH_0:\cT\rightarrow \R_{>0}$ by:
\ben
\Phi_0(x)\eqdef \frac{1}{M_0^2}\Phi_{M_0^{-1}}(x)=\frac{\Phi(M_0 x)}{M_0^2}~~,~~\cH_0(x,v)\eqdef \sqrt{v^2+2\Phi_0(x)}~~.
\een
Then:
\be
\rho_\param(M_0^{-2},1)(M_0,\Phi)=(1,\Phi_0)
\ee
and we have $\Phi(\varphi(t))=M_0^2\Phi_0(\varphi_0(t))$ and
$\cH(\varphi(t),
\dot{\varphi}(t))=\cH_0(\varphi_0(t),\dot{\varphi}_0(t))$.  Since
$\Phi'_0(x)=\frac{1}{M_0}\Phi'(M_0 x)$, we also have
$\Phi'(\varphi(t))=M_0\Phi_0'(\varphi_0(t))$ and the cosmological
equation \eqref{eomsingle} of the model $(M_0,\Phi)$ is equivalent
with:
\be
\ddot{\varphi}_0(t)+\cH_0(\varphi_0(t),\dot{\varphi}_0(t))\dot{\varphi}_0(t)+ \Phi'_0(\varphi_0(t))=0~~,
\ee
which is the cosmological equation of the model $(1,\Phi_0)$. In
particular, the positive homothety classes of the cosmological curves
of the model $(M_0,\Phi)$ depend only the reduced scalar potential
$\Phi_0$. Accordingly, one can set $M_0=1$ without loss of
generality. Every orbit of the action $\rho_\param$ intersects the
subset:
\be
\Par_1\eqdef \{(M_0,\Phi)\in \Par\vert M_0=1\}=\{1\}\times \Pot(\R)\simeq \Pot(\R)~~.
\ee
The latter is stabilized by the scale subgroup:
\be
\rT_1\eqdef \{(\lambda,\epsilon)\in \rT~\vert~\lambda=1\}=\{1\}\times \R_{>0}\simeq \R_{>0}
\ee
of $\rT$, whose action on $\Par_1$ identifies with the action of $\R_{>0}$ given by: 
\be
\Phi_0\rightarrow \frac{1}{\epsilon^2}\Phi_0~~\forall \epsilon>0~~\forall \Phi_0\in \Pot(\R)~~.
\ee
Hence the quotient $\Par/\rT$ is in bijection with the quotient of
$\Pot(\R)$ through this action, which coincides with the set:
\be
\Proj_+(\Pot(\R))\eqdef \Pot(\R)/\R_{>0}
\ee
of positive homothety classes of smooth positive functions defined on
$\R$. In particular, the overall scale of the scalar potential can
also be eliminated by performing a scale transform of $\varphi$.
\end{remark}

The $\rho_0(\lambda,\epsilon)$-transform of a curve $\varphi\in \C(\R)$
induces the transformation $\dot{\varphi}(t)\rightarrow
\frac{\lambda^{1/2}}{\epsilon} \dot{\varphi}(t/\epsilon)$ of the tangent
vector field to that curve. Since the pointed maximal cosmological
curve $\varphi_{(x,v)}:I(x,v)\rightarrow \R$ which satisfies
$\varphi_{x,v}(0)=x$ and $\dot{\varphi}_{x,v}(0)=v$ is uniquely
determined by the vector $(x,v)\in \R^2$, we have:
\be
\rho_0(\lambda,\epsilon)(\varphi_{(x,v)}) =\varphi_{\lambda^{1/2} x,\lambda^{1/2} \epsilon^{-1} v}~~\forall (x,v)\in \R^2~~\forall (\lambda,\epsilon)\in \rT~~.
\ee
In particular, the set $\Sol_m$ of pointed maximal cosmological curves
of the model $(M_0,\Phi)$ identifies with $\cT=\R^2$ and the
$\rho_0(\lambda,\epsilon)$-transform of such curves identifies with
the transformation:
\be
(x,v)\rightarrow (\lambda^{1/2}x,\frac{\lambda^{1/2}}{\epsilon} v)~~ 
\ee
of $\R^2$. Accordingly, we define: 

\begin{definition}
The {\em phase space scaling action} is the action $\rho_s$ of $\rT$ on
the phase space $\cT=\R^2$ defined through:
\ben
\label{rhos}
{\hat \rho}_0(\lambda,\epsilon)(x,v)\eqdef (\lambda^{1/2}x,\frac{\lambda^{1/2}}{\epsilon} v)~~\forall (\lambda,\epsilon)\in \R^2~~\forall (x,v)\in \R^2~~.
\een
\end{definition}

\noindent The remarks above are encoded by the following:

\begin{prop}
Let $\varphi:I\rightarrow \R$ be a cosmological curve,
$\gamma_\varphi\eqdef c(\varphi):I\rightarrow \R^2$ the corresponding
cosmological flow curve and $\cO_\varphi\eqdef \gamma_\varphi(I)$ its
cosmological flow orbit. Then the following relations hold for all
$(\lambda,\epsilon)\in \rT$:
\begin{itemize}
\item $\gamma_{\rho_0(\lambda,\epsilon)(\varphi)}={\hat \rho}_0(\lambda,\epsilon)\circ (\gamma_\varphi)_{\epsilon}$~.
\item $\cO_{\rho_0(\lambda,\epsilon)(\varphi)}={\hat \rho}_0(\lambda,\epsilon)(\cO_\varphi)$~.
\end{itemize}
\end{prop}

The restriction of $\rho(\lambda,\epsilon)$ to the
subset $\Par\times \Sol_m\simeq \Par\times \R^2$ 
identifies with the map:
\ben
\label{sim}
{\hat \rho}\eqdef \rho_\param\times {\hat \rho}_0:
\Par\times \R^2\rightarrow \Par\times \R^2~~,~~(M_0,\Phi,x,v)\rightarrow (\lambda^{1/2} M_0,\frac{\lambda}{\epsilon^2}\Phi_{\lambda^{1/2}},\lambda^{1/2}x, \frac{\lambda^{1/2}}{\epsilon} v)~~,
\een
which we call the {\em phase space similarity action}. Following \cite{ren}, we define: 

\begin{definition}
The {\em cosmological renormalization group} $\rT_\ren$ is defined through:
\be
\rT_\ren\eqdef \{(\lambda,\epsilon)\in \rT~\vert~\lambda=\epsilon^2\}\simeq \R_{>0}~~.
\ee
\end{definition}

\begin{definition}
The {\em similarity action of $\rT_\ren$} is the restriction of $\rho$
to $\rT_\ren$, which is given by:
\be
\rho_\ren(\epsilon)(M_0,\Phi,\varphi)=(\epsilon M_0,\Phi_\epsilon,\epsilon \varphi_\epsilon)~~\forall \epsilon>0~~\forall
(M_0,\Phi,\varphi)\in \Par\times\C(\R)~~.
\ee
The {\em parameter action} of $\rT_\ren$ is the restriction of
$\rho_\param$ to $\rT_\ren$, which is given by:
\be
\rho_\ren^\param(\epsilon)(M_0,\Phi)=(\epsilon M_0,\Phi_\epsilon)~~\forall \epsilon >0~~\forall
(M_0,\Phi)\in \Par~~.
\ee
The {\em curve action} of $\rT_\ren$ is the restriction of
$\rho_0$ to $\rT_\ren$, which is given by:
\be
\rho_\ren^0(\epsilon)(\varphi)=\epsilon\varphi_\epsilon~~\forall \epsilon >0~~\forall
\varphi \in \C(\R)~~.
\ee
\end{definition}

\noindent Proposition \ref{prop:sim} implies:
\be
\rho_\ren^0(\epsilon)(\Sol(M_0,\Phi))=\Sol(\rho_\ren^\param(M_0,\Phi))~~\forall \epsilon >0~~\forall
(M_0,\Phi)\in \Par~~.
\ee
The restriction of $\rho_\ren$ to $\Par\times \Sol_m$ identifies with
the action of $\rT_\ren$ on $\Par\times \R^2$ given by the corresponding
restriction of ${\hat \rho}$:
\be
{\hat \rho}_\ren(\epsilon)(M_0,\Phi,x,v)=(\epsilon M_0,\Phi_\epsilon, \epsilon x, v)~~\forall \epsilon>0~~\forall (M_0,\Phi,x,v)\in \Par\times \cT~~.
\ee

\begin{definition}
The {\em phase space renormalization group action} is the
action of $\rT_\ren$ on $\cT=\R^2$ obtained by restricting \eqref{rhos}:
\be
{\hat \rho}_\ren^0(\epsilon)(x,v)=(\epsilon x, v)~~\forall \epsilon>0~~\forall (x,v)\in \cT~~.
\ee
\end{definition}

\begin{remark}
Notice that the rescaled Planck mass can also be eliminated by
performing a renormalization group transformation with parameter
$\epsilon=\frac{1}{M_0}$, since:
\be
\rho_\ren(1/M_0)(M_0,\Phi,\varphi)=(1,\Phi_{1/M_0},\frac{1}{M_0}\varphi_{1/M_0})~~.
\ee
\end{remark}

\section{The Hamilton-Jacobi formalism}
\label{sec:HJ}

\noindent In this section, we discuss the Hamilton-Jacobi equation of
regular cosmological flow orbits, which is equivalent with the regular
flow orbit equation \eqref{floworbiteq} of Subsection
\ref{subsec:reg}. We also discuss the action of the universal
similarity group $\rT$ on Hamilton-Jacobi functions.

\subsection{The Hamilton-Jacobi transformation}

\noindent By definition, the {\em Hamilton-Jacobi half-plane} is the
open upper half-plane $\R\times \R_{>0}$ endowed with Cartesian
coordinates $(x,h)$, where $x\in \R$ and $h\in \R_{>0}$. The restriction
of the rescaled Hubble function $\cH$ to the zero section $Z$ of $\cT=T\R$
gives the positive smooth function:
\ben
\label{fH0}
\fH_0\eqdef \frac{\sqrt{2\Phi}}{M_0}:\R\rightarrow \R_{>0}~~.
\een

\begin{definition}
The {\em Hamilton-Jacobi region} of the model $(M_0,\Phi)$ is the
closed overgraph of the function $\fH_0$ inside the
Hamilton-Jacobi half-plane:
\be
\cU\eqdef \{(x,h)\in \R\times \R_{> 0}\vert h\geq \frac{\sqrt{2\Phi(x)}}{M_0}\}\subset \R\times \R_{>0}~~.
\ee
The {\em Hamilton-Jacobi locus} of the model is the frontier
of the Hamilton-Jacobi region:
\be
\cB\eqdef\pd \cU=\graph(\fH_0)=\Big\{\Big(x,\frac{\sqrt{2\Phi(x)}}{M_0}\Big)\bvert x\in \R \Big\}\subset \R\times \R_{> 0}~~.
\ee
\end{definition}

\begin{definition}
The {\em Hamilton-Jacobi transformation} of the model $(M_0,\Phi)$ is the
surjective map $\h:\cT\rightarrow \cU$ defined through:
\be
\h(x,v)=(x,\cH(x,v))~~\forall (x,v)\in \cT=\R^2~~,
\ee
where $\cH$ is the rescaled Hubble function of the model.
\end{definition}

\noindent The Hamilton-Jacobi transformation is a double cover of the
Hamilton-Jacobi region $\cU$ which is ramified above the
Hamilton-Jacobi locus and whose ramification set is the image $Z$ of
the zero section of $T\R$:
\be
\h^{-1}(\cB)=Z=\R\times \{0\}\subset \cT=\R^2~~.
\ee
The {\em positive sheet} of this cover is the open upper half-plane
$\cT^+$ while its {\em negative sheet} is the open lower half-plane
$\cT^-$. The closures of the two sheets meet along $Z$. Since
$\cH(x,v)=\frac{1}{M_0}\sqrt{v^2+2\Phi(x)}$, the $\h$-preimage of a
point $(x,h)\in \Int \cU$ is given by:
\be
\h^{-1}(\{(x,h)\})=\{(x,-\sqrt{M_0^2 h^2-2\Phi(x)}),(x,+\sqrt{M_0^2 h^2-2\Phi(x)})\}
\ee
and hence consists of two points of the phase space $\cT=\R^2$ which
are symmetric with respect to the $x$ axis $Z$. The preimage reduces
to a point which lies on the $x$ axis when $(x,h)\in \cB$. The Galois
group of the cover is $\Z_2$, whose action by $(x,v)\rightarrow
(x,-v)$ on $\cT$ permutes the two sheets. The fixed point set of this
action coincides with the ramification set $Z$. The restriction of
$\h$ to any of the sheets $\cT^\pm$ is a diffeomorphism between
that sheet and $\Int \cU$ and hence gives coordinates $(x,h)\in
\Int\cU$ defined on that sheet.

\begin{definition}
The coordinates $(x,h)\in \Int\cU$ defined by $\h$ on each of the
sheets $\cT^+$ and $\cT^-$ are called {\em phase space Hamilton-Jacobi
  coordinates}.
\end{definition}

Given a non-trivial cosmological flow orbit $\cO\subset \R^2$, its
image through $\h$ is a one-dimensional submanifold of the
Hamilton-Jacobi half-plane which is contained in the Hamilton-Jacobi
region $\cU$. We denote this image by:
\be
\cV_\cO=\h(\cO)\subset \cU
\ee
and call it the {\em Hamilton-Jacobi manifold} of $\cO$. The
Hamilton-Jacobi manifold is called {\em regular} if $\cV_\cO\cap
\cB=\emptyset$, which happens iff the orbit $\cO$ is regular
(i.e. contained in one of the open sheets $\cT^+$ or $\cT^-$).
The Hamilton-Jacobi manifold of an irregular flow orbit $\cO$
meets the Hamilton-Jacobi locus along the discrete set:
\be
\cV_\cO\cap \cB=\h(\cO \cap Z)~~.
\ee
The connected components of $\cO\setminus Z$ are the {\em regular
  components} of $\cO$. The connected components of
$\cV_\cO \setminus \cB$ are regular Hamilton-Jacobi manifolds. As we
shall see below, regular Hamilton-Jacobi manifolds coincide with
the graphs of regular {\em Hamilton-Jacobi functions}.

\begin{definition}
Let $\cO\subset \cT\setminus Z$ be a connected and regular
one-dimensional submanifold of the phase space which is the graph of a
smooth function $\v_\cO$. The {\em H-function} of $\cO$ is the smooth
function $\fH_\cO:\pi(\cO)\rightarrow \R$ defined though:
\be
\fH_\cO=\h\circ \v_\cO~~\mathrm{i.e.}~~\fH_\cO\eqdef \cH(x,\v_\cO(x))~~\forall x\in \pi(\cO)~~.
\ee
\end{definition}

\begin{definition}
The $H$-function of a regular curve $\varphi:I\rightarrow \R$ is the
$H$-function:
\be
\fH_\varphi\eqdef \fH_{\cO_\varphi}
\ee
of the phase space manifold $\cO_\varphi$ of $\varphi$.
\end{definition}

Let $\varphi:I\rightarrow \R$ be a regular curve of signature
$\zeta_\varphi$ and image $\varphi(I)=J$. Since the map
$\varphi:I\rightarrow J$ is a diffeomorphism, we can use $x\in J$ as a
parameter for $\varphi$. As in Subsection \ref{subsec:reg}, set:
\be
\t_\varphi\eqdef \varphi^{-1}:J\rightarrow I
\ee
and let $v_\varphi\eqdef \dot{\varphi}:I\rightarrow \R$ and
$\v_\varphi\eqdef v_\varphi\circ \t_\varphi:J\rightarrow \R^\times$.
Then the $H$-function of $\varphi$ is given by:
\ben
\label{fH}
\fH_\varphi(x)=\cH(x,\v_\varphi(x))=\frac{1}{M_0}\left[\v_\varphi(x)^2+2\Phi(x)\right]^{1/2}=\frac{1}{M_0}\left[\frac{1}{\t'_\varphi(x)^2}+2\Phi(x)\right]^{1/2}~~\forall x\in J~~,
\een
i.e.:
\be
\fH_\varphi\eqdef \cH_\varphi\circ \t_\varphi:J\rightarrow \R~~,
\ee
where $\cH_\varphi:I\rightarrow \R$ is the rescaled Hubble parameter
of $\varphi$:
\be
\cH_\varphi(t)= \frac{1}{M_0}\sqrt{\dot{\varphi}(t)^2+2\Phi(\varphi(t))}~~\forall t\in I~~.
\ee
We have $\sign(\v_\varphi)=\sign(\t'_\varphi)=\zeta_\varphi$.

\begin{remark}
\label{rem:hphi}
The functions $\Phi$ and $\cH$ provide local coordinates on each
connected component of the open subset:
\be
\cN\eqdef (\Noncrit\Phi)\times \R^\times\subset \R^2\setminus Z
\ee
of the phase space. Writing $\Noncrit\Phi=\R\setminus \Crit\Phi$ as an
at most countable disjoint union of intervals:
\be
\Noncrit\Phi=\sqcup_j K_j~~,
\ee
these connected components are $K_j\times \R_\pm$, where $\R_+\eqdef
\R_{>0}$ and $\R_{-}\eqdef \R_{<0}$. If $U_{j,\zeta}\eqdef K_j\cup
\R_\zeta$ with $\zeta\in \{-1,1\}$ is such a connected component, then
the map $(\Phi,\cH):U_{j,\zeta}\rightarrow S_j$ is bijective,
where:
\be
S_j=\{(\phi,h)\in F_j\times \R_{> 0}\vert h\geq \frac{\sqrt{2\phi}}{M_0}\}
\ee
is the intersection of the region $h\geq \frac{\sqrt{2\phi}}{M_0}$ of
the $(\phi,h)$ plane with a vertical strip through the interval
$F_j\eqdef \Phi(K_j)\subset \R$ located along the $\phi$-axis. This
map has inverse given by:
\be
x=(\Phi\vert_{K_j})^{-1}(\phi)~~,~~v=\zeta \sqrt{M_0^2 h^2-2\phi}~~.
\ee
In particular, one can express the restriction of any basic
cosmological observable to a connected component of $\cN$ as a
function of $\Phi$ and $\cH$.  The derivative of $\Phi$ can be
expressed as a function of $\Phi$, which we denote by $W_j$:
\be
W_j\eqdef \Phi'\circ (\Phi\vert_{K_j})^{-1}:K_j\rightarrow \R^\times
\ee
and we have $\frac{\dd \phi}{\dd x}=W_j(\phi)$ for $\phi\in F_j$.
\end{remark}

\subsection{The Hamilton-Jacobi equation and Hamilton-Jacobi functions}
\label{sec:HJfunctions}

\begin{definition}
The {\em Hamilton-Jacobi equation} of the model $(M_0,\Phi)$ is the
following first order nonlinear and non-autonomous ODE for the
real-valued function $\fH$:
\ben
\label{HJ}
\fH(x)= \sqrt{M_0^2\fH'(x)^2+\frac{2\Phi(x)}{M_0^2}}~~.
\een
A {\em Hamilton-Jacobi function} of the model is a function
$\fH:J\rightarrow \R_{>0}$ of class $C^1$ (where $J$ is a
non-degenerate interval) which satisfies the Hamilton-Jacobi equation.
\end{definition}

\noindent Notice that any Hamilton-Jacobi function $\fH$ satisfies the
inequality $\fH(x)\geq \frac{\sqrt{2\Phi(x)}}{2M_0}$ and hence its
graph is contained in the Hamilton-Jacobi region of the model:
\be
\graph(\fH)\subset \cU~~.
\ee
Moreover, we have $\fH'(x)=0$ iff $(x,\fH(x))\in \pd \cU=\cB$. Define:
\be
J_\sing\eqdef \{x\in J~\vert~\fH'(x)=0\}=\Big\{x\in J~\vert~\fH(x)=\frac{\sqrt{2\Phi(x)}}{2M_0}\Big\}
\ee
and $J_\reg\eqdef J\setminus J_\sing$. Then the usual bootstrap
argument shows that $\fH$ is piecewise-smooth, namely its restriction
to the open set $J_\reg$ is smooth. Notice that $\fH'$ may
have different signs on the various connected components of $J_\reg$.

\begin{remark}
\label{rem:hatHJ}
Let ${\hat \fH}\eqdef M_0\fH$. Then the Hamilton-Jacobi equation
\eqref{HJ} is equivalent with the following equation for ${\hat \fH}$:
\ben
\label{hatHJ}
{\hat \fH}(x)=\sqrt{M_0^2 {\hat \fH}'(x)^2+2\Phi(x)}~~,
\een
which has the advantage that the parameter $M_0$ appears only in the
term involving the first derivative of ${\hat \fH}$. One can seek
solutions of \eqref{hatHJ} which expand as a power series in $M_0^2$:
\be
{\hat \fH}(x)=\sum_{n\geq 0}{\hat \fH}_n(x) M_0^{2n}~~,
\ee
which amounts to seeking solutions of \eqref{HJ} which expand as a
Laurent series of the form:
\be
\fH(x)=\frac{1}{M_0}\sum_{n\geq 0} \fH_n(x)M_0^{2n}~~.
\ee
We will see below (see Remark \ref{rem:speedHsol}) that this is
equivalent with the rescaled Planck mass expansion of cosmological
curves mentioned in Remark \ref{rem:psi}.
\end{remark}

\begin{definition}
A Hamilton-Jacobi function $\fH:J\rightarrow \R$ is called {\em
regular} if it satisfies the condition $\fH>2\Phi$ (i.e. $\fH'\neq 0$)
on its interval of definition. In this case, the sign $\xi_\fH\eqdef
\sign(\fH'(x))\in \{-1,1\}$ is constant on $J$ and is called the {\em
signature} of $\fH$.
\end{definition}

\noindent A Hamilton-Jacobi function $\fH:J\rightarrow \R$ is regular
iff $\graph(\fH)\cap \cB=\emptyset$. A regular Hamilton-Jacobi function
of signature $\xi$ satisfies the {\em $\xi$-signed Hamilton-Jacobi equation}:
\ben
\label{HJsigned}
\fH'(x)=\frac{\xi}{M_0} \sqrt{\fH(x)^2-\frac{2\Phi(x)}{M_0^2}}~~\forall x\in J~~.
\een

The following proposition gives a bijection between regular
Hamilton-Jacobi functions and regular cosmological flow orbits. The
bijection is obtained by writing the regular flow orbit equation
\eqref{floworbiteq} in terms of Hamilton-Jacobi coordinates on that
sheet of the Hamilton-Jacobi map which contains a given regular flow
orbit. This presents the $\h$-image of that flow orbit as the graph of
the corresponding regular Hamilton-Jacobi function.

\begin{prop}
\label{prop:regflow}
The regular flow orbit equation \eqref{floworbiteq} for regular flow
orbits of signature $\zeta$ is equivalent with the signed
Hamilton-Jacobi equation \eqref{HJsigned} with $\xi=-\zeta$. More
precisely, a connected and regular one-dimensional submanifold
$\cO\subset \cT^\zeta$ signature $\zeta$ which is the graph of a
smooth function is a regular cosmological flow orbit iff its
$H$-function $\fH_\cO:J\eqdef \pi(\cO)\rightarrow \R$ is a regular
Hamilton-Jacobi function of signature $\xi=-\zeta$. In this case,
we have:
\be
\h(\cO)=\graph(\fH_\cO)=\{(x,\fH_\cO(x))~\vert~x\in J\}~~.
\ee
Conversely, every regular Hamilton-Jacobi function $\fH:J\rightarrow
\cU$ of signature $\xi$ is the $H$-function of a regular cosmological
flow orbit $\cO_\fH$ of signature $\zeta=-\xi$ which is defined by
the {\em speed equation of $\fH$}:
\ben
\label{speq}
v=-M_0^2\fH'(x)\Longleftrightarrow v=-\sign(\fH'(x)) \sqrt{M_0^2\fH(x)^2-2\Phi(x)}~~\forall x\in J~~.
\een
In particular, the regular Hamilton-Jacobi manifolds of the model
$(M_0,\Phi)$ coincide with the graphs of its regular Hamilton-Jacobi
functions.
\end{prop}

\noindent We stress that the flow orbit $\cO$ and its associated
Hamilton-Jacobi function $\fH$ have opposite signatures. The
proposition shows that the speed function $\v$ of a regular flow orbit
can be recovered from the Hamilton-Jacobi function $\fH$ of that orbit
through the formula:
\be
\v(x)=-M_0^2\fH'(x)=-\sign(\fH'(x)) \sqrt{M_0^2\fH(x)^2-2\Phi(x)}
\ee
while $\fH$ can be recovered from $\v$ as:
\be
\fH(x)=\cH(x,\v(x))~~.
\ee

\begin{proof}
Since $\cO$ is regular, we have $\v_\cO(x)\neq 0$ for all $x\in J$. The
regular flow orbit equation \eqref{floworbiteq} is equivalent with:
\ben
\label{HJequiv}
\frac{\dd}{\dd x} [\v(x)^2+2\Phi(x)]=-2\v(x) \cH(x,\v(x))\Longleftrightarrow \frac{\dd}{\dd x} \cH(x,\v(x))=-\frac{\v(x)}{M_0^2}~~,
\een
where we used the relation $\v(x)^2+2\Phi(x)=M_0^2 \cH(x,\v(x))^2$ and
the fact that $\cH>0$. Since:
\be
\v(x)=\zeta \sqrt{M_0^2\cH(x,\v(x))^2-2\Phi(x)}~~,
\ee
the second equation in \eqref{HJequiv} amounts to:
\be
\frac{\dd}{\dd x} \cH(x,\v(x))=-\frac{\zeta}{M_0^2} \sqrt{M_0^2\cH(x,\v(x))^2-2\Phi(x)}~~,
\ee
which is equivalent with \eqref{HJsigned} with $\xi=-\zeta$ and
$\fH(x)=\cH(x,\v(x))$ for all $x\in J$. Since $\v\neq 0$, we have
$\fH>\frac{\sqrt{2\Phi}}{M_0}$ on $J$ and hence $\fH$ is a regular
Hamilton-Jacobi function. Fixing a regular flow orbit $\cO$ amounts to
choosing a solution $\v$ of the regular flow orbit equation, which by
the above amounts to choosing a regular Hamilton-Jacobi function
$\fH(x)\eqdef \cH(x,\v(x))$. The remaining statements of the
proposition are obvious.
\end{proof}

\begin{remark}
The signed Hamilton-Jacobi equation for a regular orbit of signature
$\zeta$ which is contained in the open subset $U_{j,\zeta}=K_j\times
\R_\zeta$ of the phase space can be written as follows:
\be
\frac{\dd h}{\dd \phi}=-\frac{\zeta}{M_0} \frac{\sqrt{h^2-\frac{2\phi}{M_0^2}}}{W_j(\phi)}~~(\phi \in F_j)~~,
\ee
where $W_j$ was defined in Remark \ref{rem:hphi}.
\end{remark}

\subsection{Reconstruction of regular cosmological curves}

\noindent Let $\varphi:I\rightarrow \R$ be a regular smooth
curve of signature $\zeta_\varphi$ and image $\varphi(I)=J$. Since
the map $\varphi:I\rightarrow J$ is a diffeomorphism, we can use $x\in
J$ as a parameter for $\varphi$. As above, set:
\be
\t_\varphi\eqdef \varphi^{-1}:J\rightarrow I
\ee
and let $\v_\varphi\eqdef \dot{\varphi}\circ \t_\varphi:J\rightarrow
\R^\times$ be the speed function of $\varphi$.

\begin{prop}
\label{prop:varphirec}
Let $\varphi:I\rightarrow \R$ be a regular curve of signature
$\zeta_\varphi\eqdef \sign(\dot{\varphi})$ and image $J\eqdef
\varphi(I)$. Then $\varphi$ is a cosmological curve of the model
$(M_0,\Phi)$ iff one (and hence both) of the following equivalent
conditions is satisfied:
\begin{enumerate}[(a)]
\item $\fH_\varphi$ is a regular Hamilton-Jacobi function of signature
$\sign(\fH_\varphi)=-\zeta_\varphi$, i.e. it satisfies the signed
Hamilton-Jacobi equation:
\ben
\label{HJvarphi}
\fH_\varphi'(x)=-\frac{\zeta_\varphi}{M_0^2} \sqrt{M_0^2\fH_\varphi(x)^2-2\Phi(x)}~~\forall x\in J~~.
\een
\item $\v_\varphi$ satisfies the speed equation of $\fH_\varphi$:
\ben
\label{speed}
\v_\varphi(x)=-M_0^2\fH'_\varphi(x)~~.
\een
\end{enumerate}
In this case, $\fH_\varphi$ coincides with the Hamilton-Jacobi
function $\fH_\cO$ of the cosmological flow orbit $\cO$ determined by
$\varphi$.
\end{prop}

\begin{remark}
\label{rem:speedHsol}
Notice that the speed equation of $\fH_\varphi$ can be written as:
\be  
\t'_\varphi(x)=-\frac{1}{M_0^2\fH'_\varphi(x)}~~\forall x\in J~~,
\ee
with general solution:
\ben
\label{tHJ}
\t_\varphi(x)=C-\frac{1}{M_0^2}\int_{x_0}^x \frac{\dd y}{\fH_\varphi'(y)}=C+\zeta_\varphi\int_{x_0}^x \frac{\dd y}{\sqrt{M_0^2\fH_\varphi(y)^2-2\Phi(y)}}~~,
\een
where $C\in \R$ and $x_0\in J$ are arbitrary. This relation can be
inverted to find the equation $x=\varphi(t)$ of $\varphi$. In
particular $\t_\varphi$ and hence $\varphi$ depend non-locally on
$\fH_\varphi$. Also notice that the scale transform $\psi\eqdef \varphi_{M_0}$
of Remark \ref{rem:psi} has speed function:
\ben
\label{speedpsi}
\v_\psi=\frac{1}{M_0}\v_\varphi=-M_0\fH'_\varphi(x)=-{\hat \fH}'_\varphi(x)~~
\een
and hence is obtained by inverting the function:
\be
\t_\psi(x)=M_0 C'-\int_{x_0}^x \frac{\dd y}{{\hat \fH}_\varphi'(y)}=M_0\left(C'+\zeta_\varphi\int_{x_0}^x \frac{\dd y}{\sqrt{{\hat \fH_\varphi(y)}^2-2\Phi(y)}}\right)~~.
\ee
where $C=M_0 C'$. Relation \eqref{speedpsi} shows that the expansion
of $\psi$ in powers of $M_0$ mentioned in Remark \ref{rem:psi}
corresponds to the power series expansion of ${\hat \fH}_\varphi$ in
$M_0$ and hence to the Laurent expansion of $\fH_\varphi$ in $M_0$. 
\end{remark}

\begin{proof}
Since $\varphi$ is a regular curve, the cosmological equation for
$\varphi$ is equivalent with the dissipation equation \eqref{DissipEq}
by Proposition \ref{prop:regdissip}. We have:
\ben
\label{dotcH}
\dot{\cH}_\varphi=(\fH_\varphi'\circ \varphi)\dot{\varphi}~~.
\een
Hence the dissipation equation \eqref{DissipEq} is equivalent
with:
\ben
\label{dotvarphi}
\dot{\varphi}^2=-M_0^2 \dot{\varphi}\fH_\varphi'\circ \varphi\Longleftrightarrow \dot{\varphi}=-M_0^2\fH_\varphi'\circ \varphi~~,
\een
which is the speed equation \eqref{speed}. On the other hand, relation
\eqref{fH} is equivalent with:
\be
\v_\varphi=\zeta_\varphi\sqrt{M_0^2\fH_\varphi^2-2\Phi}~~,
\ee
where we used the fact that
$\sign(\v_\varphi)=\sign(\dot{\varphi})=\zeta_\varphi$. This shows
that the speed equation \eqref{speed} of $\fH_\varphi$ is equivalent
with the signed Hamilton-Jacobi equation \eqref{HJvarphi}. Since
$\dot{\cH}_\varphi<0$, relation \eqref{dotcH} gives
$\sign(\fH_\varphi')=-\sign(\dot{\varphi})=-\zeta_\varphi$ (this also
follows from the speed equation \eqref{speed}). It is clear from
\eqref{fH} that $\fH_\varphi$ coincides with the Hamilton-Jacobi
function of the cosmological flow orbit $\cO$ determined by $\varphi$.
\end{proof}

\subsection{The Hamilton-Jacobi potential of a regular cosmological flow orbit}
\label{subsec:HJpot}

\begin{definition}
The {\em Hamilton-Jacobi potential} of a regular cosmological flow
orbit $\cO$ of the model $(M_0,\Phi)$ is the function $V_\cO:J\eqdef
\pi(\cO)\rightarrow \cU$ defined through:
\be
V_\cO(x)\eqdef M_0^2\fH_\cO(x)~~\forall x\in J~~,
\ee
where $\fH_\cO$ is the Hamilton-Jacobi function of $\cO$.
\end{definition}

\noindent Let $\varphi:I\rightarrow \R$ be a regular cosmological
curve with cosmological flow orbit $\cO$. Then the
cosmological orbit of $\varphi$ is the interval $J=\pi(\cO)\subset
\R$. Since $x=\varphi(t)$ and $v_\varphi(t)=\dot{\varphi}(t)$, the
speed equation \eqref{speed} of $\fH_\varphi$ is equivalent with the
gradient flow equation:
\ben
\label{gf}
\dot{\varphi}(t)=-V'_\cO(\varphi(t))~~
\een
of $V_\cO$. This equation determines $\varphi$ up to an arbitrary
constant translation of $t$ when the cosmological flow orbit $\cO$ is
fixed.  Hence the problem of computing the regular cosmological curves
of a given model can be approached in two steps:
\begin{enumerate}
\item Solve the Hamilton-Jacobi equation to find the regular
Hamilton-Jacobi functions, which determine the regular cosmological
flow orbits $\cO$ of the model.
\item For each regular cosmological flow orbit (i.e. for each regular
Hamilton-Jacobi function), solve the one-dimensional gradient flow
equation \eqref{gf} to determine the corresponding regular
cosmological curve $\varphi$ up to translation of $t$ by an arbitrary
constant. As shown in \eqref{tHJ}, the solutions of this equation are
obtained by inverting the function:
\ben
\label{tHJ}
\t_\varphi(x)=C-\int_{x_0}^x \frac{\dd y}{V'_\cO(y)}~~,
\een
where $C\in \R$ and $x_0\in J$ are arbitrary.
\end{enumerate}

\begin{remark}
The {\em first slow roll approximation} for a regular cosmological flow orbit
$\cO$ corresponds to taking:
\be
\frac{M_0^2|\fH'_\cO(x)|}{\sqrt{2\Phi(x)}}\ll 1~~\Longleftrightarrow~~\fH_\cO\approx \fH_0~~,
\ee where $\fH_0\eqdef \frac{\sqrt{2\Phi}}{M_0}$ is independent of
$\cO$. This approximation is accurate for a cosmological flow
orbit $\cO$ which is very close to the zero set $Z$ (and hence the
corresponding portion of the Hamilton-Jacobi manifold $\cV_\cO$ is
very close to the Hamilton-Jacobi locus $\cB$). In this case, $V_\cO$
reduces to the {\em leading classical effective potential}:
\ben
\label{fV0}
V_0\eqdef M_0^2\fH_0=M_0\sqrt{2\Phi}
\een
of \cite{ren} while $\fH_\cO$ reduces to $\fH_0$. Notice that $V_0$
and $\fH_0$ depend only on the parameters $(M_0,\Phi)$ of the model,
being independent of the orbit $\cO$. Also notice that $\fH_0$ is a
Hamilton-Jacobi function only when restricted to a connected component
of the interior of the critical locus $\Crit\Phi$ (assuming that
$\Int(\Crit\Phi)$ is non-empty), in which case it corresponds to a
trivial cosmological flow orbit.  In the next section, we explain how
one can construct a better approximant for cosmological orbits which
are close to the zero section $Z$ of $T\R$.
\end{remark}

\subsection{Similarities of Hamilton-Jacobi functions}
\label{subsec:HJsim}

\noindent Under a phase space scaling action \eqref{rhos} with
parameters $(\lambda,\epsilon)$, a regular cosmological flow orbit
$\cO$ with implicit equation $v=\v(x)$ of the model $(M_0,\Phi)$ maps
to the regular flow orbit ${\tilde \cO}$ with equation $v={\tilde \v}(x)$
of the model $({\tilde M}_0,{\tilde
  \Phi})=(\lambda^{1/2}M_0,\lambda\epsilon^{-2}\Phi_{\lambda^{1/2}})$,
where:
\ben
\label{v'}
{\tilde \v}(x)=\frac{\lambda^{1/2}}{\epsilon} \v_{\lambda^{1/2}}(x)=\frac{\lambda^{1/2}}{\epsilon} \v(\lambda^{-1/2} x)~~.
\een
Accordingly, the Hamilton-Jacobi function
$\fH(x)=\frac{1}{M_0}\sqrt{\v(x)^2+2\Phi(x)}$ of $\cO$ maps to
the Hamilton-Jacobi function of ${\tilde \cO}$, which is given by:
\be
{\tilde \fH}(x)=\frac{1}{{\tilde M}_0}\sqrt{{\tilde \v}(x)^2+2{\tilde \Phi}(x)}=\frac{1}{\epsilon}\fH_{\lambda^{1/2}}(x)=\frac{1}{\epsilon} \fH(\lambda^{-1/2}x)~~.
\ee
Setting $\lambda=\epsilon^2$ gives the action of the renormalization
group $\rT_\ren$ on Hamilton-Jacobi functions:
\be
\fH\rightarrow \frac{1}{\epsilon}\fH_\epsilon~~.
\ee
Accordingly, we define: 

\begin{definition}
\label{def:rdef}
The {\em Hamilton-Jacobi scaling action} is the action $\r_0$ of $\rT$
on $\C(\R)$ given by:
\ben
\label{HJscaling}
\r_0(\lambda,\epsilon)(\fH)\eqdef \frac{1}{\epsilon}\fH_{\lambda^{1/2}}~~\forall \fH\in \C(\R)~~\forall (\lambda,\epsilon)\in \rT~~.
\een
The {\em Hamilton-Jacobi similarity action} is the action
$\r=\rho_\param\times \r_0$ of $\rT$ on $\Par\times \C(\R)$:
\be
\r(\lambda,\epsilon)(M_0,\Phi,\fH)\eqdef (\lambda^{1/2} M_0,\frac{\lambda}{\epsilon^2}\Phi_{\lambda^{1/2}},
\frac{1}{\epsilon}\fH_{\lambda^{1/2}})~~\forall (M_0,\Phi,\fH)\in \Par\times \C(\R)~~\forall (\lambda,\epsilon)\in \rT~~.
\ee
\end{definition}

\noindent Let $\H(M_0,\Phi)$ and $\H_\reg(M_0,\Phi)$ be the sets of
all (respectively regular) Hamilton-Jacobi functions of the model
$(M_0,\Phi)$ and $\H_\reg^\pm(M_0,\Phi)$ the subsets of regular
Hamilton-Jacobi functions of positive and negative signature. We have:
\be
\H_\reg(M_0,\Phi)=\H_\reg^+(M_0,\Phi)\sqcup \H_\reg^-(M_0,\Phi)\subset \H(M_0,\Phi)~~.
\ee

\begin{prop}
\label{prop:HJsim}
For any $(\lambda,\epsilon)\in \rT$ and $(M_0,\Phi)\in \Par$, we have:
\be
\r_0(\lambda,\epsilon)(\H(M_0,\Phi))=\H(\rho_\param(\lambda,\epsilon)(M_0,\Phi))~~,~~\r_0(\lambda,\epsilon)(\H^\pm_\reg(M_0,\Phi))=
\H_\reg^\pm(\rho_\param(\lambda,\epsilon)(M_0,\Phi))~~.
\ee
\end{prop}

\begin{proof}
It is clear that $\fH\in \C(\R)$ satisfies the Hamilton-Jacobi
equation \eqref{HJ} of the model $(M_0,\Phi)$ iff
$\r(\lambda,\epsilon)(\fH)$ satisfies the Hamilton-Jacobi equation of
the model $\rho_\param(\lambda,\epsilon)(M_0,\Phi)=(\lambda^{1/2}
M_0,\frac{\lambda}{\epsilon^2}\Phi_{\lambda^{1/2}})$. Moreover, the
derivative of $\fH$ is non-vanishing iff that of
$\r_0(\lambda,\epsilon)(\fH)$ is, in which case the two derivatives
have the same sign.
\end{proof}

\begin{definition}
The {\em Hamilton-Jacobi RG scaling action} is the restriction
$\r_\ren^0$ of $\r_0$ to $\rT_\ren$, which is given by:
\ben
\label{HJren}
\r_\ren^0(\epsilon)(\fH)\eqdef \frac{1}{\epsilon}\fH_{\epsilon}~~\forall \fH\in \C(\R)~~\forall \epsilon>0~~.
\een
The {\em Hamilton-Jacobi RG similarity action} is the restriction
$\r_\ren$ of $\r$ to $\rT_\ren$, which is given by:
\ben
\label{HJren}
\r_\ren(\epsilon)(M_0,\Phi,\fH)\eqdef (\epsilon M_0, \Phi_\epsilon,\frac{1}{\epsilon}\fH_{\epsilon})~~\forall (M_0,\Phi,\fH)\in \Par\times \C(\R)~~\forall \epsilon>0~~.
\een
\end{definition}

\noindent Proposition \ref{prop:HJsim} implies:
\ben
\label{rren}
\r_\ren(\epsilon)(\H(M_0,\Phi))=\H(\rho_\ren^\param(\epsilon)(M_0,\Phi))~~,~~\r_\ren(\epsilon)(\H^\pm_\reg(M_0,\Phi))=\H_\reg^\pm(\rho_\ren^\param(\epsilon)(M_0,\Phi))
\een
for all $\epsilon>0$ and all $(M_0,\Phi)\in \Par$. In particular, the
following subsets of $\Par\times \C(\R)$:
\be
\H\eqdef \sqcup_{(M_0,\Phi)\in \Par} \H(M_0,\Phi)~~\mathrm{and}~~\H_\reg^\pm\eqdef \sqcup_{(M_0,\Phi)\in \Par} \H_\reg^\pm(M_0,\Phi)
\ee
are invariant under the similarity action $\r_\ren$ of $\rT_\ren$.

\section{The IR Hamilton-Jacobi function and potential}
\label{sec:deformed}

In this section, we discuss the construction of a universal formal IR
Hamilton-Jacobi function $\fH_\IR$ as a Laurent power series in the
rescaled Planck mass $M_0$ and of the associated formal
Hamilton-Jacobi potential $V_\IR$, which admits a power series
expansion in $M_0$. For technical reasons, we start by constructing
slightly more general objects which also depend on the scale parameter
and reduce to $\fH_\IR$ and $V_\IR$ when this parameter is formally
set to one. We show that the coefficients of the corresponding
expansions are controlled by certain quasi-homogeneous polynomials in
the scalar potential and its derivatives, which are determined by a
recursion relation. Using the homogeneity properties of these
polynomials, we show that the resulting expansion can be written in
terms of the slow-roll parameters of \cite{LPB} and coincides with the
slow roll expansion of loc. cit. This provides a conceptually clear
explanation of the slow roll expansion as an expansion in powers of
the rescaled Planck mass. In particular, we give an explicit recursion
procedure for constructing the slow roll expansion, thus providing an
efficient method for computing its higher order terms. In Appendix
\ref{app:list}, we list the the coefficients of this expansion up to 10th
order, thus improving markedly on the results of loc. cit. 

\begin{definition}
For any $\epsilon>0$, the {\em $\epsilon$-deformed Hamilton-Jacobi
  equation} of the model $(M_0,\Phi)$ is the following first order ODE
for the unknown function $u:J\rightarrow \R_{>0}$, where $J\subset
\R$ is a non-degenerate interval:
\ben
\label{HJeps}
u(x)=\sqrt{\epsilon^2 M_0^2 u'(x)^2+\frac{2\Phi(x)}{M_0^2}}\Longleftrightarrow u(x)=\frac{\sqrt{2\Phi(x)}}{M_0}\left[1+\frac{\epsilon^2 M_0^4}{2\Phi(x)} u'(x)^2 \right]^{1/2}~~(x\in J)~~.
\een
\end{definition}

\noindent This reduces to the Hamilton-Jacobi equation of $(M_0,\Phi)$
when $\epsilon=1$ and to the algebraic equation
$u=\frac{\sqrt{2\Phi}}{M_0}$ for $\epsilon\rightarrow 0$. Notice that
the parameter $\epsilon$ ``counts the number of derivatives'' and that
\eqref{HJeps} is equivalent with the equation:
\ben
\label{Phieps}
\frac{2\Phi(x)}{M_0^2}=u(x)^2-\epsilon^2 M_0^2 u'(x)^2
\een
together with the condition $u>0$. Setting ${\hat u}=M_0 u$
(cf. Remark \ref{rem:hatHJ}), equation \eqref{HJeps} is equivalent
with:
\ben
\label{hat HJeps}
{\hat u}(x)=\sqrt{(\epsilon M_0)^2 {\hat u}'(x)+2\Phi(x)}~~,
\een
which is obtained from \eqref{hatHJ} by replacing $M_0$ with $\epsilon
M_0$. This gives the following result, which relates the solutions of
\eqref{HJeps} to constant rescalings of the Hamilton-Jacobi functions of a
different model.

\begin{prop}
\label{prop:hatHJepsilon}
Let $\mu>0$ and consider a family of smooth functions $u^{(M_0)}\in
\cC^\infty(J)$.  Then $u^{(M_0)}$ satisfies the $\epsilon$-deformed
Hamilton-Jacobi equation of $(M_0,\Phi)$ for all $M_0\in (0,\mu)$
iff:
\ben
\label{fHrecM}
\fH^{(M_0)}\eqdef \frac{1}{\epsilon} u^{(M_0/\epsilon)}~~
\een
satisfies the Hamilton-Jacobi equation of $(M_0,\Phi)$ for
all $M_0\in (0,\epsilon \mu)$.
\end{prop}

\noindent We also have the following result, which relates the
solutions of \eqref{HJeps} to scale transforms of the Hamilton-Jacobi
functions of a different model. 

\begin{prop}
\label{prop:HJepsilon}
The $\epsilon$-deformed Hamilton-Jacobi equation of the model
$(M_0,\Phi)$ coincides with the Hamilton-Jacobi equation of
the model $(\epsilon M_0,\epsilon^2\Phi)$. Moreover, $\fH$ satisfies
the Hamilton-Jacobi equation of $(M_0,\Phi)$ iff $u\eqdef \fH_\epsilon$
satisfies the $\epsilon$-deformed Hamilton-Jacobi equation of
the model $(M_0,\Phi_\epsilon)$.
\end{prop}

\begin{remark}
\label{rem:fHrec}
In particular, the solutions $\fH$ of the Hamilton-Jacobi equation of
$(M_0,\Phi)$ can be recovered from the solutions $u^\Phi$ of \eqref{HJeps}
by replacing $\Phi$ with $\Phi_\epsilon$ and setting
$\fH=u_{1/\epsilon}$:
\ben
\label{fHrec}
\fH=(u^{\Phi_\epsilon})_{1/\epsilon}~~\mathrm{i.e.}~~\fH(x)=u^{\Phi_\epsilon}(\epsilon x)~~.
\een
\end{remark}

\begin{proof}
The first statement is obvious. The second statement follows by noticing that we have: 
\ben
\label{parsim}
\r(\epsilon^2,1)(M_0,\Phi,\fH)=(\epsilon M_0,\epsilon^2\Phi_\epsilon,\fH_\epsilon)~~,
\een
where $\r$ is the Hamilton-Jacobi similarity action of Subsection \ref{subsec:HJsim}.
\end{proof}

\noindent When $\epsilon$ is small, the behavior of $\fH_\epsilon$
away from the origin captures the behavior of $\fH(x)$ for $x$ very
large in absolute value. Indeed, we have $\fH_\epsilon(\pm 1)=\fH(\pm
1/\epsilon)$ and $1/\epsilon$ tends to infinity when
$\epsilon\rightarrow 0$.

Below, we construct a formal solution $u$ of \eqref{HJeps} for all
$M_0>0$ as a power series in $\epsilon^2$. Performing the
transformation \eqref{fHrec} (or \eqref{fHrecM}) will then produce a
``universal'' formal solution $\fH_\IR$ of the Hamilton-Jacobi
equation of the model $(M_0,\Phi)$ such that $\fH_\IR/\fH_0$ is given
by a power series in $M_0^2$ whose coefficients are rational functions
of $\Phi$ and its derivatives (more precisely, they are ratios between
a polynomial in $\Phi$ and its derivatives and a power of
$\Phi$). This also produces a series for the {\em infrared
  Hamilton-Jacobi potential} $V_\IR\eqdef M_0^2\fH_\IR$. Conceptually,
this constructs successive approximations for a ``formal infrared flow
orbit'' $\cO_\IR$ of the model $(M_0,\Phi)$ -- which corresponds to an
actual flow orbit only when the series are uniformly convergent on
some interval of the real axis. The resulting formal power series
expansion of $u$ will have the form:
\be
u=\frac{1}{M_0}\sum_{n\geq 0} (\epsilon M_0)^{2n} U_n~~,
\ee
where $U_n\in \cC^\infty(\R)$ are expressed as functions of $\Phi$ and
its derivatives (as we confirm in Proposition \ref{prop:uuprimeseries}
below). Accordingly,  $\fH_\IR$ has the expansion:
\be
\fH_\IR=\frac{1}{M_0}\sum_{n\geq 0} M_0^{2n} U_n
\ee
which can be obtained from the expansion of $u$ by formally
setting $\epsilon=1$ (a fact which we confirm in Proposition
\ref{prop:subs}). Hence $\fH_\IR$ is given by a
MacLaurent expansion in $M_0$, as could be
expected\footnote{Notice that the expansion of $\fH_\IR$ could be
obtained directly starting from equation \eqref{hatHJ}. However, the
approach followed below allows to separate $M_0$ from the scale
parameter $\epsilon$ and hence is useful for understanding the effect
of RG transformations.}  from the fact that $M_0$ ``counts the number
of derivatives'' in equation \eqref{hatHJ} for the function ${\hat
  \fH}=M_0\fH$ (see Remark \ref{rem:hatHJ}). As a consequence, $V_\IR$ expands as a
power series in $M_0$.

\subsection{Formal IR Hamilton-Jacobi families}

\noindent One can extract a formal equation from \eqref{HJeps} by
taking $J=\R$ and expanding the right hand side as a formal power
series in $\epsilon$. Recall that the element $1+X\in \R[[X]]$ admits
a square root given by the binomial expansion:
\ben
\label{binomial}
(1+X)^{1/2}=\sum_{k\geq 0} a_k X^k=1+\sum_{k\geq 1}a_k X^k~~,
\een
where $a_0=1$ and:
\ben
\label{ak}
a_k=\left(\begin{array}{c}1/2\\ k \end{array}\right)\eqdef \frac{(1/2)(1/2-1)(1/2-2)\ldots (1/2-k+1)}{k!}=(-1)^{k-1} \frac{(2k-3)!}{2^{2k-2} k!(k-2)!}~~\forall k\geq 1
\een
are the generalized binomial coefficients, the first of which are:
\beqan
a_1=\frac{1}{2}~~,~~a_2=-\frac{1}{8}~~,~~a_3=\frac{1}{16}~~,~~a_4=-\frac{5}{128}~~.
\eeqan
Using \eqref{binomial} in \eqref{HJeps} produces the {\em formal deformed Hamilton-Jacobi equation}:
\ben
\label{HJexp}
u=\frac{\sqrt{2\Phi}}{M_0}\left[1+\sum_{k\geq 1} a_k \frac{\epsilon^{2k} M_0^{4k}}{(2\Phi)^{k}} (u')^{2k}\right]~~,
\een
where the unknown:
\ben
\label{uexp}
u=\sum_{n\geq 0} u_n \epsilon^{2n} \in \cC^\infty(\R)[[\epsilon]]
\een
is a formal power series in $\epsilon$ whose coefficients are smooth
univariate functions $u_n\in \cC^\infty(\R)$ and we defined:
\ben
\label{uprimeexp}
u'\eqdef \sum_{n\geq 0} u'_n \epsilon^{2n}\in \cC^\infty(\R)[[\epsilon]]~~.
\een

\begin{definition}
The {\em formal IR Hamilton-Jacobi family} of the model $(M_0,\Phi)$
is the formal power series $u\in \cC^\infty(\R)[[\epsilon]]$ which
satisfies the formal deformed Hamilton-Jacobi equation
\eqref{HJexp}. The {\em formal IR effective potential family} $V_\eff\in
\cC^\infty(\R)[[\epsilon]]$ of the model $(M_0,\Phi)$ is defined
through:
\ben
\label{Veff}
V_\eff \eqdef M_0^2 u~~.
\een
The {\em formal IR Hamilton-Jacobi function}  of the model $(M_0,\Phi)$ is the 
following series with coefficients in the ring $\cC^\infty(\R)$:
\ben
\label{fHIR}
\fH_\IR\eqdef (u^{\Phi_\epsilon})_{1/\epsilon}~~.
\een
The {\em formal IR effective potential} of the model $(M_0,\Phi)$ is
the following series with coefficients in the ring $\cC^\infty(\R)$:
\ben
\label{VIR}
V_\IR\eqdef M_0^2\fH_\IR=M_0^2 (u^{\Phi_\epsilon})_{1/\epsilon}~~,
\een
Here, the right hand sides of \eqref{fHIR} and \eqref{VIR} are
understood as the sequences of their partial sums.
\end{definition}

\noindent We will see in a moment that $u$ is uniquely determined by
$M_0$ and $\Phi$ -- and hence so are $V_\eff$, $\fH_\IR$ and $V_\IR$.

\begin{prop}
\label{prop:un}
The formal IR Hamilton-Jacobi family $u\in \cC^\infty(\R)[[\epsilon]]$
of the model $(M_0,\Phi)$ is unique. Moreover, we have:
\ben
\label{u0}
u_0=\fH_0=\frac{\sqrt{2\Phi}}{M_0}
\een
and the functions $u_n\in \cC^\infty(\R)$ satisfy the recursion relation: 
\ben
\label{unrec}
u_n=\sum_{k=1}^n\sum_{s_1+\ldots +s_{2k}=n-k} a_k M_0^{2k}\frac{u'_{s_1}\ldots u'_{s_{2k}}}{u_0^{2k-1}}~~\forall n\geq 1~~,
\een
where $s_j\in \N$ for all $j\geq 1$.
\end{prop}

\noindent The first two relations in \eqref{unrec} read:
\beqan
\label{u12}
&&u_1=a_1M_0^2 \frac{(u'_0)^2}{u_0}\Longrightarrow u'_1=a_1M_0^2 \frac{2u_0u'_0u_0''-u_0'^3}{u_0^2}\\
&&u_2=a_2 M_0^4\frac{u_0'^4}{u_0^3}+2a_1 M_0^2\frac{u_0'u_1'}{u_0}=(a_2-2a_1^2)M_0^4 \frac{u_0'^4}{u_0^3}+4a_1^2 M_0^4\frac{(u_0')^2u_0''}{u_0^2}~~.\nn
\eeqan

\begin{proof}
Using the expansions \eqref{uexp} and \eqref{uprimeexp} in equation
\eqref{HJexp} and rearranging terms gives \eqref{u0} and the recursion
relation:
\ben
\label{un}
u_n=\frac{\sqrt{2\Phi}}{M_0} \sum_{k=1}^n \sum_{s_1+\ldots +s_{2k}=n-k} \frac{a_k M_0^{4k}}{(2\Phi)^{k}}u'_{s_1}\ldots u'_{s_{2k}}~~\forall n\geq 1~~.
\een
The latter is equivalent to \eqref{unrec} upon using \eqref{u0}.
\end{proof}

\subsection{Expressing $u$ through $u_0$ and its derivatives}

\paragraph{Preparations.}

Let $\cC^\infty_x$ denote the space of germs of smooth
real-valued univariate functions at $x\in \R$ and $J_x^\nu$ denote the
space of jets of order $\nu$ of such functions at $x$, where $\nu\in
\N\cup \{\infty\}$. The $\R$-vector space $J^\infty_x$ has a natural
ring structure which we recall in Appendix \ref{app:jets}. For any
non-empty open interval $K\subset \R$, any point $x\in J$ and any
smooth real-valued function $f\in \cC^\infty(K)$, let $f_x\in
\cC^\infty_x$ and $\rj^\nu_x(f)\in J^\nu_x$ denote respectively the
germ and jet of order $\nu$ of $f$ at $x$ (see Appendix \ref{app:jets}
for the precise definition of $\rj_x^\nu(f)$). We identify
$J_x^\infty$ with the infinite-dimensional vector space
$\R^\infty\eqdef \R^\N$ by sending $\rj^\infty_x(f)$ to the
{\em reduced infinite jet} $j^\infty_x(f)\in \R^\N$, 
which we define as the infinite vector with components:
\be
[j^\infty_x(f)]^k\eqdef f^{(k)}(x)~~\forall k\in \N~~.
\ee
Accordingly, for any $n\in \N$, we identify $J_x^n$ with the vector
space $\R^{n+1}$ by sending $\rj^n_x(f)$ to the {\em reduced $n$-th
  order jet} $j^n_x(f)\in \R^{n+1}$, which we define as the
$(n+1)$-vector with components:
\be
[j^n_x(f)]^k=f^{(k)}(x)~~\forall k=0,\ldots,n~~.
\ee
This gives linear isomorphisms:
\be
j_x^\nu:J_x^\nu\stackrel{\sim}{\rightarrow}\R^\nu~~\forall \nu\in \N\cup\{\infty\}~~,
\ee
where surjectivity follows from Borel's lemma. For any function $f\in
\cC^\infty(K)$ as above and all $\nu\in \N\cup\{\infty\}$, we denote
by $j^\nu(f):K\rightarrow \R^\nu$ the function defined through:
\be
j^\nu(f)(x)\eqdef j_x^\nu(f)~~\forall x\in K~~.
\ee

For later reference, we define the following cones in $\R^n$ and $\R^\N$:
\be
\R_0^n\eqdef \{w=(w_0,\ldots, w_n)\in \R^n~~\vert~w_0>0\}~~\forall n\in \N~~,~~\R_0^\N \eqdef \{w=(w_n)_{n\geq 0}\in \R^\N~~\vert~w_0>0\}~~.
\ee
For every $n\in \N$, let:
\ben
\label{Xleqn}
w_{\leq n}\eqdef (w_0,\ldots w_n)
\een
denote the $n$-th truncation of the infinite sequence $w=(w_j)_{j\geq
  0}\in \R^\N$ and let $\tau_n:\R^\N\rightarrow \R^{n+1}$ be the
corresponding truncation map:
\be
\tau_n(w)=w_{\leq n}~~\forall w\in \R^\N~~.
\ee

\paragraph{The Hamilton-Jacobi polynomials.}

Since the sum in the right hand side forces $s_1+\ldots
+s_{2k}=n-k<n$ and hence $s_j<n$, relation \eqref{unrec} expresses $u_n$
in terms of $u_0$ and $u'_0,\ldots,u'_{n-1}$. Together with \eqref{u0}, this
determines $u_n$ recursively for all $n\geq 1$ as a function of
$u_0=\frac{\sqrt{2\Phi}}{M_0}$ and its derivatives of order $\leq n$:
\be
u_n(x)=f_n(u^{(0)}_0(x),\ldots,u^{(n)}_0(x))~~\mathrm{i.e.}~~u_n(x)=(f_n\circ j^n_x)(u_0)~~
\ee
for some rational function $f_n:\R_0^{n+1}\rightarrow \R$. The
expansion \eqref{uexp} becomes:
\ben
\label{usol}
u=\sum_{n\geq 0} \epsilon^{2n} (f_n\circ j^n)(u_0)\in \cC^\infty(\R)[[\epsilon]]~~.
\een

To describe the structure of $f_n$, we will use the algebra $\cA$ of
polynomials with real coefficients in a countable set of variables
$X_k$ ($k\in \N$). Formally, $\cA\eqdef \R[\Sigma]\simeq \R[\N]$ is
the free associative and commutative algebra on the countable set:
\be
\Sigma\eqdef \{X_k~\vert~k\in \N\}~~.
\ee
This algebra admits an exhaustive ascending filtration by the
subspaces $\cA^n=\R[X_0,\ldots, X_n]$:
\be
0\subset \cA^0\subset \cA^1\subset \cA^2\subset \ldots \subset \cA~~,~~\cA=\cup_{n\geq 0}\cA^n
\ee
and admits gradings defined by the following degrees:
\begin{itemize}
\item The {\em ordinary polynomial degree} $\deg$, with respect to which:
\be
\deg(X_k)=1~~\forall k\in \N~~.
\ee
\item The {\em order-weighted degree} $\wdeg$, with respect to which:
\be
\wdeg(X_k)=k~~\forall k\in \N~~.  
\ee
\end{itemize}
We combine these into the $\N\times\N$-valued bigrading which associates to each monomial $T$
the bidegree $(\deg(T),\wdeg(T))$ (in this order).

Let $\pd_{X_j}$ be the canonical derivations of $\cA$. Given a
polynomial $T\in \R[X_0,\ldots, X_n]$ and any $\alpha\in \N$, define
$T'\in \R[X_0,\ldots, X_{n+1}]$ and $D T \in \R[X_0,\ldots, X_{n+1}]$
by:
\be
T'\eqdef \sum_{j=0}^n X_{j+1} \frac{\pd T}{\pd X_j}~~,~~D(\alpha) T\eqdef X_0 T'-\alpha X_1 T~~.
\ee
Then $(-)'$ and $D(\alpha)$ are differential operators on $\cA$ of
bidegrees $(0,1)$ and $(1,1)$ respectively; in particular, they map
bi-homogeneous polynomials to bi-homogeneous polynomials. Also notice the
following relation in the algebra of rational functions $\R(\Sigma)$:
\ben
\label{Dfrac}
\left(\frac{T}{X_0^\alpha}\right)'=\frac{D(\alpha) T}{X_0^{\alpha+1}}~~,
\een
which implies:
\ben
\label{Dfrack}
\left(\frac{T}{X_0^\alpha}\right)^{(k)}=\frac{D(\alpha+k-1)\ldots D(\alpha+1)D(\alpha)T}{X_0^{\alpha+k}}~~\forall k\geq 1~~.
\een

\begin{definition}
The {\em Hamilton Jacobi polynomials} $Q_n\in \Q[X_0,\ldots,X_n]$ are
defined through the recursion relation:
\ben
\label{Qrec}
Q_n=\sum_{k=1}^n \sum_{s_1+\ldots+ s_{2k}=n-k} a_k (D(2s_1-1) Q_{s_1})\ldots (D(2s_{2k}-1) Q_{s_{2k}})~~\forall n\geq 1
\een
with the initial condition $Q_0=1$, where $a_k$ are the generalized
binomial coefficients \eqref{ak}.
\end{definition}

\begin{remark}
For all $n\geq 0$, we have:
{\footnotesize
\ben
\label{Dn}  
D(2n-1) Q_n=\left[X_0\sum_{j=0}^{n}  X_{j+1}\frac{\pd}{\pd X_j} -(2n-1) X_1 \right] Q_n=X_1\left[X_0 \frac{\pd}{\pd X_0} -(2n-1)\right]Q_n +
X_0 \sum_{j=1}^{n}  X_{j+1}\frac{\pd}{\pd X_j} Q_n~~.
\een}
\end{remark}

\noindent The first nontrivial polynomials $Q_n$ are (see Appendix
\ref{app:list} for a list of these polynomials up $n=10$):
{\scriptsize \beqan
&&Q_1=\frac{1}{2}X_1^2~~,~~Q_2=-\frac{5}{8}X_1^4+X_0 X_1^2 X_2~~,~~Q_3=\frac{37}{16}X_1^6-\frac{11}{2} X_0 X_1^4 X_2+\frac{5}{2} X_0^2  X_1^2 X_2^2+X_0^2 X_1^3 X_3~~,\\
&& Q_4=-\frac{1773}{128} X_1^8+\frac{347}{8} X_0 X_1^6  X_2-\frac{147}{4} X_0^2  X_1^4 X_2^2+7 X_0^3 X_1^2 X_2^3-\frac{19}{2}X_0^2  X_1^5 X_3 + 9 X_0^3  X_1^3 X_2 X_3+X_0^3 X_1^4  X_4 ~~.\nn
\eeqan}

\begin{prop}
\label{prop:HJpoly}
For all $n\geq 0$, the Hamilton-Jacobi polynomial $Q_n$ is homogeneous
of bidegree $(2n,2n)$. If $n\geq 1$, then the largest power of $X_0$
in the constituent monomials of $Q_n$ is strictly smaller than $n$.
\end{prop}

\begin{proof}
The bi-homogeneity statement is true for $n=0$. Let $n\geq 1$ and
assume that the statement holds for all $s<n$. Then relation
\eqref{Qrec} implies that the statement also holds for $n$ and hence
it holds for all $n$ by induction.

The statement about the largest power of $X_0$ in $Q_n$ also follows
by induction over $n$. The statement is obviously true for $n=1$ since
$Q_1=\frac{1}{2}X_1^2$. If $s\geq 1$ and the largest power of $X_0$ in
$Q_s$ is strictly smaller than $s$, then relation \eqref{Dn} shows
that the largest power of $X_0$ in $D(2s-1)Q_s$ is at most $s$. The
latter is also true if $s=0$ since $D(-1)Q_0=X_1$. Now if
$n\geq 2$ and the statement holds for all $0\leq s<n$, then the
previous remark and the recursion relation \eqref{Qrec} shows that it
also holds for $n$.
\end{proof}

\begin{remark}
\label{rem:admissibleQ}
Consider a monomial:
\be
q\eqdef X_0^{k_0}\ldots X_n^{k_n}\in \Q[X_0,\ldots, X_n]~~
\ee
where $n, k_0,\ldots,k_n \in \N$. This monomial has bidegree $(2n,2n)$
iff $k_0,\ldots, k_n$ satisfy:
\ben
\label{kconds}
\sum_{j=0}^n{k_j}=2n~~\&~~\sum_{j=0}^n j k_j=2n\Longleftrightarrow \sum_{j=1}^n j k_j=2n ~~\&~~ k_0=2n-\sum_{j=1}^{n} k_j~~.
\een
If $n\geq 1$ and $q$ appears as a monomial in $Q_n$, then Proposition
\ref{prop:HJpoly} shows that we also have $k_0<n$, i.e.  $\sum_{j=1}^n
k_j>n$. 
\end{remark}

\paragraph{The functions $f_n$.}

\begin{prop}
\label{prop:unpoly}
For any $n\geq 0$, we have:
\ben
\label{ununprime}
u_n(x)=M_0^{2n}\frac{Q_n(j^n_x(u_0))}{u_0(x)^{2n-1}}~~\mathrm{and}~~u'_n(x)=M_0^{2n}\frac{(D(2n-1) Q_n)(j^{n+1}_x(u_0))}{u_0(x)^{2n}}~~.
\een
\end{prop}

\begin{proof}
We proceed by induction over $n$. Relations \eqref{ununprime} hold for
$n=0$ since $Q_0=1$. They also hold for $n=1$ by the first equation in
\eqref{u12}. Assume that $n\geq 2$ and that \eqref{ununprime} holds
for all $s=0,\ldots, n-1$. Then the second relation in
\eqref{ununprime} gives:
\be
u_s'(x)=M_0^{2s}\frac{(D(2s-1)Q_s)(u_0^{(0)}(x),\ldots, u^{(s+1)}_0(x))}{u_0(x)^{2s}}~~\forall s=0,\ldots, n-1~~.
\ee
Now the recursion relation \eqref{unrec} and the definition of $Q_n$ implies that
the first relation in \eqref{ununprime} holds for $u_n$. In turn, this implies that the
second relation in \eqref{ununprime} holds for $u_n$. By induction we conclude that
relations \eqref{ununprime} hold for all $n\geq 0$.
\end{proof}

\begin{cor}
\label{cor:f}
The rational functions $f_n$ are homogeneous of bidegree $(1,2n)$, being given by:
\be
f_n(w_0,\ldots, w_n)=M_0^{2n}\frac{Q_n(w_0,\ldots, w_n)}{w_0^{2n-1}}~~\forall n\geq 0~~\forall (w_0,\ldots,w_n)\in \R_0^{n+1}~~.
\ee
In particular, we have $f_0(w_0)=w_0$.
\end{cor}

\noindent In the proposition below and in what follows, we use the
same notation for a polynomial and the polynomial function which it
induces by evaluation over a field.

\begin{prop}
\label{prop:uuprimeseries}
We have:
\ben
\label{uexpQ}
u=\sum_{n\geq 0} (\epsilon M_0)^{2n} \frac{(Q_n\circ j^n)(u_0)}{u_0^{2n-1}}=\sum_{n\geq 0} \epsilon^{2n}  M_0^{2n-1}\frac{(Q_n\circ j^n)(\sqrt{2\Phi})}{(2\Phi)^{n-1/2}}
\een
and:
\ben
\label{uprimeexpQ}
u'=\sum_{n\geq 0} (\epsilon M_0)^{2n}\frac{((D(2n-1)Q_n)\circ j^{n+1})(u_0)}{u_0^{2n}}=\sum_{n\geq 0} \epsilon^{2n} M_0^{2n-1}\frac{((D(2n-1)Q_n)\circ j^{n+1})(\sqrt{2\Phi})}{(2\Phi)^{n}}~~.
\een
\end{prop}

\begin{proof}
Relations \eqref{uexpQ} and \eqref{uprimeexpQ} follow immediately from
Corollary \ref{cor:f} and the expansion \eqref{uexp} upon using
\eqref{u0} and the fact that $Q_n$ and $D(2n-1) Q_n$ are homogeneous
of degrees $2n$ and $2n+1$ respectively with respect to the ordinary
polynomial grading.
\end{proof}

\subsection{The expansions of $\fH_\IR$ and $V_\IR$}

\begin{definition}
The {\em weighted scaling operator} of order $n\in \N$ and parameter
$\kappa>0$ is the linear operator $\fS_n(\kappa):\R^{n+1}\rightarrow
\R^{n+1}$ defined through:
\be
\fS_n(\kappa)(w_0,\ldots, w_n)\eqdef (w_0,\kappa w^1,\ldots, \kappa^n w_n)~~\forall (w_0,\ldots, w_n)\in \R^{n+1}.
\ee
\end{definition}

\noindent Notice that $\fS_0=\id_\R$.

\begin{lemma}
\label{lemma:unscalings}
The functions $u_n\in \cC^\infty(\R)$ are homogeneous of degree $1/2$
under a positive constant rescaling of $\Phi$ and homogeneous of
degree $2n-1$ under a positive constant rescaling of $M_0$:
\be
u_n^{\xi \Phi}=\xi^{1/2} u_n^\Phi~~,~~u_n^{\xi M_0}=\xi^{2n-1} u_n^{M_0}~~\forall \xi>0~~.
\ee
Moreover, they behave as follows under a scale transformation of $\Phi$:
\be
u_n^{\Phi_\epsilon}=\frac{1}{\epsilon^{2n}} (u_n^\Phi)_\epsilon~~\forall \epsilon>0~~
\ee
and hence they satisfy:
\be
(u_n^{\Phi_\epsilon})_{1/\epsilon}=\frac{1}{\epsilon^{2n}} u_n^\Phi~~\forall \epsilon>0~~.
\ee
\end{lemma}

\begin{proof}
Notice that $u_0=\frac{\sqrt{2\Phi}}{M_0}$ is homogeneous of degree
$1/2$ respectively $-1$ under a positive constant rescaling of $\Phi$ or $M_0$:
\ben
\label{u0hom}
u_0^{\xi \Phi}=\xi^{1/2}u_0^\Phi~~,~~u_0^{\xi M_0}=\xi^{-1}u_0^{M_0}~~\forall \xi>0
\een
and that it transforms as follows under a scale transformation of $\Phi$:
\be
u_0^{\Phi_\epsilon}=(u_0^\Phi)_\epsilon~~\forall \epsilon>0~~.
\ee
The last relation implies: 
\ben
\label{ju0scaling}
j^n_x(u_0^{\Phi_\epsilon})=\fS_n(1/\epsilon) (j^n_x(u_0^\Phi))_\epsilon~~\forall \epsilon>0~~.
\een
Proposition \ref{prop:uuprimeseries} now implies the conclusion upon
using \eqref{u0hom} and \eqref{ju0scaling} and the fact that $Q_n$ is
homogeneous of bidegree $(2n,2n)$.
\end{proof}

\begin{prop}
\label{prop:subs}
The formal IR Hamilton-Jacobi function and formal IR effective
potential can be obtained respectively by substituting $\epsilon=1$ in
$u$ and $V_\eff$:
\ben
\label{fHIRexp0}
\fH_\IR=u\vert_{\epsilon=1}=\sum_{n\geq 0} M_0^{2n} \frac{(Q_n\circ j^n)(u_0)}{u_0^{2n-1}}=\sum_{n\geq 0} M_0^{2n-1}\frac{(Q_n\circ j^n)(\sqrt{2\Phi})}{(2\Phi)^{n-1/2}}
\een
and:
\ben
\label{VIRexp0}
V_\IR=M_0^2 \fH_\IR=\sum_{n\geq 0} M_0^{2n+1}\frac{(Q_n\circ j^n)(\sqrt{2\Phi})}{(2\Phi)^{n-1/2}}~~,
\een
where the series in the right hand side are understood as sequences of
their partial sums.
\end{prop}

\begin{proof}
Follows immediately from \eqref{fHIR} and \eqref{VIR} using Proposition
\ref{prop:uuprimeseries} and Lemma \ref{lemma:unscalings}.
\end{proof}

\noindent Let us momentarily denote $\fH_\IR$ by
$\fH_\IR^{(M_0,\Phi)}$ to indicate its dependence on the parameters of
the model. The following result shows that the set:
\be
\{(M_0,\Phi,\fH_\IR^{(M_0,\Phi)})\vert (M_0,\Phi)\in \Par\}
\ee
is invariant under the Hamilton-Jacobi similarity action $\r$ of the
group $\rT=\R_{>0}^2$ (see Subsection \ref{subsec:HJsim}) and in
particular under the action $\r_\ren$ of the renormalization group
$\rT_\ren$. 

\begin{cor}
The functions $\fH_\IR$ and $V_\IR$ transform as follows under the
action $\rho_\param$ of the universal similarity group $\rT$ on the
set of parameters $(M_0,\Phi)$:
\beqa
&& \fH_\IR^{\rho_\param(\lambda,\epsilon)(M_0,\Phi)}=\frac{1}{\epsilon}(\fH_\IR^{(M_0,\Phi)})_{\lambda^{1/2}}=\r_0(\lambda,\epsilon)(\fH_\IR^{(M_0,\Phi)})\nn\\
&& V_\IR^{\rho_\param(\lambda,\epsilon)(M_0,\Phi)}= \frac{\lambda}{\epsilon}(V_\IR^{(M_0,\Phi)})_{\lambda^{1/2}}=\lambda\,\r_0(\lambda,\epsilon)(V_\IR^{(M_0,\Phi)})~~,
\eeqa
where $(\lambda,\epsilon)\in \rT$ and $\r_0(\lambda,\epsilon)$ (which
was given in Definition \ref{def:rdef}) is applied to the partial sums
of series. In particular, $\fH_\IR$ and $V_\IR$ satisfy:
\be
\fH_\IR^{\rho_\ren^\param(\epsilon)(M_0,\Phi)}=\frac{1}{\epsilon}(\fH_\IR^{(M_0,\Phi)})_{\epsilon}=\r_\ren^0(\epsilon)(\fH_\IR^{(M_0,\Phi)})~~,~~
V_\IR^{\rho_\ren^\param(\epsilon)(M_0,\Phi)}= \epsilon (V_\IR^{(M_0,\Phi)})_{\epsilon}=\epsilon^2 \r_\ren^0(\epsilon)(V_\IR^{(M_0,\Phi)})~~
\ee
for all $\epsilon>0$.
\end{cor}

\begin{proof}
Recall that
$\rho_\param(\lambda,\epsilon)(M_0,\Phi)=(\lambda^{1/2}M_0,
\frac{\lambda}{\epsilon^2}\Phi_{\lambda^{1/2}})$. Lemma
\ref{lemma:unscalings} gives:
\be
u_n^{(\lambda^{1/2}M_0, \frac{\lambda}{\epsilon^2}\Phi_{\lambda^{1/2}})}=\frac{1}{\epsilon} (u_n^{(M_0,\Phi)})_{\lambda^{1/2}}=\r_0(\lambda,\epsilon)(u_n^{(M_0,\Phi)})~~,
\ee
which implies the conclusion. 
\end{proof}

\subsection{Expansion in powers of $\Phi$ and its derivatives}

\begin{lemma}
\label{lemma:DerSqrtPhi}
For any $n\geq 1$, we have:
\be
(\sqrt{\Phi})^{(n)}=\sqrt{\Phi}\sum_{l=1}^n\sum_{{\scriptsize \begin{array}{c}k_1+\ldots+k_l=n\\k_1,\ldots,k_l\geq 1\end{array}}} a_l \frac{n!}{k_1!\ldots k_l!} \frac{\Phi^{(k_1)}\ldots \Phi^{(k_l)}}{\Phi^l}~~
\ee
and hence:
\ben
\label{sqrtder}
(\sqrt{2\Phi})^{(n)}=\sqrt{2\Phi} \frac{R_n(j^n(\Phi))}{\Phi^n}~~\forall n\geq 1~~,
\een
where:
\be
R_n=\sum_{l=1}^n\sum_{{\scriptsize \begin{array}{c}k_1+\ldots+k_l=n\\k_1,\ldots,k_l\geq 1\end{array}}} a_l \frac{n!}{k_1!\ldots k_l!} X_0^{n-l} X_{k_1}\ldots X_{k_l}\in \Q[X_0,\ldots X_n]
\ee
is a homogeneous polynomial of bidegree $(n,n)$. 
\end{lemma}

\begin{remark}
\label{rem:admissibleR}
Relation \eqref{sqrtder} also holds for $n=0$ if we define $R_0\eqdef
1\in \Q[X_0]$. If $n\geq 1$, then notice that the largest power of
$X_0$ which appears in $R_n$ is strictly smaller than $n$.
\end{remark}

\begin{proof}
The Taylor expansion of $\Phi$ around $x$ and the binomial expansion of the square root give:
\beqa
&&\sqrt{\Phi(x+\eta)}=\sqrt{\Phi(x)}\sqrt{1+\sum_{k\geq 1}
  \eta^k\frac{\Phi^{(k)}(x)}{k!\Phi(x)}}=\sqrt{\Phi(x)}\left[1+\sum_{l\geq
    1} a_l \left(\sum_{k\geq 1}
  \eta^k\frac{\Phi^{(k)}(x)}{k!\Phi(x)}\right)^l\right]\nn\\ &&=\sqrt{\Phi(x)}\left[1+\sum_{n\geq
    1} \eta^n \sum_{l=1}^n \sum_{{\scriptsize \begin{array}{c}k_1+\ldots+k_l=n\\k_1,\ldots,k_l\geq 1\end{array}}} a_l
  \frac{1}{k_1!\ldots k_l!} \frac{\Phi^{(k_1)}(x)\ldots
    \Phi^{(k_l)}(x)}{\Phi(x)^l}\right]~~.
\eeqa
The conclusion now follows by comparing this with the Taylor expansion of $\sqrt{\Phi}$ around $x$:
\be
\sqrt{\Phi(x+\eta)}=\sqrt{\Phi(x)}+\sum_{n\geq 1} \eta^n \frac{(\sqrt{\Phi})^{(n)}(x)}{n!}~~.
\ee
\end{proof}

\noindent The first nontrivial polynomials $R$ are:
{\footnotesize \beqan
&&R_1=\frac{X_1}{2}~~,~~R_2=-\frac{1}{4}X_1^2+\frac{1}{2}X_0 X_2~~,~~R_3=\frac{3}{8} X_1^3-\frac{3}{4} X_0 X_1 X_2+\frac{1}{2} X_0^2 X_3\nn\\
&&R_4=-\frac{15}{16} X_1^4+\frac{9}{4} X_0 X_1^2 X_2-\frac{3}{4} X_0^2 X_2^2-X_0^2 X_1 X_3+\frac{1}{2} X_0^3 X_4~~.
\eeqan}

\begin{cor}
\label{cor:jsqrt}
For any $n\geq 0$, we have:
\ben
\label{jg}
j^n(\sqrt{2\Phi})=\sqrt{2\Phi} (g_n\circ j^n)(\Phi)
\een
where $g_n:\R_0^{n+1}\rightarrow \R^{n+1}$ is a function whose
components $g_n^k$ are rational and homogeneous of bidegree $(0,k)$:
\ben
\label{gn}
g_n^k(w_0,\ldots w_n)=\frac{R_k(w_0,\ldots, w_k)}{w_0^k}~~\forall k=0,\ldots, n~~.
\een
\end{cor}

\begin{proof}
Follows immediately from Lemma \ref{lemma:DerSqrtPhi}.
\end{proof}

\noindent Notice that $g_0(w_0)=1$ for all $w_0\in \R_{>0}$. 

\begin{prop}
\label{prop:ExpPhi}
We have:
\ben
\label{uexpQPhi}
u=\fH_0\sum_{n\geq 0} (\epsilon M_0)^{2n}(v_n\circ j^n)(\Phi)~~\mathrm{and}~~V_\eff =V_0\sum_{n\geq 0} (\epsilon M_0)^{2n} (v_n\circ j^n)(\Phi)~~,
\een
where $\fH_0=\frac{\sqrt{2\Phi}}{M_0}$, $V_0\eqdef M_0\sqrt{2\Phi}$ and:
\ben
\label{vnf}
v_n\eqdef Q_n\circ g_n:\R_0^{n+1}\rightarrow \R
\een
is a rational function of bidegree $(0,2n)$ with coefficients in $\Q$,
whose abstract form is:
\ben
\label{vn}
v_n(X_0,\ldots, X_n)=\frac{S_n(X_0,\ldots,X_n)}{X_0^{2n}}\in \Q(X_0,\ldots,X_n)~~,
\een
where:
\ben
\label{Sn}
S_n(X_0,\ldots,X_n)\eqdef Q_n(1,R_1(X_0,X_1),\ldots, R_n(X_0,\ldots,X_n))\in \Q[X_0,\ldots, X_n]~~
\een
is a homogeneous polynomial of bidegree $(2n,2n)$. If $n\geq 1$, then the
largest power of $X_0$ which appears in $S_n$ is strictly smaller than $n$. 
\end{prop}

\noindent Notice that $v_0=1$. 

\begin{proof}
Using \eqref{jg} in \eqref{uexpQ} gives \eqref{uexpQPhi} with $v_n$ as
in \eqref{vnf} upon recalling that $Q_n$ is homogeneous of degree $2n$
with respect to the ordinary polynomial grading. Relations \eqref{vnf}
and \eqref{gn} give:
\ben
\label{vnR}
v_n=Q_n\left(1,\frac{R_1}{X_0},\frac{R_2}{X_0^2},\ldots, \frac{R_n}{X_0^n}\right)~~.
\een
This is equivalent with \eqref{vn} with $S_n$ as in \eqref{Sn} since
$Q_n$ is homogeneous of degree $2n$ with respect to the order-weighted
grading. To prove the last statement, let us assume that $n\geq
1$. Then Remark \ref{rem:admissibleR} shows that the largest power of
$X_0$ in each polynomial $R_j$ with $j\in \{1,\ldots, n\}$ is strictly
smaller than $n$ and hence for all such $j$ the rational function
$\frac{R_j}{X_0^j}$ is a sum of Laurent monomials in which $X_0$
appears with power $\leq -1$. Since $Q_n$ has
ordinary polynomial degree $2n$ and $X_0$ appears in $Q_n$ with
exponent strictly smaller than $n$ (see Proposition
\ref{prop:HJpoly}), relation \eqref{vnR} shows that $v_n$ is a sum of
Laurent monomials of the form:
\be
\prod_{j=1}^{n} \left(\frac{R_j}{X_0^j}\right)^{k_j}
\ee
where $k_j\in \N$ satisfy $\sum_{j=1}^n k_j>n$. Together with the
observation above, this implies that $X_0$ enters the Laurent
monomials of $v_n$ with powers which are strictly smaller than $-n$
and hence enters the monomials of $S_n$ with powers
which are strictly smaller than $n$.
\end{proof}

We have:
{\scriptsize \beqan
&&v_1=\frac{X_1^2}{8 X_0^2}~~,~~v_2=-\frac{13 X_1^4}{128 X_0^4}+\frac{X_1^2 X_2}{8 X_0^3}~~,~~
v_3=\frac{213 X_1^6}{1024 X_0^6}-\frac{27 X_1^4 X_2}{64 X_0^5}+\frac{5 X_1^2 X_2^2}{32 X_0^4}+\frac{X_1^3 X_3}{16 X_0^4}\nn\\
&& v_4=-\frac{21949 X_1^8}{32768 X_0^8}+\frac{1907 X_1^6 X_2}{1024 X_0^7}-\frac{351 X_1^4 X_2^2}{256 X_0^6}+\frac{7
  X_1^2 X_2^3}{32 X_0^5}-\frac{45 X_1^5 X_3}{128 X_0^6}+\frac{9 X_1^3
  X_2 X_3}{32 X_0^5}+\frac{X_1^4 X_4}{32 X_0^5}
\eeqan}
and hence:
{\scriptsize \beqan
\label{Veff4}
&&\frac{V_\eff}{V_0}=1+(\epsilon M_0)^2 \frac{ \Phi'(x)^2}{8 \Phi
    (x)^2}+(\epsilon M_0)^4 \Big(-\frac{13 \Phi'(x)^4}{128 \Phi
    (x)^4}+\frac{\Phi'(x)^2 \Phi''(x)}{8 \Phi (x)^3}\Big)+\\
&&+(\epsilon M_0)^6 \Big(\frac{213 \Phi'(x)^6}{1024 \Phi
    (x)^6}-\frac{27 \Phi'(x)^4 \Phi''(x)}{64 \Phi (x)^5}+\frac{5
    \Phi'(x)^2 \Phi''(x)^2}{32 \Phi (x)^4}+\frac{\Phi'(x)^3 \Phi
    ^{(3)}(x)}{16 \Phi(x)^4}\Big)+\nn\\
&&+(\epsilon M_0)^8 \Big(-\frac{21949 \Phi'(x)^8}{32768
    \Phi(x)^8}+\frac{1907 \Phi'(x)^6 \Phi''(x)}{1024 \Phi
    (x)^7}-\frac{351\Phi'(x)^4 \Phi''(x)^2}{256 \Phi (x)^6}+ \frac{7
  \Phi'(x)^2 \Phi''(x)^3}{32 \Phi (x)^5}-\nn\\
&& -\frac{45 \Phi'(x)^5 \Phi ^{(3)}(x)}{128 \Phi(x)^6} +\frac{9
  \Phi'(x)^3 \Phi''(x) \Phi^{(3)}(x)}{32 \Phi (x)^5}+ \frac{\Phi'(x)^4
  \Phi ^{(4)}(x)}{32 \Phi (x)^5}\Big)~~.\nn
\eeqan}

\begin{cor}
\label{cor:IRexpansions}
The formal IR Hamilton-Jacobi function and formal IR effective 
potential of the model $(M_0,\Phi)$ are given by:
\ben
\label{fHIRexp}
\fH_\IR =u\vert_{\epsilon=1}=\fH_0\sum_{n\geq 0} M_0^{2n} (v_n\circ j^n)(\Phi)~~,
\een
and:
\ben
\label{VIRexp}
V_\IR =V_\eff\vert_{\epsilon=1}=V_0\sum_{n\geq 0} M_0^{2n} (v_n\circ j^n)(\Phi)~~,
\een
where $\fH_0\eqdef \frac{\sqrt{2\Phi}}{M_0}$, $V_0\eqdef
M_0\sqrt{2\Phi}$ and the series in the right hand side are understood
as the sequences of their partial sums.
\end{cor}

\begin{remark}
The series in Corollary \ref{cor:IRexpansions} can also be viewed as
formal power series in the rescaled Planck mass $M_0$ (with
coefficients in the ring $\cC^\infty(\R)$).
\end{remark}

\subsection{Rewriting the expansion in term of slow roll functions}

\noindent Suppose that $n\geq 1$. Since the polynomial $S_n$ of
Proposition \ref{prop:ExpPhi} has bidegree $(2n,2n)$, the rational
function $v_n\in \Q(X_0,\ldots, X_n)$ of \eqref{vn} is a finite linear
combination of Laurent monomials of the form:
\be
W=\frac{X_0^{k_0} X_1^{k_1}\ldots X_n^{k_n}}{X_0^{2n}}
\ee
where $k_0,\ldots, k_n\in \N$ satisfy conditions \eqref{kconds}
and we have $k_0=2n-\sum_{j=1}^n k_j<n$. These conditions allow us
to write $W$ as:
\be
W=\frac{X_1^{k_1}\ldots X_n^{k_n}}{X_0^{\sum_{j=1}^n k_j}}~~,
\ee
where $k_1,\ldots, k_n\in \N$ satisfy $\sum_{j=1}^n
jk_j=2n$ and $\sum_{j=1}^n k_j>n$.

\begin{definition}
The {\em slow roll rational functions} $(\beta_j)_{j\geq 0}$ at
rescaled Planck mass $M_0$ are defined through:
\beqan
\label{betadef}
&&\beta_0 \eqdef M_0^2 \frac{X_1^2}{X_0^2}~~\\
&&\beta_j\eqdef M_0^{2j} \frac{X_1^{j-1} X_{j+1}}{X_0^j}~~\forall j\geq 1~~.\nn
\eeqan
\end{definition}

\noindent We have $\beta_0\in \R(X_0,X_1)$, $\beta_1\in
\R(X_0,X_2)$ and $\beta_j\in \R(X_0,X_1, X_{j+1})$ for all $j\geq
2$.

\begin{lemma}
\label{lemma:monSR}
Let $n\geq 1$ and assume that $k_1,\ldots, k_n\in \N$ satisfy the
conditions:
\be
\sum_{j=1}^n j k_j=2n~~\mathrm{and}~~\sum_{j=1}^n k_j >n~~.
\ee
Define:
\be
\alpha\eqdef -n+\sum_{j=1}^n k_j=n-k_0\in \N_{>0}~~,
\ee
where:
\be
k_0\eqdef 2n-\sum_{j=1}^nk_j \in \{0,\ldots, n-1\}~~.
\ee
Then the following identities are satisfied:
\beqan
&& k_1=2\alpha+\sum_{j=2}^n (j-2)k_j\label{id1}\\
&& \sum_{j=1}^n k_j=2\alpha+\sum_{j=2}^n (j-1) k_j\label{id2}\\
&& n=\alpha+\sum_{j=2}^n (j-1) k_j\label{id3}~~
\eeqan
and we have:
\ben
\label{monbeta}
M_0^{2n} \frac{X_1^{k_1}\ldots X_n^{k_n}}{X_0^{\sum_{j=1}^n k_j}}=\beta_0^\alpha \prod_{j=1}^{n-1} \beta_j^{k_{j+1}}~~.
\een
\end{lemma}

\begin{proof}

\
  
\begin{enumerate}[1.]
\item We have $\sum_{j=2}^n (j-2)k_j=k_1+\sum_{j=1}^n
  (j-2)k_j=k_1+2n-2\sum_{j=1}^n k_j=k_1-2\alpha$, which implies
  \eqref{id1}.
\item Relation \eqref{id2} follows immediately by substituting
  $2\alpha$ from \eqref{id1}.
\item We have $\sum_{j=2}^n (j-1) k_j=\sum_{j=1}^n (j-1) k_j=2n
  -\sum_{j=1}^n k_j=2n-(\alpha+n)=n-\alpha$, which implies
  \eqref{id3}.
\end{enumerate}
\noindent Using \eqref{betadef}, the right hand side of \eqref{monbeta} reads:
\beqa
&&\beta_0^\alpha \prod_{j=1}^{n-1} \beta_j^{k_{j+1}}=M_0^{2\left(\alpha+\sum_{j=1}^{n-1} j k_{j+1}\right)}\frac{X_1^{2\alpha+\sum_{j=1}^{n-1} (j-1)k_{j+1}} \prod_{j=1}^{n-1} X_{j+1}^{k_{j+1}}}{X_0^{2\alpha+\sum_{j=1}^{n-1} j k_{j+1}}}=\nn\\
&&=M_0^{2\left(\alpha+\sum_{j=2}^n (j-1) k_j\right)}\frac{X_1^{2\alpha+\sum_{j=2}^n (j-2)k_j} \prod_{j=2}^n X_j^{k_j}}{X_0^{2\alpha+\sum_{j=2}^n (j-1) k_j}}~~,
\eeqa
which equals the left hand side of \eqref{monbeta} upon using relations \eqref{id1}, \eqref{id2} and \eqref{id3}.
\end{proof}

\begin{cor}
\label{cor:vsigma}
There exist unique polynomials $\sigma_n\in \Q[X_0,\ldots, X_n]$ with
$n\geq 1$ such that:
\be
M_0^{2n} v_n(X_0,\ldots, X_n)=\sigma_n(\beta_0(X_0,X_1),\beta_1(X_0,X_2),\beta_2(X_0,X_1,X_3)\ldots, \beta_{n-1}(X_0,X_1,X_{n}))~~\forall n\geq 1~~.
\ee
\end{cor}

\begin{definition}
The polynomials $(\sigma_n)_{n\geq 1}$ are called the {\em slow roll
  polynomials}.
\end{definition}

\begin{definition}
The {\em potential slow roll functions} of the model $(M_0,\Phi)$ are
the smooth functions $\upbeta_n\in \cC^\infty(\R)$ defined through:
\beqan
\label{upbetadef}
&& \upbeta_0(x)=\sigma_0(\Phi(x),\Phi'(x))=M_0^2 \frac{\Phi'(x)^2}{\Phi(x)^2}~~\nn\\
&& \upbeta_1(x)=\sigma_1(\Phi(x),\Phi''(x))=M_0^2 \frac{\Phi''(x)}{\Phi(x)}~~\\
&& \upbeta_j(x)=\sigma_n(\Phi(x),\Phi'(x),\Phi^{(j+1)}(x))=M_0^{2j} \frac{\Phi'(x)^{j-1} \Phi^{(j+1)}(x)}{\Phi(x)^j}~~\forall j\geq 2~~.\nn 
\eeqan
\end{definition}

\noindent Notice that reference \cite{LPB} prefers to work with the
          {\em reduced potential slow roll functions}:
\beqan
\label{upbetadefred}
&& ~^0\upbeta_0\eqdef \frac{3}{4}\upbeta_0\nn\\
&&~^j\upbeta\eqdef \frac{3}{2} \upbeta_j^{1/j}~~\forall j\geq 1~~.
\eeqan
It is customary to use the notations:
\be
\upepsilon\eqdef ~^0\upbeta_0~~,~~\upeta\eqdef ~^1\upbeta_1~~,~~~\upxi\eqdef ~~^2\upbeta_2~~,~~\upzeta\eqdef ~~^3\upbeta_3~~.
\ee

\noindent The following result shows in particular that the expansions
of the functions $\fH_\IR$ and $V_\IR$ coincide with their {\em slow
  roll expansions} in the sense of reference \cite{LPB}.

\begin{prop}
\label{prop:SR}
We have:
\beqan
\label{uSR}
&&u(x)=\fH_0(x)[1+\sum_{n\geq 1} \epsilon^{2n} \sigma_n(\upbeta_0(x),\ldots, \upbeta_{n-1}(x))]~~\\
&&V_\eff(x) =V_0(x)[1+\sum_{n\geq 1} \epsilon^{2n} \sigma_n(\upbeta_0(x),\ldots, \upbeta_{n-1}(x))]~~,
\eeqan
and:
\beqan
&&\fH_\IR(x)=\fH_0(x)[1+\sum_{n\geq 1} \sigma_n(\upbeta_0(x),\ldots, \upbeta_{n-1}(x))]~~\\
&& V_\IR(x) =V_0(x)[1+\sum_{n\geq 1} \sigma_n(\upbeta_0(x),\ldots, \upbeta_{n-1}(x))]~~,
\eeqan
where $\fH_0=\frac{\sqrt{2\Phi}}{M_0}$ and $V_0\eqdef M_0\sqrt{2\Phi}$.
\end{prop}

\begin{proof}
Follows immediately from Propositions \ref{prop:ExpPhi} and
\ref{prop:subs} by using Corollary \ref{cor:vsigma}.
\end{proof}

\noindent Using \eqref{upbetadef} in \eqref{Veff4} gives:
{\scriptsize \be
\left(\frac{V_\IR}{V_0}\right)=\frac{\upepsilon }{3}-\frac{\upepsilon
  ^2}{3}+\frac{25 \upepsilon ^3}{27}-\frac{109 \upepsilon ^4}{27}+\frac{2
  \upepsilon \upeta }{9}-\frac{26 \upepsilon ^2 \eta }{27}+\frac{460
  \upepsilon ^3 \upeta }{81}+\frac{5 \upepsilon \upeta ^2}{27}-\frac{172
  \upepsilon ^2 \upeta ^2}{81}+\frac{14 \upepsilon \upeta ^3}{81}+\frac{2
  \upepsilon \upxi ^2}{27}-\frac{44 \upepsilon ^2 \upxi ^2}{81}+\frac{2}{9}
\upepsilon \upeta \upxi ^2+\frac{2 \upepsilon \upsigma ^3}{81}+\O(M_0^{10})~~,
\ee}
\!\!which agrees with \cite[eq. (5.4)]{LPB}. Notice that
loc. cit. uses conventions in which the Planck mass (which is denoted
there by $m$) is related to $M_0$ through $m=2\sqrt{3\pi} M_0$. We
stress that the computational procedure proposed in \cite{LPB} is
rather impractical and far less efficient that the approach of the
present paper, which relies on the recursive construction of
Hamilton-Jacobi polynomials. We illustrate this in Appendix
\ref{app:list}, where we list the rational functions $v_n$ up to
$n=10$. These determine the expansions of $\fH_\IR$ and $V_\IR$ up to
that order through relations \eqref{fHIRexp} and
\eqref{VIRexp}.

\subsection{The $N$-th order approximate IR flow orbit}

\begin{definition}
Let $N\geq 0$. The {\em $N$-th order IR Hamilton-Jacobi function} and
{\em $N$-th order IR effective potential} $\fH_\IR^{(N)},V_\IR^{(N)}\in
\cC^\infty(\R)$ of the model $(M_0,\Phi)$ are the $N$-th partial sums
of the series \eqref{fHIRexp} and \eqref{VIRexp}:
\beqa
&&\fH_\IR^{(N)}\eqdef \fH_0\sum_{n= 0}^N M_0^{2n} (v_n\circ j^n)(\Phi)\nn\\
&&V_\IR^{(N)}\eqdef V_0\sum_{n= 0}^N M_0^{2n} (v_n\circ j^n)(\Phi)~~.
\eeqa
\end{definition}

\noindent Consider the smooth function $\v_\IR^{(N)}:\R\rightarrow \R$
defined through (cf. eq. \eqref{speed}):
\ben
\v_\IR^{(N)}(x)=-\frac{\dd V_\IR^{(N)}(x)}{\dd x}=-M_0^2\frac{\dd \fH_\IR^{(N)}(x)}{\dd x}~~(x\in \R)~~.
\een

\begin{definition}
The {\em $N$-th order approximate IR flow orbit} $\cO^{(N)}_\IR\subset
\R^2$ is the graph of the function $\v_\IR^{(N)}$:
\be
\cO^{(N)}_\IR\eqdef \{(x,\v_\IR^{(N)}(x))~\vert~x\in \R\}~~.
\ee
\end{definition}

\noindent The one-dimensional manifold $\cO^{(N)}$ can be viewed as an
approximate cosmological flow orbit. Then $\fH_\IR^{(N)}$ is its
approximate Hamilton-Jacobi function and $V_\IR^{(N)}$ its approximate
Hamilton-Jacobi potential (cf. Subsection \ref{subsec:HJpot}). This
one-manifold determines an approximate cosmological curve
$\varphi^{(N)}_\IR$ (up to a constant translation of the parameter
$t$) though the equation:
\ben
\label{speedN}
\dot{\varphi}_\IR^{(N)}(t)=\v_\IR^{(N)}(\varphi^{(N)}_\IR(t))=-\frac{\dd V_\IR^{(N)}}{\dd x}\bvert_{x=\varphi_\IR^{(N)}(t)}~~,
\een
which is the gradient flow equation of $V_\IR^{(N)}$. Away from the
critical points of $V_\IR^{(N)}$, the solutions of \eqref{speedN} are
obtained by inverting the function (cf. eq. \eqref{tHJ}):
\be
\t^{(N)}_\IR(x)=C-\int_{x_0}^x \frac{\dd y}{(V^{(N)}_\IR)'(y)}~~,
\ee
where $C,x_0\in \R$. 

\subsection{The formal IR Hamilton-Jacobi jet expansion family}

\noindent For all $n,k\in \N$, let $\Delta_{n,k}:\cA\rightarrow \cA$
be the differential operator defined through $\Delta_{n,0}\eqdef
\id_\cA$, $\Delta_{n,1}\eqdef D(2n-1)$ and:
\be
\Delta_{n,k}\eqdef D(2n+k-2) \ldots D(2n) D(2n-1)~~\forall k\geq 2~~.
\ee
We have:
\be
\Delta_{n,k+1}=D(2n+k-1) \Delta_{n,k}~~\forall n,k\geq 0~~.
\ee

\begin{prop}
\label{prop:unjet}
For any $n\geq 1$ and $k\geq 0$ we have:
\ben
\label{unk}
u_n^{(k)}(x)=M_0^{2n}\frac{(\Delta_{n,k} Q_n)(j^{n+k}_x(u_0))}{u_0^{2n+k-1}(x)} ~~ \forall x\in \R
\een
and:
\ben
\label{unkPhi}
u_n^{(k)}(x)=M_0^{2n-1}\sqrt{2\Phi} (\Delta_{n,k} Q_n)(g_{n+k}(j^{n+k}_x(\Phi)))~~ \forall x\in \R~~.
\een
\end{prop}

\begin{proof}
Relation \eqref{unk} follows immediately from Proposition
\ref{prop:unpoly} upon using \eqref{Dfrack}. Now \eqref{unk} implies
\eqref{unkPhi} upon using Corollary \ref{cor:jsqrt} and the fact that
the polynomial $\Delta_{n,k} Q_n$ is homogeneous of degree $2n+k$ with
respect to the polynomial grading.
\end{proof}

\begin{cor}
\label{cor:IRjet}
For any $n\geq 0$, we have:
\be
j^\infty_x(u_n)=M_0^{2n-1}\sqrt{2\Phi(x)}e_n(j^\infty_x(\Phi))~~\forall x\in \R~~,
\ee
where $e_n:\R_0^\N\rightarrow \R^\N$ are rational functions of bidegree
$(0,2n+k)$ with components:
\be
e_n^k=(\Delta_{n,k} Q_n)\circ g_{n+k}\circ \tau_{n+k}~~\forall k\geq 0~~.
\ee
\end{cor}

\begin{definition}
Define the {\em formal jet} of $u$ at $x\in \R$ through:
\be
j_x^\infty(u)\eqdef \sum_{n\geq 0} j_x^\infty(u_n)\otimes_\R \epsilon^n\in \R^\N[[\epsilon]]~~.
\ee
\end{definition}

\begin{definition}
The {\em formal IR Hamilton-Jacobi jet expansion family} of the model
$(M_0,\Phi)$ is the map $\re:\R_0^\N\rightarrow \R^\N[[\epsilon]]$
defined through:
\ben
\label{e}
\re(w)=\sum_{n\geq 0} M_0^{2n-1} \sqrt{2w_0} e_n(w) \otimes_\R \epsilon^n~~\forall w\in \R_0^\N~~.
\een
\end{definition}

\noindent Notice that $\re$ is positively homogeneous of degree
$1/2$ (where we take $\deg(\epsilon)=0$):
\be
\re(\xi w)=\xi^{1/2}\re(w)~~\forall \xi>0~~\forall w\in \R_0^\N~~.
\ee
Corollary \ref{cor:IRjet} implies:

\begin{cor}
\label{cor:jetIR}
For any $x\in \R$, we have: 
\ben
\label{erel}
j^\infty_x(u)=\re(j^\infty_x(\Phi))~~.
\een
\end{cor}

\begin{prop}
\label{prop:jeteps}
For all $x\in \R$ and $n\geq 0$ we have:
\ben
\label{jeteqneps}
\Phi^{(n)}(x)=\rP_n(j^{n+1}_x(u))~~,
\een
where $\rP_n\in \R[\epsilon][X_0,\ldots,
  X_{n+1}]=\R[\epsilon,X_0,\ldots,X_{n+1}]$ is the polynomial given
by:
\ben
\label{Pneps}
\rP_n\eqdef \frac{M_0^2}{2}\sum_{k+l=n}\frac{n!}{k!l!}\left(X_k X_l-\epsilon^2 M_0^2 X_{k+1}X_{l+1} \right)~~.
\een
Moreover, we have:
\ben
\label{jeteqeps}
j^\infty_x(\Phi) =\uprho(j^\infty_x(u))~~\forall x\in ~\R,
\een
where $\uprho:\R^\N[[\epsilon]]\rightarrow \R^\N[[\epsilon]]$ is the
quadratic map given by:
\be
\uprho(w)=(\rP_0(w_{\leq 1}), \rP_1(w_{\leq 2}),\ldots, \rP_n(w_{\leq n+1}),\ldots)~~\forall w=(w_k)_{k\geq 0}\in \R^\N[[\epsilon]]~~.
\ee
\end{prop}

\noindent For example, we have:
\beqa
&& \frac{\Phi}{M_0^2}=\frac{1}{2}\left[u^2-\epsilon^2 M_0^2(u')^2\right]~~\nn\\
&& \frac{\Phi'}{M_0^2}=u'u-\epsilon^2 M_0^2u'u''~~\\
&& \frac{\Phi''}{M_0^2}=u'' u +(u')^2-\epsilon^2 M_0^2 (u''' u' + (u'')^2)~~\nn
\eeqa
and hence:
\beqan
&&\rP_0[X_0,X_1]=\frac{M_0^2}{2}\left(X_0^2-\epsilon^2 M_0^2 X_1^2\right)\nn\\
&&\rP_1[X_0,X_1,X_2]=M_0^2(X_0X_1-\epsilon^2 M_0^2 X_1 X_2)\\
&&\rP_2[X_0,X_1,X_2,X_3]=M_0^2\left[X_0X_2+X_1^2-\epsilon^2 M_0^2 (X_1X_3+X_2^2)\right]~~.\nn
\eeqan

\begin{proof}
Recall that $u$ satisfies \eqref{Phieps}. Repeatedly differentiating
this equation with respect to $x$ allows us to express
$\frac{\Phi^{(n)}}{M_0^2}$ as a function of $\epsilon$ and
the derivatives $u^{(j)}$ with $j=0,\ldots, n+1$:
\ben
\label{Phineps}
\frac{\Phi^{(n)}(x)}{M_0^2}=\rP_n(u^{(0)}(x),u^{(1)}(x),\ldots, u^{(n+1)}(x))~~\forall n\geq 0~~,
\een
where $P_n\in \R[\epsilon, X_0,\ldots, X_{n+1}]$ is a polynomial. The
higher order differentiation formula for functions $f,g\in
\cC^\infty(\R)$:
\be
(fg)^{(n)}=\sum_{k+l=n} \frac{(k+l)!}{k!l!}f^{(k)} g^{(l)}~~\forall n\geq 0
\ee
gives:
\be
\frac{\dd^n}{\dd x^n}(u^2)=\sum_{k+l=n} \frac{n!}{k!l!} u^{(k)}u^{(l)}~~,~~\frac{\dd^n}{\dd x^n}[(u')^2]=\sum_{k+l=n; k,l\geq 1} \frac{n!}{k!l!} u^{(k+1)} u^{(l+1)} 
\ee
and hence:
\be
\frac{\Phi^{(n)}}{M_0^2}=\frac{1}{2}\sum_{k+l=n}\frac{n!}{k!l!}\left(u^{(k)} u^{(l)} -\epsilon^2 M_0^2u^{(k+1)} u^{(l+1)} \right)~~,
\ee
which implies \eqref{Pneps}. The remaining statements are obvious.
\end{proof}

\begin{cor}
The formal IR Hamilton-Jacobi jet expansion family satisfies:
\be
\uprho\circ \re=\iota~~,
\ee
where $\iota:\R_0^\N\hookrightarrow \R^N[[\epsilon]]$ is the injective
map given by:
\be
\iota(w)\eqdef w\in \R_0^\N\subset \R^\N[[\epsilon]]~~\forall w\in \R_0^\N~~.
\ee
\end{cor}

\begin{proof}
Follows immediately from Proposition \ref{prop:jeteps} and Corollary
\ref{cor:jetIR} using the fact that $\Phi:\R\rightarrow \R_{>0}$
is an arbitrary positive smooth function.
\end{proof}

\section{Conclusions and further directions}
\label{sec:conclusions}

Following the ideas of \cite{ren,grad}, we constructed the infrared
scale expansion of single field cosmological models using the
Hamilton-Jacobi formalism of \cite{SB}, finding that
it corresponds to seeking a solution of the Hamilton-Jacobi equation
of such models which admits a Laurent expansion in powers of the
Planck mass. We showed that this recovers the celebrated slow roll
expansion of Liddle, Parsons and Barrow \cite{LPB}, which it
identifies with the small Planck mass expansion of such models. We
also described the single field realization of the universal
similarity group action of \cite{ren, grad} both at the level of
cosmological curves and in the Hamilton-Jacobi formalism.

Unlike the approach of \cite{LPB}, the method of the present paper
leads to an explicit recursive construction of the coefficients of the
slow roll expansion. We showed that the latter are obtained by
evaluating certain model-independent rational functions $v_n$ on the
scalar potential of the model and its successive derivatives. In turn,
these rational functions are obtained by explicit universal formulas
from the {\em Hamilton-Jacobi polynomials} $Q_n$. The latter are
model-independent multivariate polynomials with rational coefficients
defined by a recursion relation which allows for their efficient
computation. As an application, we listed the coefficients of the
Hamilton-Jacobi slow roll/small Planck mass expansion up to order 10
inclusively (see Appendix \ref{app:list}). The identification of the
small Planck mass and slow roll expansions follows from the fact that
the polynomials $Q_n$ are bihomogeneous with respect to a certain
bigrading of the algebra of polynomials in a countable number of
variables.

The present paper opens up a few avenues for further research.  First,
one could perform a systematic study of Hamilton-Jacobi polynomials
and their deeper properties, on which we only touched upon in the
present paper. This would make it possible to extract asymptotic
bounds for the small Planck mass/slow roll expansion which would
enable one to analyze its summability and resummation. In this regard,
we notice that the deformed Hamilton-Jacobi equation considered in the
present paper can be approached with the well-established methods of
singular perturbation theory for ODEs with a small
parameter. Second, one can apply similar methods to the UV expansion
of \cite{ren}, which in the single field case should be related to the
psi series expansion discussed in \cite{MA,HLT}. Finally, one can
study the IR and UV expansions proposed in \cite{ren} for multifield
cosmological models, which should provide a generalization of the slow
roll and psi series expansions to the multifield case. We hope to
report on these and related questions in future work.

\acknowledgments{\noindent This work was supported by grant PN
  19060101/2019-2022. The author thanks the Erwin Schrodinger
  Institute for hospitality. }

\appendix

\section{Jets of univariate real-valued functions}
\label{app:jets}

\noindent Let $\cC_x^\infty$ be the commutative
$\R$-algebra of germs of smooth univariate real-valued functions at a
point $x\in \R$. Let $\id_\R,1_\R\in \cC^\infty(\R)$ be the identity
and unit functions of $\R$ and $e_x\in \cC^\infty(\R)$ be the germ at
$x$ of the function $\id_\R-x 1_\R\in \cC^\infty(\R)$. Then
$\cC_x^\infty$ is a (non-Noetherian) commutative local ring with maximal ideal:
\be
m_x=(e_x)=\{f\in \cC_x^\infty~\vert~f(x)=0\}~~
\ee
and residue field $\R$ and for all $n\in \N$ we have:
\be
m_x^{n+1}=\{f\in \cC_x^\infty~\vert~f^{(k)}(x)=0~~\forall k=0,\ldots, n\}~~.
\ee
Define:
\be
m_x^\infty\eqdef \cap_{n\geq 0} {m_x^{n+1}}=\{f\in \cC_x^\infty~\vert~f^{(k)}(x)=0~~\forall k\geq 0\}~~.
\ee
For all $\nu\in \N\cup\{\infty\}$, let $J_x^\nu$ be the commutative
$\R$-algebra of jets of order $\nu$ of univariate real-valued functions
at $x$ and $\epsilon_x \in J_x^\infty$ be the infinite order jet of
the identity function $\id_\R$. Then $J_x^\infty$ is a Noetherian
local ring with maximal ideal:
\be
\mu_x=(\epsilon_x)
\ee
which identifies with the $\R$-algebra $\R[[\epsilon_x]]$ of formal
power series in the variable $\epsilon_x$. Accordingly, $J_x^n$ identifies
with the space $\R[[\epsilon_x]]//\mu_x^{n+1}$ of polynomials of degree $n$ in the
variable $\epsilon_x$. The infinite jet
$\rj_x^\infty(f)$ of a germ $f\in \cC_x^\infty$ identifies with the
formal Taylor series of $f$ at $x$, while $\rj^n_x(f)$ identifies with
its $n$-th Taylor polynomial:
\be
\rj_x^\infty(f)\equiv \sum_{k=0}^\infty \frac{f^{(k)}(x)}{k!}\epsilon_x^k~~,~~\rj_x^n(f)\equiv \sum_{k=0}^n \frac{f^{(k)}(x)}{k!}\epsilon_x^k~~.
\ee
The higher order differentiation formula for functions $f,g\in \cC^\infty(\R)$:
\be
(fg)^{(n)}=\sum_{k+l=n} \frac{(k+l)!}{k!l!}f^{(k)} g^{(l)}~~\forall n\geq 0
\ee
implies that infinite jet prolongation gives a morphism of local rings:
\be
\rj_x^\infty:\cC^\infty_x\rightarrow J_x^\infty~~,~~
\ee
which satisfies $\rj_x^\infty(e_x)=\epsilon_x$. This morphism is surjective by Borel's lemma
and induces isomorphisms of rings:
\be
\rj_x^\nu:\cC^\infty_x/m_x^{\nu+1}\stackrel{\sim}{\rightarrow} J_x^\nu~~\forall \nu\in \N\cup\{\infty\}~~.
\ee

\section{The polynomials $Q_n$ and rational functions $v_n$ for $n\leq 10$}
\label{app:list}

Recall from Section \ref{sec:deformed} that the slow roll expansions
of $\fH_\IR$ and $V_\IR$ are obtained from the Planck mass expansions
of Corollary \ref{cor:IRexpansions} by changing variables to the slow
roll functions \eqref{betadef} (or \eqref{upbetadefred}) using Lemma
\ref{lemma:monSR}. Below, we list the Hamilton-Jacobi polynomials
$Q_n$ and the rational functions $v_n$ which control the expansions
\eqref{fHIRexp} and \eqref{VIRexp} for all $n\leq 10$. The
corresponding coefficients of the expansions are obtained by
substituting $X_j\rightarrow \Phi^{(j)}$ into $v_n$.

\subsection{The Hamilton-Jacobi polynomials $Q_n$ for $n\leq 10$}

\noindent {\footnotesize
\begin{itemize}
\item $Q_0=1$    
\item $ Q_1=\frac{X_1^2}{2} $
\item $ Q_2=-\frac{5 X_1^4}{8}+X_0 X_1^2 X_2 $
\item $ Q_3=\frac{37 X_1^6}{16}-\frac{11}{2} X_0 X_1^4 X_2+\frac{5}{2}
  X_0^2 X_1^2 X_2^2+X_0^2 X_1^3 X_3 $
\item $ Q_4=-\frac{1773 X_1^8}{128}+\frac{347}{8} X_0 X_1^6
  X_2-\frac{147}{4} X_0^2 X_1^4 X_2^2+7 X_0^3 X_1^2 X_2^3-\frac{19}{2}
  X_0^2 X_1^5 X_3+9 X_0^3 X_1^3 X_2 X_3+X_0^3 X_1^4 X_4 $
\item $ Q_5=\frac{28895 X_1^{10}}{256}-\frac{6983}{16} X_0 X_1^8
  X_2+\frac{8543}{16} X_0^2 X_1^6 X_2^2-222 X_0^3 X_1^4 X_2^3+21 X_0^4
  X_1^2 X_2^4+\frac{819}{8} X_0^2 X_1^7 X_3-\frac{359}{2} X_0^3 X_1^5
  X_2 X_3+58 X_0^4 X_1^3 X_2^2 X_3+\frac{19}{2} X_0^4 X_1^4
  X_3^2-\frac{29}{2} X_0^3 X_1^6 X_4+14 X_0^4 X_1^4 X_2 X_4+X_0^4
  X_1^5 X_5 $
\item $ Q_6=-\frac{1181457 X_1^{12}}{1024}+\frac{677619}{128} X_0
  X_1^{10} X_2-\frac{269319}{32} X_0^2 X_1^8 X_2^2+\frac{43575}{8}
  X_0^3 X_1^6 X_2^3-\frac{10217}{8} X_0^4 X_1^4 X_2^4+66 X_0^5 X_1^2
  X_2^5-\frac{20567}{16} X_0^2 X_1^9 X_3+\frac{25843}{8} X_0^3 X_1^7
  X_2 X_3-\frac{4315}{2} X_0^4 X_1^5 X_2^2 X_3+327 X_0^5 X_1^3 X_2^3
  X_3-\frac{981}{4} X_0^4 X_1^6 X_3^2+171 X_0^5 X_1^4 X_2
  X_3^2+\frac{1647}{8} X_0^3 X_1^8 X_4-369 X_0^4 X_1^6 X_2 X_4+129
  X_0^5 X_1^4 X_2^2 X_4+34 X_0^5 X_1^5 X_3 X_4-\frac{41}{2} X_0^4
  X_1^7 X_5+20 X_0^5 X_1^5 X_2 X_5+X_0^5 X_1^6 X_6 $
\item $ Q_7=\frac{28933037 X_1^{14}}{2048}-\frac{19165717}{256} X_0
  X_1^{12} X_2+\frac{37359855}{256} X_0^2 X_1^{10}
  X_2^2-\frac{1026471}{8} X_0^3 X_1^8 X_2^3+\frac{799577}{16} X_0^4
  X_1^6 X_2^4-\frac{14301}{2} X_0^5 X_1^4 X_2^5+\frac{429}{2} X_0^6
  X_1^2 X_2^6+\frac{2374083}{128} X_0^2 X_1^{11} X_3-\frac{963939}{16}
  X_0^3 X_1^9 X_2 X_3+\frac{122417}{2} X_0^4 X_1^7 X_2^2 X_3-21075
  X_0^5 X_1^5 X_2^3 X_3+1719 X_0^6 X_1^3 X_2^4 X_3+\frac{85481}{16}
  X_0^4 X_1^8 X_3^2-7486 X_0^5 X_1^6 X_2 X_3^2+1917 X_0^6 X_1^4 X_2^2
  X_3^2+180 X_0^6 X_1^5 X_3^3-\frac{50853}{16} X_0^3 X_1^{10}
  X_4+\frac{32609}{4} X_0^4 X_1^8 X_2 X_4-5735 X_0^5 X_1^6 X_2^2
  X_4+988 X_0^6 X_1^4 X_2^3 X_4-1120 X_0^5 X_1^7 X_3 X_4+824 X_0^6
  X_1^5 X_2 X_3 X_4+\frac{69}{2} X_0^6 X_1^6 X_4^2+\frac{2975}{8}
  X_0^4 X_1^9 X_5-677 X_0^5 X_1^7 X_2 X_5+250 X_0^6 X_1^5 X_2^2 X_5+55
  X_0^6 X_1^6 X_3 X_5-\frac{55}{2} X_0^5 X_1^8 X_6+27 X_0^6 X_1^6 X_2
  X_6+X_0^6 X_1^7 X_7 $
\item $ Q_8=-\frac{6591947773 X_1^{16}}{32768}+\frac{1236367727 X_0
  X_1^{14} X_2}{1024}-\frac{1424076061}{512} X_0^2 X_1^{12}
  X_2^2+\frac{395323017}{128} X_0^3 X_1^{10}
  X_2^3-\frac{108856779}{64} X_0^4 X_1^8 X_2^4+\frac{3429947}{8} X_0^5
  X_1^6 X_2^5-\frac{157549}{4} X_0^6 X_1^4 X_2^6+715 X_0^7 X_1^2
  X_2^7-\frac{77550429}{256} X_0^2 X_1^{13} X_3+\frac{153843435}{128}
  X_0^3 X_1^{11} X_2 X_3-\frac{26217965}{16} X_0^4 X_1^9 X_2^2
  X_3+\frac{7241243}{8} X_0^5 X_1^7 X_2^3 X_3-\frac{365341}{2} X_0^6
  X_1^5 X_2^4 X_3+8665 X_0^7 X_1^3 X_2^5 X_3-\frac{3720065}{32} X_0^4
  X_1^{10} X_3^2+\frac{1954939}{8} X_0^5 X_1^8 X_2
  X_3^2-\frac{541547}{4} X_0^6 X_1^6 X_2^2 X_3^2+17204 X_0^7 X_1^4
  X_2^3 X_3^2-\frac{18469}{2} X_0^6 X_1^7 X_3^3+5131 X_0^7 X_1^5 X_2
  X_3^3+\frac{6933787}{128} X_0^3 X_1^{12} X_4-\frac{1433753}{8} X_0^4
  X_1^{10} X_2 X_4+\frac{1513777}{8} X_0^5 X_1^8 X_2^2 X_4-70445 X_0^6
  X_1^6 X_2^3 X_4+6821 X_0^7 X_1^4 X_2^4 X_4+\frac{117565}{4} X_0^5
  X_1^9 X_3 X_4-42828 X_0^6 X_1^7 X_2 X_3 X_4+12020 X_0^7 X_1^5 X_2^2
  X_3 X_4+1406 X_0^7 X_1^6 X_3^2 X_4-\frac{5639}{4} X_0^6 X_1^8
  X_4^2+1078 X_0^7 X_1^6 X_2 X_4^2-\frac{111185}{16} X_0^4 X_1^{11}
  X_5+18117 X_0^5 X_1^9 X_2 X_5-\frac{26545}{2} X_0^6 X_1^7 X_2^2
  X_5+2510 X_0^7 X_1^5 X_2^3 X_5-\frac{4531}{2} X_0^6 X_1^8 X_3
  X_5+1735 X_0^7 X_1^6 X_2 X_3 X_5+125 X_0^7 X_1^7 X_4
  X_5+\frac{4971}{8} X_0^5 X_1^{10} X_6-\frac{2289}{2} X_0^6 X_1^8 X_2
  X_6+440 X_0^7 X_1^6 X_2^2 X_6+83 X_0^7 X_1^7 X_3 X_6-\frac{71}{2}
  X_0^6 X_1^9 X_7+35 X_0^7 X_1^7 X_2 X_7+X_0^7 X_1^8 X_8 $
\item $ Q_9=\frac{213861782587 X_1^{18}}{65536}-\frac{44761605983 X_0
  X_1^{16} X_2}{2048}+\frac{118797886803 X_0^2 X_1^{14}
  X_2^2}{2048}-\frac{622034129}{8} X_0^3 X_1^{12}
  X_2^3+\frac{3557208131}{64} X_0^4 X_1^{10} X_2^4-\frac{82240313}{4}
  X_0^5 X_1^8 X_2^5+\frac{14060677}{4} X_0^6 X_1^6 X_2^6-214728 X_0^7
  X_1^4 X_2^7+2431 X_0^8 X_1^2 X_2^8+\frac{5662417463 X_0^2 X_1^{15}
    X_3}{1024}-\frac{6619146041}{256} X_0^3 X_1^{13} X_2
  X_3+\frac{2828141365}{64} X_0^4 X_1^{11} X_2^2
  X_3-\frac{135618571}{4} X_0^5 X_1^9 X_2^3 X_3+\frac{23060981}{2}
  X_0^6 X_1^7 X_2^4 X_3-1465806 X_0^7 X_1^5 X_2^5 X_3+42496 X_0^8
  X_1^3 X_2^6 X_3+\frac{676563673}{256} X_0^4 X_1^{12}
  X_3^2-\frac{59158947}{8} X_0^5 X_1^{10} X_2 X_3^2+\frac{25706377}{4}
  X_0^6 X_1^8 X_2^2 X_3^2-1892114 X_0^7 X_1^6 X_2^3 X_3^2+135438 X_0^8
  X_1^4 X_2^4 X_3^2+\frac{1379725}{4} X_0^6 X_1^9 X_3^3-402698 X_0^7
  X_1^7 X_2 X_3^3+85672 X_0^8 X_1^5 X_2^2 X_3^3+\frac{10685}{2} X_0^8
  X_1^6 X_3^4-\frac{260611419}{256} X_0^3 X_1^{14}
  X_4+\frac{262801405}{64} X_0^4 X_1^{12} X_2 X_4-\frac{23086241}{4}
  X_0^5 X_1^{10} X_2^2 X_4+3367779 X_0^6 X_1^8 X_2^3 X_4-752572 X_0^7
  X_1^6 X_2^4 X_4+44100 X_0^8 X_1^4 X_2^5 X_4-\frac{2979009}{4} X_0^5
  X_1^{11} X_3 X_4+1614434 X_0^6 X_1^9 X_2 X_3 X_4-950720 X_0^7 X_1^7
  X_2^2 X_3 X_4+137052 X_0^8 X_1^5 X_2^3 X_3 X_4-\frac{169801}{2}
  X_0^7 X_1^8 X_3^2 X_4+50538 X_0^8 X_1^6 X_2 X_3^2
  X_4+\frac{701519}{16} X_0^6 X_1^{10} X_4^2-\frac{131719}{2} X_0^7
  X_1^8 X_2 X_4^2+\frac{39571}{2} X_0^8 X_1^6 X_2^2 X_4^2+3957 X_0^8
  X_1^7 X_3 X_4^2+\frac{17743131}{128} X_0^4 X_1^{13}
  X_5-\frac{3727327}{8} X_0^5 X_1^{11} X_2 X_5+\frac{1016195}{2} X_0^6
  X_1^9 X_2^2 X_5-201466 X_0^7 X_1^7 X_2^3 X_5+22176 X_0^8 X_1^5 X_2^4
  X_5+\frac{567145}{8} X_0^6 X_1^{10} X_3 X_5-\frac{213375}{2} X_0^7
  X_1^8 X_2 X_3 X_5+32151 X_0^8 X_1^6 X_2^2 X_3 X_5+3204 X_0^8 X_1^7
  X_3^2 X_5-\frac{12487}{2} X_0^7 X_1^9 X_4 X_5+4922 X_0^8 X_1^7 X_2
  X_4 X_5+\frac{251}{2} X_0^8 X_1^8 X_5^2-\frac{221603}{16} X_0^5
  X_1^{12} X_6+\frac{292737}{8} X_0^6 X_1^{10} X_2 X_6-\frac{55379}{2}
  X_0^7 X_1^8 X_2^2 X_6+5619 X_0^8 X_1^6 X_2^3 X_6-\frac{8385}{2}
  X_0^7 X_1^9 X_3 X_6+3312 X_0^8 X_1^7 X_2 X_3 X_6+209 X_0^8 X_1^8 X_4
  X_6+\frac{7827}{8} X_0^6 X_1^{11} X_7-\frac{3637}{2} X_0^7 X_1^9 X_2
  X_7+721 X_0^8 X_1^7 X_2^2 X_7+119 X_0^8 X_1^8 X_3 X_7-\frac{89}{2}
  X_0^7 X_1^{10} X_8+44 X_0^8 X_1^8 X_2 X_8+X_0^8 X_1^9 X_9 $
\item $ Q_{10}=-\frac{15567784510059
  X_1^{20}}{262144}+\frac{14377260617795 X_0 X_1^{18}
  X_2}{32768}-\frac{5393824549071 X_0^2 X_1^{16}
  X_2^2}{4096}+\frac{2116195323375 X_0^3 X_1^{14}
  X_2^3}{1024}-\frac{1862997228891 X_0^4 X_1^{12}
  X_2^4}{1024}+\frac{28683710727}{32} X_0^5 X_1^{10}
  X_2^5-\frac{1863626829}{8} X_0^6 X_1^8 X_2^6+\frac{55821327}{2}
  X_0^7 X_1^6 X_2^7-\frac{4650921}{4} X_0^8 X_1^4 X_2^8+8398 X_0^9
  X_1^2 X_2^9-\frac{228637979167 X_0^2 X_1^{17}
    X_3}{2048}+\frac{614576237127 X_0^3 X_1^{15} X_2
    X_3}{1024}-\frac{315433939225}{256} X_0^4 X_1^{13} X_2^2
  X_3+\frac{155422675633}{128} X_0^5 X_1^{11} X_2^3
  X_3-\frac{4739625715}{8} X_0^6 X_1^9 X_2^4 X_3+132974554 X_0^7 X_1^7
  X_2^5 X_3-11151514 X_0^8 X_1^5 X_2^6 X_3+204492 X_0^9 X_1^3 X_2^7
  X_3-\frac{32623871227}{512} X_0^4 X_1^{14}
  X_3^2+\frac{28460434021}{128} X_0^5 X_1^{12} X_2
  X_3^2-\frac{8482334789}{32} X_0^6 X_1^{10} X_2^2
  X_3^2+\frac{510649451}{4} X_0^7 X_1^8 X_2^3 X_3^2-22602272 X_0^8
  X_1^6 X_2^4 X_3^2+978132 X_0^9 X_1^4 X_2^5
  X_3^2-\frac{187838583}{16} X_0^6 X_1^{11} X_3^3+\frac{169442771}{8}
  X_0^7 X_1^9 X_2 X_3^3-10019604 X_0^8 X_1^7 X_2^2 X_3^3+1095280 X_0^9
  X_1^5 X_2^3 X_3^3-\frac{3747879}{8} X_0^8 X_1^8 X_3^4+216018 X_0^9
  X_1^6 X_2 X_3^4+\frac{21490971643 X_0^3 X_1^{16}
    X_4}{1024}-\frac{12750162103}{128} X_0^4 X_1^{14} X_2
  X_4+\frac{22347970303}{128} X_0^5 X_1^{12} X_2^2
  X_4-\frac{1117333047}{8} X_0^6 X_1^{10} X_2^3 X_4+50945645 X_0^7
  X_1^8 X_2^4 X_4-7333184 X_0^8 X_1^6 X_2^5 X_4+272676 X_0^9 X_1^4
  X_2^6 X_4+\frac{1233663527}{64} X_0^5 X_1^{13} X_3 X_4-55321742
  X_0^6 X_1^{11} X_2 X_3 X_4+50294115 X_0^7 X_1^9 X_2^2 X_3
  X_4-16061372 X_0^8 X_1^7 X_2^3 X_3 X_4+1347128 X_0^9 X_1^5 X_2^4 X_3
  X_4+3636387 X_0^7 X_1^{10} X_3^2 X_4-4458696 X_0^8 X_1^8 X_2 X_3^2
  X_4+1045352 X_0^9 X_1^6 X_2^2 X_3^2 X_4+73740 X_0^9 X_1^7 X_3^3
  X_4-\frac{40950891}{32} X_0^6 X_1^{12} X_4^2+2844442 X_0^7 X_1^{10}
  X_2 X_4^2-1758348 X_0^8 X_1^8 X_2^2 X_4^2+278810 X_0^9 X_1^6 X_2^3
  X_4^2-278425 X_0^8 X_1^9 X_3 X_4^2+174738 X_0^9 X_1^7 X_2 X_3
  X_4^2+3991 X_0^9 X_1^8 X_4^3-\frac{762737935}{256} X_0^4 X_1^{15}
  X_5+\frac{390380533}{32} X_0^5 X_1^{13} X_2 X_5-\frac{281870691}{16}
  X_0^6 X_1^{11} X_2^2 X_5+\frac{43106169}{4} X_0^7 X_1^9 X_2^3
  X_5-2618948 X_0^8 X_1^7 X_2^4 X_5+180280 X_0^9 X_1^5 X_2^5
  X_5-\frac{33267771}{16} X_0^6 X_1^{12} X_3 X_5+\frac{37046737}{8}
  X_0^7 X_1^{10} X_2 X_3 X_5-2871173 X_0^8 X_1^8 X_2^2 X_3 X_5+457440
  X_0^9 X_1^6 X_2^3 X_3 X_5-\frac{453409}{2} X_0^8 X_1^9 X_3^2
  X_5+142820 X_0^9 X_1^7 X_2 X_3^2 X_5+\frac{1826639}{8} X_0^7
  X_1^{11} X_4 X_5-352021 X_0^8 X_1^9 X_2 X_4 X_5+111780 X_0^9 X_1^7
  X_2^2 X_4 X_5+19455 X_0^9 X_1^8 X_3 X_4 X_5-\frac{30073}{4} X_0^8
  X_1^{10} X_5^2+6070 X_0^9 X_1^8 X_2 X_5^2+\frac{40999155}{128} X_0^5
  X_1^{14} X_6-\frac{17461149}{16} X_0^6 X_1^{12} X_2
  X_6+\frac{4892157}{4} X_0^7 X_1^{10} X_2^2 X_6-510501 X_0^8 X_1^8
  X_2^3 X_6+62008 X_0^9 X_1^6 X_2^4 X_6+\frac{1237553}{8} X_0^7
  X_1^{11} X_3 X_6-239099 X_0^8 X_1^9 X_2 X_3 X_6+76220 X_0^9 X_1^7
  X_2^2 X_3 X_6+6607 X_0^9 X_1^8 X_3^2 X_6-\frac{25133}{2} X_0^8
  X_1^{10} X_4 X_6+10153 X_0^9 X_1^8 X_2 X_4 X_6+461 X_0^9 X_1^9 X_5
  X_6-\frac{410763}{16} X_0^6 X_1^{13} X_7+\frac{548633}{8} X_0^7
  X_1^{11} X_2 X_7-53281 X_0^8 X_1^9 X_2^2 X_7+11424 X_0^9 X_1^7 X_2^3
  X_7-\frac{14489}{2} X_0^8 X_1^{10} X_3 X_7+5866 X_0^9 X_1^8 X_2 X_3
  X_7+329 X_0^9 X_1^9 X_4 X_7+\frac{11759}{8} X_0^7 X_1^{12} X_8-2752
  X_0^8 X_1^{10} X_2 X_8+1118 X_0^9 X_1^8 X_2^2 X_8+164 X_0^9 X_1^9
  X_3 X_8-\frac{109}{2} X_0^8 X_1^{11} X_9+54 X_0^9 X_1^9 X_2
  X_9+X_0^9 X_1^{10} X_{10} $
\end{itemize}}

\subsection{The rational functions $v_n$ for $n\leq 10$}

\noindent {\footnotesize \begin{itemize}
\item $ v_0=1 $

\item $ v_1=\frac{X_1^2}{8 X_0^2} $

\item $ v_2=-\frac{13 X_1^4}{128 X_0^4}+\frac{X_1^2 X_2}{8 X_0^3} $

\item $ v_3=\frac{213 X_1^6}{1024 X_0^6}-\frac{27 X_1^4 X_2}{64
  X_0^5}+\frac{5 X_1^2 X_2^2}{32 X_0^4}+\frac{X_1^3 X_3}{16 X_0^4} $

\item $ v_4=-\frac{21949 X_1^8}{32768 X_0^8}+\frac{1907 X_1^6
  X_2}{1024 X_0^7}-\frac{351 X_1^4 X_2^2}{256 X_0^6}+\frac{7 X_1^2
  X_2^3}{32 X_0^5}-\frac{45 X_1^5 X_3}{128 X_0^6}+\frac{9 X_1^3 X_2
  X_3}{32 X_0^5}+\frac{X_1^4 X_4}{32 X_0^5} $

\item $ v_5=\frac{758751 X_1^{10}}{262144 X_0^{10}}-\frac{83271 X_1^8
  X_2}{8192 X_0^9}+\frac{45671 X_1^6 X_2^2}{4096 X_0^8}-\frac{261
  X_1^4 X_2^3}{64 X_0^7}+\frac{21 X_1^2 X_2^4}{64 X_0^6}+\frac{4327
  X_1^7 X_3}{2048 X_0^8}-\frac{837 X_1^5 X_2 X_3}{256 X_0^7}+\frac{29
  X_1^3 X_2^2 X_3}{32 X_0^6}+\frac{19 X_1^4 X_3^2}{128 X_0^6}-\frac{67
  X_1^6 X_4}{256 X_0^7}+\frac{7 X_1^4 X_2 X_4}{32 X_0^6}+\frac{X_1^5
  X_5}{64 X_0^6} $

\item $ v_6=-\frac{65215177 X_1^{12}}{4194304 X_0^{12}}+\frac{17242595
  X_1^{10} X_2}{262144 X_0^{11}}-\frac{3131859 X_1^8 X_2^2}{32768
  X_0^{10}}+\frac{228839 X_1^6 X_2^3}{4096 X_0^9}-\frac{23801 X_1^4
  X_2^4}{2048 X_0^8}+\frac{33 X_1^2 X_2^5}{64 X_0^7}-\frac{235941
  X_1^9 X_3}{16384 X_0^{10}}+\frac{134103 X_1^7 X_2 X_3}{4096
  X_0^9}-\frac{9965 X_1^5 X_2^2 X_3}{512 X_0^8}+\frac{327 X_1^3 X_2^3
  X_3}{128 X_0^7}-\frac{2249 X_1^6 X_3^2}{1024 X_0^8}+\frac{171 X_1^4
  X_2 X_3^2}{128 X_0^7}+\frac{8427 X_1^8 X_4}{4096 X_0^9}-\frac{211
  X_1^6 X_2 X_4}{64 X_0^8}+\frac{129 X_1^4 X_2^2 X_4}{128
  X_0^7}+\frac{17 X_1^5 X_3 X_4}{64 X_0^7}-\frac{93 X_1^7 X_5}{512
  X_0^8}+\frac{5 X_1^5 X_2 X_5}{32 X_0^7}+\frac{X_1^6 X_6}{128 X_0^7}
  $

\item $ v_7=\frac{3334526813 X_1^{14}}{33554432
  X_0^{14}}-\frac{1029318213 X_1^{12} X_2}{2097152
  X_0^{13}}+\frac{929166495 X_1^{10} X_2^2}{1048576
  X_0^{12}}-\frac{11726215 X_1^8 X_2^3}{16384 X_0^{11}}+\frac{4148465
  X_1^6 X_2^4}{16384 X_0^{10}}-\frac{33117 X_1^4 X_2^5}{1024
  X_0^9}+\frac{429 X_1^2 X_2^6}{512 X_0^8}+\frac{58200075 X_1^{11}
  X_3}{524288 X_0^{12}}-\frac{10863477 X_1^9 X_2 X_3}{32768
  X_0^{11}}+\frac{627551 X_1^7 X_2^2 X_3}{2048 X_0^{10}}-\frac{48377
  X_1^5 X_2^3 X_3}{512 X_0^9}+\frac{1719 X_1^3 X_2^4 X_3}{256
  X_0^8}+\frac{433577 X_1^8 X_3^2}{16384 X_0^{10}}-\frac{533 X_1^6 X_2
  X_3^2}{16 X_0^9}+\frac{1917 X_1^4 X_2^2 X_3^2}{256 X_0^8}+\frac{45
  X_1^5 X_3^3}{64 X_0^8}-\frac{563479 X_1^{10} X_4}{32768
  X_0^{11}}+\frac{164759 X_1^8 X_2 X_4}{4096 X_0^{10}}-\frac{13033
  X_1^6 X_2^2 X_4}{512 X_0^9}+\frac{247 X_1^4 X_2^3 X_4}{64
  X_0^8}-\frac{1265 X_1^7 X_3 X_4}{256 X_0^9}+\frac{103 X_1^5 X_2 X_3
  X_4}{32 X_0^8}+\frac{69 X_1^6 X_4^2}{512 X_0^8}+\frac{14831 X_1^9
  X_5}{8192 X_0^{10}}-\frac{1525 X_1^7 X_2 X_5}{512 X_0^9}+\frac{125
  X_1^5 X_2^2 X_5}{128 X_0^8}+\frac{55 X_1^6 X_3 X_5}{256
  X_0^8}-\frac{123 X_1^8 X_6}{1024 X_0^9}+\frac{27 X_1^6 X_2 X_6}{256
  X_0^8}+\frac{X_1^7 X_7}{256 X_0^8} $

\item $ v_8=-\frac{1577687661981 X_1^{16}}{2147483648
  X_0^{16}}+\frac{139033974807 X_1^{14} X_2}{33554432
  X_0^{15}}-\frac{74888368881 X_1^{12} X_2^2}{8388608
  X_0^{14}}+\frac{9664024681 X_1^{10} X_2^3}{1048576
  X_0^{13}}-\frac{1227426451 X_1^8 X_2^4}{262144
  X_0^{12}}+\frac{17643971 X_1^6 X_2^5}{16384 X_0^{11}}-\frac{363401
  X_1^4 X_2^6}{4096 X_0^{10}}+\frac{715 X_1^2 X_2^7}{512
  X_0^9}-\frac{4018672971 X_1^{13} X_3}{4194304
  X_0^{14}}+\frac{3707614707 X_1^{11} X_2 X_3}{1048576
  X_0^{13}}-\frac{291667191 X_1^9 X_2^2 X_3}{65536
  X_0^{12}}+\frac{36805607 X_1^7 X_2^3 X_3}{16384
  X_0^{11}}-\frac{835327 X_1^5 X_2^4 X_3}{2048 X_0^{10}}+\frac{8665
  X_1^3 X_2^5 X_3}{512 X_0^9}-\frac{40882029 X_1^{10} X_3^2}{131072
  X_0^{12}}+\frac{9831307 X_1^8 X_2 X_3^2}{16384
  X_0^{11}}-\frac{1228903 X_1^6 X_2^2 X_3^2}{4096 X_0^{10}}+\frac{4301
  X_1^4 X_2^3 X_3^2}{128 X_0^9}-\frac{41639 X_1^7 X_3^3}{2048
  X_0^{10}}+\frac{5131 X_1^5 X_2 X_3^3}{512 X_0^9}+\frac{164027379
  X_1^{12} X_4}{1048576 X_0^{13}}-\frac{7839267 X_1^{10} X_2
  X_4}{16384 X_0^{12}}+\frac{7581045 X_1^8 X_2^2 X_4}{16384
  X_0^{11}}-\frac{159407 X_1^6 X_2^3 X_4}{1024 X_0^{10}}+\frac{6821
  X_1^4 X_2^4 X_4}{512 X_0^9}+\frac{583489 X_1^9 X_3 X_4}{8192
  X_0^{11}}-\frac{96323 X_1^7 X_2 X_3 X_4}{1024 X_0^{10}}+\frac{3005
  X_1^5 X_2^2 X_3 X_4}{128 X_0^9}+\frac{703 X_1^6 X_3^2 X_4}{256
  X_0^9}-\frac{12591 X_1^8 X_4^2}{4096 X_0^{10}}+\frac{539 X_1^6 X_2
  X_4^2}{256 X_0^9}-\frac{1195773 X_1^{11} X_5}{65536
  X_0^{12}}+\frac{89479 X_1^9 X_2 X_5}{2048 X_0^{11}}-\frac{59515
  X_1^7 X_2^2 X_5}{2048 X_0^{10}}+\frac{1255 X_1^5 X_2^3 X_5}{256
  X_0^9}-\frac{10109 X_1^8 X_3 X_5}{2048 X_0^{10}}+\frac{1735 X_1^6
  X_2 X_3 X_5}{512 X_0^9}+\frac{125 X_1^7 X_4 X_5}{512
  X_0^9}+\frac{24259 X_1^{10} X_6}{16384 X_0^{11}}-\frac{5093 X_1^8
  X_2 X_6}{2048 X_0^{10}}+\frac{55 X_1^6 X_2^2 X_6}{64 X_0^9}+\frac{83
  X_1^7 X_3 X_6}{512 X_0^9}-\frac{157 X_1^9 X_7}{2048
  X_0^{10}}+\frac{35 X_1^7 X_2 X_7}{512 X_0^9}+\frac{X_1^8 X_8}{512
  X_0^9} $

\item $ v_9=\frac{105827327217339 X_1^{18}}{17179869184
  X_0^{18}}-\frac{10475370519711 X_1^{16} X_2}{268435456
  X_0^{17}}+\frac{13099696441515 X_1^{14} X_2^2}{134217728
  X_0^{16}}-\frac{128691280227 X_1^{12} X_2^3}{1048576
  X_0^{15}}+\frac{85807052791 X_1^{10} X_2^4}{1048576
  X_0^{14}}-\frac{918217321 X_1^8 X_2^5}{32768
  X_0^{13}}+\frac{71881751 X_1^6 X_2^6}{16384 X_0^{12}}-\frac{123491
  X_1^4 X_2^7}{512 X_0^{11}}+\frac{2431 X_1^2 X_2^8}{1024
  X_0^{10}}+\frac{615304657155 X_1^{15} X_3}{67108864
  X_0^{16}}-\frac{337447157091 X_1^{13} X_2 X_3}{8388608
  X_0^{15}}+\frac{67268613445 X_1^{11} X_2^2 X_3}{1048576
  X_0^{14}}-\frac{1494135783 X_1^9 X_2^3 X_3}{32768
  X_0^{13}}+\frac{116498187 X_1^7 X_2^4 X_3}{8192
  X_0^{12}}-\frac{1671363 X_1^5 X_2^5 X_3}{1024 X_0^{11}}+\frac{83
  X_1^3 X_2^6 X_3}{2 X_0^{10}}+\frac{15885206489 X_1^{12}
  X_3^2}{4194304 X_0^{14}}-\frac{643914029 X_1^{10} X_2 X_3^2}{65536
  X_0^{13}}+\frac{128491287 X_1^8 X_2^2 X_3^2}{16384
  X_0^{12}}-\frac{1070609 X_1^6 X_2^3 X_3^2}{512 X_0^{11}}+\frac{67719
  X_1^4 X_2^4 X_3^2}{512 X_0^{10}}+\frac{6831329 X_1^9 X_3^3}{16384
  X_0^{12}}-\frac{113183 X_1^7 X_2 X_3^3}{256 X_0^{11}}+\frac{10709
  X_1^5 X_2^2 X_3^3}{128 X_0^{10}}+\frac{10685 X_1^6 X_3^4}{2048
  X_0^{10}}-\frac{13029980589 X_1^{14} X_4}{8388608
  X_0^{15}}+\frac{6135385481 X_1^{12} X_2 X_4}{1048576
  X_0^{14}}-\frac{249975727 X_1^{10} X_2^2 X_4}{32768
  X_0^{13}}+\frac{33517905 X_1^8 X_2^3 X_4}{8192
  X_0^{12}}-\frac{849053 X_1^6 X_2^4 X_4}{1024 X_0^{11}}+\frac{11025
  X_1^4 X_2^5 X_4}{256 X_0^{10}}-\frac{498703 X_1^{11} X_3 X_4}{512
  X_0^{13}}+\frac{15923461 X_1^9 X_2 X_3 X_4}{8192
  X_0^{12}}-\frac{533005 X_1^7 X_2^2 X_3 X_4}{512
  X_0^{11}}+\frac{34263 X_1^5 X_2^3 X_3 X_4}{256
  X_0^{10}}-\frac{378739 X_1^8 X_3^2 X_4}{4096 X_0^{11}}+\frac{25269
  X_1^6 X_2 X_3^2 X_4}{512 X_0^{10}}+\frac{3420375 X_1^{10}
  X_4^2}{65536 X_0^{12}}-\frac{293173 X_1^8 X_2 X_4^2}{4096
  X_0^{11}}+\frac{39571 X_1^6 X_2^2 X_4^2}{2048 X_0^{10}}+\frac{3957
  X_1^7 X_3 X_4^2}{1024 X_0^{10}}+\frac{406556795 X_1^{13}
  X_5}{2097152 X_0^{14}}-\frac{39679391 X_1^{11} X_2 X_5}{65536
  X_0^{13}}+\frac{1246277 X_1^9 X_2^2 X_5}{2048 X_0^{12}}-\frac{112561
  X_1^7 X_2^3 X_5}{512 X_0^{11}}+\frac{693 X_1^5 X_2^4 X_5}{32
  X_0^{10}}+\frac{2761245 X_1^{10} X_3 X_5}{32768
  X_0^{12}}-\frac{474493 X_1^8 X_2 X_3 X_5}{4096 X_0^{11}}+\frac{32151
  X_1^6 X_2^2 X_3 X_5}{1024 X_0^{10}}+\frac{801 X_1^7 X_3^2 X_5}{256
  X_0^{10}}-\frac{27601 X_1^9 X_4 X_5}{4096 X_0^{11}}+\frac{2461 X_1^7
  X_2 X_4 X_5}{512 X_0^{10}}+\frac{251 X_1^8 X_5^2}{2048
  X_0^{10}}-\frac{2323479 X_1^{12} X_6}{131072 X_0^{13}}+\frac{1418337
  X_1^{10} X_2 X_6}{32768 X_0^{12}}-\frac{122777 X_1^8 X_2^2 X_6}{4096
  X_0^{11}}+\frac{5619 X_1^6 X_2^3 X_6}{1024 X_0^{10}}-\frac{18515
  X_1^9 X_3 X_6}{4096 X_0^{11}}+\frac{207 X_1^7 X_2 X_3 X_6}{64
  X_0^{10}}+\frac{209 X_1^8 X_4 X_6}{1024 X_0^{10}}+\frac{37527
  X_1^{11} X_7}{32768 X_0^{12}}-\frac{8011 X_1^9 X_2 X_7}{4096
  X_0^{11}}+\frac{721 X_1^7 X_2^2 X_7}{1024 X_0^{10}}+\frac{119 X_1^8
  X_3 X_7}{1024 X_0^{10}}-\frac{195 X_1^{10} X_8}{4096
  X_0^{11}}+\frac{11 X_1^8 X_2 X_8}{256 X_0^{10}}+\frac{X_1^9
  X_9}{1024 X_0^{10}} $

\item $ v_{10}=-\frac{15869874201854147 X_1^{20}}{274877906944
  X_0^{20}}+\frac{6968235530124323 X_1^{18} X_2}{17179869184
  X_0^{19}}-\frac{1239226092495387 X_1^{16} X_2^2}{1073741824
  X_0^{18}}+\frac{229652647583487 X_1^{14} X_2^3}{134217728
  X_0^{17}}-\frac{95081095828539 X_1^{12} X_2^4}{67108864
  X_0^{16}}+\frac{684656180415 X_1^{10} X_2^5}{1048576
  X_0^{15}}-\frac{20650143993 X_1^8 X_2^6}{131072
  X_0^{14}}+\frac{284065323 X_1^6 X_2^7}{16384
  X_0^{13}}-\frac{10679265 X_1^4 X_2^8}{16384 X_0^{12}}+\frac{4199
  X_1^2 X_2^9}{1024 X_0^{11}}-\frac{51777526098357 X_1^{17}
  X_3}{536870912 X_0^{18}}+\frac{65745451959171 X_1^{15} X_2
  X_3}{134217728 X_0^{17}}-\frac{15872302764663 X_1^{13} X_2^2
  X_3}{16777216 X_0^{16}}+\frac{3658886479929 X_1^{11} X_2^3
  X_3}{4194304 X_0^{15}}-\frac{51831619715 X_1^9 X_2^4 X_3}{131072
  X_0^{14}}+\frac{167190241 X_1^7 X_2^5 X_3}{2048
  X_0^{13}}-\frac{12693319 X_1^5 X_2^6 X_3}{2048 X_0^{12}}+\frac{51123
  X_1^3 X_2^7 X_3}{512 X_0^{11}}-\frac{1620053231895 X_1^{14}
  X_3^2}{33554432 X_0^{16}}+\frac{661505567877 X_1^{12} X_2
  X_3^2}{4194304 X_0^{15}}-\frac{91656289433 X_1^{10} X_2^2
  X_3^2}{524288 X_0^{14}}+\frac{2541271583 X_1^8 X_2^3 X_3^2}{32768
  X_0^{13}}-\frac{12767375 X_1^6 X_2^4 X_3^2}{1024
  X_0^{12}}+\frac{244533 X_1^4 X_2^5 X_3^2}{512
  X_0^{11}}-\frac{2007604153 X_1^{11} X_3^3}{262144
  X_0^{14}}+\frac{835306623 X_1^9 X_2 X_3^3}{65536
  X_0^{13}}-\frac{351417 X_1^7 X_2^2 X_3^3}{64 X_0^{12}}+\frac{68455
  X_1^5 X_2^3 X_3^3}{128 X_0^{11}}-\frac{8364263 X_1^8 X_3^4}{32768
  X_0^{12}}+\frac{108009 X_1^6 X_2 X_3^4}{1024
  X_0^{11}}+\frac{2253822123327 X_1^{16} X_4}{134217728
  X_0^{17}}-\frac{19664047737 X_1^{14} X_2 X_4}{262144
  X_0^{16}}+\frac{516379192119 X_1^{12} X_2^2 X_4}{4194304
  X_0^{15}}-\frac{12007772903 X_1^{10} X_2^3 X_4}{131072
  X_0^{14}}+\frac{252344843 X_1^8 X_2^4 X_4}{8192
  X_0^{13}}-\frac{4128751 X_1^6 X_2^5 X_4}{1024 X_0^{12}}+\frac{68169
  X_1^4 X_2^6 X_4}{512 X_0^{11}}+\frac{28183020879 X_1^{13} X_3
  X_4}{2097152 X_0^{15}}-\frac{4706661189 X_1^{11} X_2 X_3 X_4}{131072
  X_0^{14}}+\frac{123445225 X_1^9 X_2^2 X_3 X_4}{4096
  X_0^{13}}-\frac{8987261 X_1^7 X_2^3 X_3 X_4}{1024
  X_0^{12}}+\frac{168391 X_1^5 X_2^4 X_3 X_4}{256
  X_0^{11}}+\frac{17706977 X_1^{10} X_3^2 X_4}{8192
  X_0^{13}}-\frac{9925559 X_1^8 X_2 X_3^2 X_4}{4096
  X_0^{12}}+\frac{130669 X_1^6 X_2^2 X_3^2 X_4}{256
  X_0^{11}}+\frac{18435 X_1^7 X_3^3 X_4}{512 X_0^{11}}-\frac{429493939
  X_1^{12} X_4^2}{524288 X_0^{14}}+\frac{27603415 X_1^{10} X_2
  X_4^2}{16384 X_0^{13}}-\frac{976309 X_1^8 X_2^2 X_4^2}{1024
  X_0^{12}}+\frac{139405 X_1^6 X_2^3 X_4^2}{1024
  X_0^{11}}-\frac{615585 X_1^9 X_3 X_4^2}{4096 X_0^{12}}+\frac{87369
  X_1^7 X_2 X_3 X_4^2}{1024 X_0^{11}}+\frac{3991 X_1^8 X_4^3}{2048
  X_0^{11}}-\frac{36897440619 X_1^{15} X_5}{16777216
  X_0^{16}}+\frac{8851281549 X_1^{13} X_2 X_5}{1048576
  X_0^{15}}-\frac{2977346665 X_1^{11} X_2^2 X_5}{262144
  X_0^{14}}+\frac{210413865 X_1^9 X_2^3 X_5}{32768
  X_0^{13}}-\frac{729995 X_1^7 X_2^4 X_5}{512 X_0^{12}}+\frac{22535
  X_1^5 X_2^5 X_5}{256 X_0^{11}}-\frac{348231785 X_1^{12} X_3
  X_5}{262144 X_0^{14}}+\frac{179473365 X_1^{10} X_2 X_3 X_5}{65536
  X_0^{13}}-\frac{1592603 X_1^8 X_2^2 X_3 X_5}{1024
  X_0^{12}}+\frac{14295 X_1^6 X_2^3 X_3 X_5}{64
  X_0^{11}}-\frac{1001561 X_1^9 X_3^2 X_5}{8192 X_0^{12}}+\frac{35705
  X_1^7 X_2 X_3^2 X_5}{512 X_0^{11}}+\frac{8763587 X_1^{11} X_4
  X_5}{65536 X_0^{13}}-\frac{379 X_1^9 X_2 X_4 X_5}{2
  X_0^{12}}+\frac{27945 X_1^7 X_2^2 X_4 X_5}{512 X_0^{11}}+\frac{19455
  X_1^8 X_3 X_4 X_5}{2048 X_0^{11}}-\frac{65921 X_1^{10} X_5^2}{16384
  X_0^{12}}+\frac{3035 X_1^8 X_2 X_5^2}{1024 X_0^{11}}+\frac{913545795
  X_1^{14} X_6}{4194304 X_0^{15}}-\frac{181589513 X_1^{12} X_2
  X_6}{262144 X_0^{14}}+\frac{23576079 X_1^{10} X_2^2 X_6}{32768
  X_0^{13}}-\frac{282231 X_1^8 X_2^3 X_6}{1024 X_0^{12}}+\frac{7751
  X_1^6 X_2^4 X_6}{256 X_0^{11}}+\frac{5926069 X_1^{11} X_3 X_6}{65536
  X_0^{13}}-\frac{526593 X_1^9 X_2 X_3 X_6}{4096 X_0^{12}}+\frac{19055
  X_1^7 X_2^2 X_3 X_6}{512 X_0^{11}}+\frac{6607 X_1^8 X_3^2 X_6}{2048
  X_0^{11}}-\frac{55075 X_1^{10} X_4 X_6}{8192 X_0^{12}}+\frac{10153
  X_1^8 X_2 X_4 X_6}{2048 X_0^{11}}+\frac{461 X_1^9 X_5 X_6}{2048
  X_0^{11}}-\frac{4214341 X_1^{13} X_7}{262144 X_0^{14}}+\frac{2615381
  X_1^{11} X_2 X_7}{65536 X_0^{13}}-\frac{58513 X_1^9 X_2^2 X_7}{2048
  X_0^{12}}+\frac{357 X_1^7 X_2^3 X_7}{64 X_0^{11}}-\frac{31717
  X_1^{10} X_3 X_7}{8192 X_0^{12}}+\frac{2933 X_1^8 X_2 X_3 X_7}{1024
  X_0^{11}}+\frac{329 X_1^9 X_4 X_7}{2048 X_0^{11}}+\frac{55547
  X_1^{12} X_8}{65536 X_0^{13}}-\frac{6011 X_1^{10} X_2 X_8}{4096
  X_0^{12}}+\frac{559 X_1^8 X_2^2 X_8}{1024 X_0^{11}}+\frac{41 X_1^9
  X_3 X_8}{512 X_0^{11}}-\frac{237 X_1^{11} X_9}{8192
  X_0^{12}}+\frac{27 X_1^9 X_2 X_9}{1024 X_0^{11}}+\frac{X_1^{10}
  X_{10}}{2048 X_0^{11}} $
\end{itemize}}

\end{document}